\documentclass[reqno, 11pt]{amsart}
\usepackage[utf8]{inputenc}
\usepackage{geometry}
\usepackage{amsfonts}
\usepackage{hyperref}
\usepackage[backend=biber, style=numeric-comp]{biblatex}
\hypersetup{
pdftitle={TBA},
pdfsubject={},
pdfauthor={Shijia Jin},
pdfkeywords={}
}
\usepackage{amsmath}
\usepackage{xcolor}
\usepackage{ulem}
\usepackage{amsthm}
\usepackage{pdflscape}
\usepackage{pgfplots}
\usepackage{mathrsfs}
\calclayout

\DeclareMathOperator*{\argmin}{arg\,min}

\newtheorem{theorem}{Theorem}[section]

\newtheorem{proposition}[theorem]{Proposition}
\newtheorem{lemma}[theorem]{Lemma}

\newtheorem{assumption}[theorem]{Assumption}
\newtheorem{remark}[theorem]{Remark}

\theoremstyle{definition}
\newtheorem{definition}[theorem]{Definition}
\newtheorem{example}[theorem]{Example}

\usepackage{relsize}

\usepackage{mlmodern}
\usepackage[T1]{fontenc}

\newcommand{\kong}{\vspace{0.08cm}}

\title{Macroscopic Market Making Games via Multidimensional Decoupling Field}

\author{Ivan Guo}
\thanks{\hspace{0.2cm} Ivan Guo, Email: ivan.guo@monash.edu, Address: Centre for Quantitative Finance and Investment Strategies, School of Mathematics, Monash University, Wellington Rd, Clayton VIC 3800, Australia}
\author{Shijia Jin} 
\thanks{\hspace{0.2cm} Corresponding author: Shijia Jin, Email: shijia.jin@monash.edu, Address: School of Mathematics, Monash University, Wellington Rd, Clayton VIC 3800, Australia}

\begin{document}
\relscale{1.047}

\begin{abstract}
Building on the macroscopic market making framework as a control problem, this paper investigates its extension to stochastic games. In the context of price competition, each agent is benchmarked against the best quote offered by the others. We begin with the linear case. While constructing the solution directly, the \textit{ordering property} and the dimension reduction in the equilibrium are revealed. For the non-linear case, we extend the decoupling approach by introducing a multidimensional \textit{characteristic equation} to analyse the well-posedness of the forward-backward stochastic differential equations. Properties of the coefficients in this characteristic equation are derived using tools from non-smooth analysis. Several new well-posedness results are presented.\\


\end{abstract}

\maketitle
\tableofcontents

\section{Introduction}
\label{intro}\noindent This paper is mainly devoted to the strategic interactions between market makers, who act as liquidity providers in the financial market. Market makers offer bid and ask prices for one or multiple assets, generating profits from the bid-ask spread (the price difference between buying and selling orders). Market making as a stochastic control problem has been extensively explored in market microstructure literature, introduced by \cite{ho1981optimal} and further developed by \cite{avellaneda2008high}. Recently, a macroscopic model \`a la Avellaneda-Stoikov \cite{avellaneda2008high} has been proposed by \cite{guo2023macroscopic}, while still following the single-player optimization setting. Therefore, we extend the investigation to stochastic games within this macroscopic framework, with two main motivations in mind. Our model mainly focuses on quote-driven markets and order-driven markets where the ratio of bid-ask spread to tick size is large. We refer readers to the introduction of \cite{guo2023macroscopic} for further details on the macroscopic model and a comparison with the Avellaneda-Stoikov model.

The study of market making games \`a la Avellaneda-Stoikov is relatively rare, while two recent notable exceptions are \cite{luo2021dynamic} and \cite{cont2022dynamics}. As mentioned by \cite{luo2021dynamic}, literatures on the market making as a stochastic control problem and trading via limit orders (e.g., \cite{gueant2012optimal}, \cite{gueant2013dealing}, \cite{gueant2017optimal}, \cite{bayraktar2014liquidation}, \cite{campi2020optimal}) share a common assumption. In detail, the probability that a limit order is executed at the quoted price is only a function of the price gap between the order and a reference price. While this approach offers tractability, it neglects the influences of prices offered by other market makers. In our setting, this `probability' depends on the difference between the agent's price and the best price offered by the others. The consideration of price competition is a key distinction in our paper compared to these stochastic control models. In comparison to \cite{luo2021dynamic} and \cite{cont2022dynamics}, we provide a distinct modelling of competition, which delivers several explicit mathematical results.

The macroscopic market making model seeks to bridge several gaps between market making and optimal execution problems. For comprehensive treatments of these two topics, we refer the reader to \cite{gueant2016financial} and \cite{cartea2015algorithmic}. One key contribution of the macroscopic model is its alignment with optimal execution problems in representing market orders through trading rates. On the other hand, price impact functions---introduced by \cite{almgren2001optimal} and \cite{bertsimas1998optimal}---were designed to capture the effect of large orders on asset prices. These functions are typically pre-specified in an execution problem to accommodate stochastic control techniques. For instance, non-linear impact functions were studied in \cite{almgren2003optimal}, while \cite{souza2022regularized} considered stochastic coefficients. However, such impact functions serve primarily as approximations of actual price dynamics. In contrast, our approach models price impact as a consequence of the quoting strategies of market makers. Considering the number market participants, explaining price changes can be better achieved through the interactions of multiple market makers.

\textit{Contribution}: This paper proposes novel game models for market makers. Mathematically, the contributions lie in the global well-posedness of Nash equilibria and the development of associated tools to achieve these results. From an economic perspective, along with novel models for price competitions, we aim to establish connections between the market making problem, price impact, and the optimal execution problem. The outline of each section is given as follows:

\begin{itemize}
    \item Section \ref{paper 2 section 2} studies the game in a linear setting, of which the well-posedness result is presented in Theorem \ref{N_players_and_4_players}. After introducing the general non-linear setting, Section \ref{paper 2 section 3} characterizes the Nash equilibrium as the solution of a forward-backward system \eqref{general FBSDE}. The local well-posedness is well-known, due to the Lipschitz property obtained by a novel version of the implicit function theorem.\\
    \vspace{-0.2cm}
    
    \item To achieve the global well-posedness of \eqref{general FBSDE}, the generalized derivatives of the coefficient functions are analysed in Section \ref{paper 2 section 4}. The finding is summarized in Theorem \ref{global_Jacobian}, using a novel non-smooth analysis method. Based on the generalized derivative, Section \ref{paper 2 section 5} explores a key property in the equilibrium, described in Theorem \ref{global_Jacobian_M0}. Additionally, by an extension of the decoupling approach, Theorem \ref{mult_char_bsre} bridges the global well-posedness between \eqref{general FBSDE} and a Riccati equation.\\
    \vspace{-0.2cm}

    \item With above tools, several global well-posedness results (e.g., Theorem \ref{2_dim bsre}) are presented in Section \ref{paper 2 section 6}, along with a price impact decomposition \eqref{explicit price} in a simple example. Concerning heterogenous coefficients, Section \ref{paper 2 section 7} studies a two-player game; see Theorem \ref{heter game}.
\end{itemize}
\kong

We then provide a more detailed introduction. The paper starts with the linear case in Section \ref{paper 2 section 2} for analytic tractability. In the game setting, we let the portion be a linear decreasing function with respect to $\delta-\bar{\delta}$, where $\delta$ is the price gap of the agent and $\bar{\delta}$ is the smallest gap provided by the other market makers. In other words, each player is compared with the best quote from the others. The concavity of the objective functional helps characterize the Nash equilibrium as the solution of a forward-backward stochastic differential equation (FBSDE). Thanks to the linear structure, we can explicitly derive the solution of the equation. The \textit{ordering property} in the stochastic control setting (Theorem 3.8 in \cite{guo2023macroscopic})---agent with higher inventory tends to place ask orders at a lower price---persists in the equilibrium. This, in turn, effectively reduces the $N$-player game into a $4$-player game. 

Next, in Section \ref{paper 2 section 3} we delve into the case where the portion becomes a general decreasing function, as introduced in \cite{gueant2017optimal}, with respect to $\delta-\bar{\delta}$. After verifying the Isaac's condition, we introduce a non-smooth implicit function theorem to ensure the Lipschitz continuity of the implicit function. A version of the stochastic maximum principle then characterizes the equilibrium as an FBSDE, the local well-posedness of which is well-known due to the Lipschitz condition. 

To achieve the global well-posedness result, in Section \ref{paper 2 section 4},  we look at the `derivative' of the coefficient functions. Since the functions are Lipschitz, we utilize the \textit{(Clarke) generalized derivatives} and the corresponding \textit{non-smooth analysis}. While the derivative of the implicit function in the smooth implicit function theorem is well-known, the generalized derivative of the implicit function in the non-smooth scenario is not as well established. A novel analysis reveals that the generalized derivative of the coefficient functions consists of \textit{$M$-matrices}.

Based on the derivative information, in Section \ref{paper 2 section 5}, we show the ordering property remains in the non-linear setting. With respect to the FBSDE system, we adopt the method of decoupling fields introduced by \cite{ma2015well} for the one-dimensional equation. A multi-dimensional extension of the \textit{characteristic BSDE} is then presented. Consequently, the well-posedness of the FBSDE can then be guaranteed by the existence of a unique bounded solution to some backward stochastic Riccati equation (BSRE).

The majority of literature on Riccati equations deals with either symmetric or positive definite coefficient matrices, let alone literature on BSREs. Section \ref{paper 2 section 6} establishes several new well-posedness results, particularly providing a complete study of the $2$-dimensional equation. In companion with this $2$-dimensional result, additional conditions are presented so that the $4$-player game can be further reduced to a $2$-player game. As an example, we study the case when all players are identical. In this scenario, the permanent price impact can be further decomposed into two components: an \textit{ex post} impact indicates the effect of past order imbalances, and an \textit{ex ante} impact specifies the influences of expected imbalances in the future.

While most prior discussions focus on agents with homogeneous risk parameters, Section \ref{paper 2 section 7} examines a stochastic game consisting of two heterogeneous agents. Despite the heterogeneity, the concavity of the objective functional remains unaffected, allowing the equilibrium to be characterized as the solution of an FBSDE. The well-posedness of the equation is derived based on \cite{peng1992stochastic}. Additionally, we provide an example demonstrating that the ordering property can fail in the heterogeneous case.

\textit{Notation}: Throughout the present work, we fix $T > 0$ to represent our finite trading horizon. We denote by $(\Omega, \mathcal{F}, \mathbb{F}=(\mathcal{F}_t)_{0\leq t\leq T}, \mathbb{P})$ a complete filtered probability space, with $\mathcal{F}_T = \mathcal{F}$. An $m$-dimensional Brownian motion $W=(W^1,\dots,W^m)$ is defined on such space, for a fixed positive integer $m$, and the filtration $\mathbb{F}$ is generated by $W$ and augmented. Let $\mathcal{G}$ represents an arbitrary $\sigma$-algebra contained in $\mathcal{F}$ and consider the following spaces:
\begin{gather*}
    L^p(\Omega, \mathcal{G}):=\big\{X: X \text{ is } \mathcal{G} \text{-measurable and } \mathbb{E}|X|^p<\infty \big\};\\
    \mathbb{H}^p:=\Big\{X: X \text{ is } \mathbb{F} \text{-progressively measurable and } \mathbb{E}\Big[\big(\int_0^T |X_t|^2\,dt\big)^{p/2}\Big]<\infty\, \Big\};\\
    \mathbb{S}^p:=\Big\{X\in \mathbb{H}^p: \mathbb{E}\big[\sup_{0\leq t\leq T} |X_t|^p\big]<\infty \Big\};\\
    \mathbb{M}:= \big\{ M : M_t \in L^2(\Omega, \mathcal{F}_t) \text{ for a.e. } t \in [0, T] \text{ and } \{M_t, \mathcal{F}_t\}_{0\leq t\leq T}\text{ is a continuous martingale} \big\}.
\end{gather*}
We use superscripts for enumerating purposes. For example, superscripts in $Q^1, Q^2$ are to mark objects which are associated with player $1$ and player $2$ respectively. In particular, $Q^2$ is not to be confused with quadratic powers, which will be explicitly denoted with brackets like $(Q)^2$.\\

\section{Linear Market Making Game}
\label{paper 2 section 2}
\noindent In \cite{guo2023macroscopic}, macroscopic marketing making was studied as a control problem for a market maker. Here we explore its extension as a stochastic game between several market makers. We start with the case of the linear intensity function and subsequently extend our analysis to general intensity functions. 
The linear structure of the model offers extra analytic tractability, allowing us to construct the solution explicitly. Moreover, we can see more directly how the $N$-player game can be reduced to the \textit{four-player framework}.

We consider a scenario involving multiple liquidity takers engaged in trading a single asset within the market. The collective trading rates of market orders on the ask and bid side are represented by the processes $a := (a_t)_{t\in[0, T]} \in\mathbb{H}^2$ and $b := (b_t)_{t\in[0, T]}\in\mathbb{H}^2$, respectively, where $a_t\in(0,\bar{a}]$ and $b_t\in(0,\bar{b}]$ for some constants $\bar{a}, \bar{b}>0$. These trading rates reflect the liquidity demands of liquidity takers. To meet this demand, a group of $N\in \mathbb{N}$ \textit{homogeneous} market makers provides the liquidity. Each market maker, identified by $i\in\{1,...,N\}$, dynamically places buy limit orders at the price level $S_t-\delta_t^{i,b}$ and sell limit orders at $S_t+\delta_t^{i,a}$. Here, vector $\boldsymbol{\delta}^i := (\delta^{i,a}, \delta^{i,b}) \in \mathbb{H}^2\times \mathbb{H}^2$ represents the control strategy employed by agent $i$, and $\{S_t,\mathcal{F}_t\}_{t\in[0,T]}$---the fundamental price---is a square-integrable martingale. Recalling that the linear intensity function is defined by $\Lambda(\delta)=\zeta-\gamma\,\delta$, we introduce its game extension as follows:

\kong

\begin{assumption}[Linear intensity]
\label{linear assumption}
The quantity of order flow executed by agent $i$ depends linearly on the difference between her offered price and the best price offered by the others. On the ask side, the executed flow is determined by the difference  $\delta_t^{i,a}-\bar{\delta}_t^{i,a}$, and on the bid side, it is determined by $\delta_t^{i,b}-\bar{\delta}_t^{i,b}$, where
\begin{equation*}
    \bar{\delta}_t^{i,a}:=\min_{j\neq i}\delta_t^{j,a} \quad\text{and}\quad \bar{\delta}_t^{i,b}:=\min_{j\neq i}\delta_t^{j,b}.
\end{equation*}
Specifically, if we write $v_t^{i,a}$ and $v_t^{i,b}$ as the (passive) selling and buying rates of the agent at time $t$, the linear dependence can be represented by 
\begin{equation}
    v_t^{i,a}=a_t\cdot\big(\,\zeta-\gamma\,(\delta_t^{i,a}-\bar{\delta}_t^{i,a})\,\big) \quad\text{and}\quad v_t^{i,b}=b_t\cdot\big(\,\zeta-\gamma\,(\delta_t^{i,b}-\bar{\delta}_t^{i,b})\,\big),
    \label{lin_game}
\end{equation}
for some constants $\zeta$, $\gamma>0$. We have also assumed the bid-ask symmetry for notational convenience.
\end{assumption}

\kong

\begin{remark}
The proposed model serves as an approximation for the price competition. While the model does not guarantee the market clearing condition, it retains a crucial element from the Avellaneda-Stoikov model: the gap between the offered price and the `best' price. Both the stochastic control problems in \cite{guo2023macroscopic} and \cite{avellaneda2008high} can then be viewed as special cases where the `best' price is assumed to be the fundamental price. On the other hand, it is possible in \eqref{lin_game} that $v_t^{i,a}>a_t$. To address this, one can further propose a bounded action space for all agents and introduce an additional drift term in \eqref{lin_game}:
\begin{equation*}
    v_t^{i,a}=a_t\,\Big(\,\zeta-\xi_3-\gamma\,\big(\delta_t^{i,a}\vee(-\xi_1)\wedge\xi_2-\bar{\delta}_t^{i,a}\vee(-\xi_1)\wedge\xi_2\big)\,\Big),
\end{equation*}
where $\xi_1, \xi_2, \xi_3 > 0$ are additional coefficients. We choose to not follow this approach, due to considerations of notation convenience and its little impact on the mathematical difficulty of the latter general case.
\end{remark}

\kong

\noindent The inventory and cash of the agent $i$ are modelled by $X_t^i$ and $Q_t^i$ accordingly:
\begin{gather*}
    X_t^i=\int_0^t(S_u+\delta_u^{i,a})\,a_u\,(\,\zeta-\gamma\,(\delta_u^{i,a}-\bar{\delta}_u^{i,a})\,)\,du-\int_0^t(S_u-\delta_u^{i,b})\,b_u\,(\,\zeta-\gamma\,(\delta_u^{i,b}-\bar{\delta}_u^{i,b})\,)\,du,\\
    Q_t^i=q_0^i-\int_0^ta_u\,(\,\zeta-\gamma\,(\delta_u^{i,a}-\bar{\delta}_u^{i,a})\,)\,du+\int_0^tb_u\,(\,\zeta-\gamma\,(\delta_u^{i,b}-\bar{\delta}_u^{i,b})\,)\,du,
    \nonumber   
\end{gather*}
where $q_0^i\in\mathbb{R}$ denotes the initial inventory level. The player $i$ aims at maximizing the objective functional
\begin{equation}
\begin{aligned}
    J(\boldsymbol{\delta}^i; \boldsymbol{\delta}^{-i}):&=\mathbb{E}\Big[X_T^i+S_T\,Q_T^i-\int_0^T\phi_t\big(Q_t^i\big)^2\,dt-A\,\big(Q_T^i\big)^2\Big]\\
     &= \mathbb{E}\Big[\int_0^T\delta_t^{i,a}\, a_t\, (\,\zeta-\gamma\,(\delta_t^{i,a}-\bar{\delta}_t^{i,a})\,)\,dt+\int_0^T\delta_t^{i,b}\,b_t \, (\,\zeta-\gamma\,(\delta_t^{i,b}-\bar{\delta}_t^{i,b})\,)\,dt\\
     &\hspace{3cm} -\int_0^T \phi_t\big(Q_t^i\big)^2\,dt-A\big(Q_T^i\big)^2 \Big].
    \label{lin_game_obj}
\end{aligned}
\end{equation}
Here, penalty coefficients $\phi:=(\phi_t)_{t\in[0,T]}\in\mathbb{H}^2$, $A\in L^2(\Omega, \mathcal{F}_T)$ are non-negative and (uniformly) bounded by constants $\bar{\phi}, \bar{A}>0$ respectively. The simplification is deduced by the martingale property of $S$ and It\^o's formula; see the \cite{guo2023macroscopic} for the interpretation and the simplification step of \eqref{lin_game_obj}. The goal is to find a Nash equilibrium in which all agents solve
their maximization problems simultaneously in the following sense:

\kong

\begin{definition}
A strategy profile $(\hat{\boldsymbol{\delta}}^j)_{j=1}^N \in (\mathbb{H}^2\times\mathbb{H}^2)^{N}$ is called
a Nash equilibrium if, for all $1\leq i\leq N$ and any admissible strategies $\boldsymbol{\delta}^i\in\mathbb{H}^2\times\mathbb{H}^2$, it holds that
\begin{equation}
    J(\boldsymbol{\delta}^i; \hat{\boldsymbol{\delta}}^{-i})\leq J(\hat{\boldsymbol{\delta}}^i; \hat{\boldsymbol{\delta}}^{-i}).
    \label{Nash}
\end{equation}
\end{definition}

\kong

\noindent Thanks to the linear-quadratic structure, we utilize the convex-analytic method, introduced in \cite{bank2017hedging}, to characterize the Nash equilibrium as a system of FBSDEs.

\kong

\begin{theorem}
\label{conv_ana}
A strategy profile $(\boldsymbol{\delta}^j)_{j=1}^N \in (\mathbb{H}^2\times\mathbb{H}^2)^{N}$ forms a Nash equilibrium if and only if it solves the following system of FBSDEs:
\begin{equation}
\left\{
\begin{aligned}
\;& dQ_t^i  = -a_t\,\big(\zeta+\gamma\bar{\delta}^{i,a}_t-\gamma\delta_t^{i,a}\big)dt+b_t\,\big(\zeta+\gamma\bar{\delta}^{i,b}_t-\gamma\delta_t^{i,b}\big)dt, \\
& d\delta_t^{i,a}=d\bar{\delta}^{i,a}_t/2+\phi_t\,Q_t^i\,dt-dM_t^i,\\
& d\delta_t^{i,b}=d\bar{\delta}^{i,b}_t/2-\phi_t\,Q_t^i\,dt+dM_t^i,\\
& Q_0^i=q_0^i,\quad \delta_T^{i,a}=\zeta/(2\gamma)+\bar{\delta}^{i,a}_T/2-A\,Q_T^i,\quad \delta_T^{i,b}=\zeta/(2\gamma)+\bar{\delta}^{i,b}_T/2+A\,Q_T^i,
\end{aligned}
\right.
\label{N_linear_fbsdes}
\end{equation}
\noindent where $M^i_t\in\mathbb{M}$, for all $i\in\{1,\dots, N\}$.
\end{theorem}

\begin{proof}
Consider an agent $i$ with $i\in\{1,\dots, N\}$. If the profile $(\boldsymbol{\delta}^j)_{1\leq j\leq N}$ represents a Nash equilibrium, agent $i$ can be seen as solving a stochastic control problem with a fixed $\boldsymbol{\delta}^{-i}$. In comparison to the problem studied in the \cite{guo2023macroscopic}, the only necessary modifications are replacing the `old' constant $\zeta$ with the processes $\zeta+\gamma\bar{\delta}^{i,b}_t$ and $\zeta+\gamma\bar{\delta}^{i,a}_t$. However, these modifications do not alter the concave nature of the functional \eqref{lin_game_obj} with respect to the control $\boldsymbol{\delta}^i$. The proof in Theorem 2.6 in \cite{guo2023macroscopic} demonstrates that the concavity relies solely on the linear impact of $\boldsymbol{\delta}^i$ on $Q^i$, and the non-negative property of $\gamma$. Thanks to the concavity, the necessary and sufficient condition for the optimality of $\boldsymbol{\delta}^i$ can be characterized by the first-order condition: for any $\boldsymbol{w}\in\mathbb{H}^2\times\mathbb{H}^2$ that is also uniformly bounded, the G\^ateaux derivative of the functional \eqref{lin_game_obj} with respect to the direction $\boldsymbol{w}$ should vanish, i.e.,
\begin{equation}
\begin{aligned}
    0=\big\langle\nabla\, J(\boldsymbol{\delta}^i; \boldsymbol{\delta}^{-i}), \boldsymbol{w}\big\rangle=\mathbb{E}\Big[&\int_0^T a_t\,w_t^a\,\big(\zeta+\gamma\,\bar{\delta}^{i,a}_t-2\gamma\,\delta_t^{i,a}-2\gamma A\,Q_T^i-2\gamma\int_t^T \phi_s\,Q_s^i\,ds\big)\,dt\\
    &\; +\int_0^Tb_t\,w_t^b\,\big(\zeta+\gamma\,\bar{\delta}^{i,b}_t-2\gamma\,\delta_t^{i,b}+2\gamma A\,Q_T^i+2\gamma\int_t^T\phi_s\,Q_s^i\,ds\big)\,dt\Big].
    \nonumber
\end{aligned}    
\end{equation}
Due to the arbitrariness of $\boldsymbol{w}$, the law of total expectation yields
\begin{equation*}
\begin{aligned}
    \delta_t^{i,a}&=\frac{\zeta}{2\gamma}+\frac{1}{2}\bar{\delta}^{i,a}_t-\mathbb{E}_t[A\,Q_T^i]-\mathbb{E}_t\int_t^T\phi_s\,Q_s^i\,ds,\\
    \delta_t^{i,b}&=\frac{\zeta}{2\gamma}+\frac{1}{2}\bar{\delta}^{i,b}_t+\mathbb{E}_t[A\,Q_T^i]+\mathbb{E}_t\int_t^T\phi_s\,Q_s^i\,ds.
\end{aligned}
\end{equation*}
that holds $d\mathbb{P}\times dt$ almost everywhere. Since the above relation holds for all agents, we obtain the system \eqref{N_linear_fbsdes}.
\end{proof}

\kong

The remainder of this section focuses on the well-posedness of the system $\eqref{N_linear_fbsdes}$. We first introduce a crucial lemma named as the \textit{ordering property}, which provides notable simplifications to the system. Especially, based on Theorem 3.8 in \cite{guo2023macroscopic}, it has been established that the optimal strategy in the control problem is monotonic with respect to the initial inventory. In the equilibrium of the game setting, we will demonstrate that this property still holds true.

\kong

\begin{lemma}[Ordering property]
\label{no_crossing}
Every equilibrium possesses the ordering property. Specifically, for each pair $i, j\in\{1,\dots, N\}$ with $q_0^i \geq q_0^j$, the equilibrium strategies $\boldsymbol{\delta}^i$ and $\boldsymbol{\delta}^j$ satisfy that
\begin{equation}
    \delta^{i,a}_t \leq \delta^{j,a}_t\quad \text{and}\quad \delta^{i,b}_t \geq \delta^{j,b}_t,
    \nonumber
\end{equation}
$\mathbb{P}$-a.s. for all $t\in[0,T]$.

\begin{proof}
Let $(\boldsymbol{\delta}^j)_{1\leq j\leq N}\in(\mathbb{H}^2\times\mathbb{H}^2)^{N}$ be an equilibrium strategy profile. As a necessary condition, for any player $i$ the following holds $d\mathbb{P}\times dt$ a.e. that:
\begin{equation*}
\begin{aligned}
    \delta_t^{i,a}&=\frac{\zeta}{2\gamma}+\frac{1}{2}\bar{\delta}^{i,a}_t-\mathbb{E}_t[A\,Q_T^i]-\mathbb{E}_t\int_t^T\phi_s\,Q_s^i\,ds,\\
    \delta_t^{i,b}&=\frac{\zeta}{2\gamma}+\frac{1}{2}\bar{\delta}^{i,b}_t+\mathbb{E}_t[A\,Q_T^i]+\mathbb{E}_t\int_t^T\phi_s\,Q_s^i\,ds.
\end{aligned}
\end{equation*}
Consider another player $j$ and the same procedure yields $d\mathbb{P}\times dt$ a.e. that:
\begin{equation}
\begin{aligned}
    \delta_t^{i,a}-\delta_t^{j,a}&=\frac{1}{2}(\bar{\delta}^{i,a}_t-\bar{\delta}_t^{j,a})-\mathbb{E}_t[A\,(Q_T^i-Q_T^j)]-\mathbb{E}_t\int_t^T\phi_s\,(Q_s^i-Q_s^j)\,ds,\\
    \delta_t^{i,b}-\delta_t^{j,b}&=\frac{1}{2}(\bar{\delta}^{i,b}_t-\bar{\delta}_t^{j,b})+\mathbb{E}_t[A \, (Q_T^i - Q_T^j)]+\mathbb{E}_t\int_t^T\phi_s\,(Q_s^i-Q_s^j)\,ds,\\
     \delta_t^{i,a}-\delta_t^{j,a}-\frac{1}{2}(\bar{\delta}^{i,a}_t-\bar{\delta}_t^{j,a})&=(-1)\,\Big[\delta_t^{i,b}-\delta_t^{j,b}-\frac{1}{2}(\bar{\delta}^{i,b}_t-\bar{\delta}_t^{j,b})\Big].
\end{aligned}
\label{dim_reduct}
\end{equation}
Due to the equilibrium characterization, the profile $(\boldsymbol{\delta}^i)_{1\leq i\leq N}$ solves the system \eqref{N_linear_fbsdes}. By taking the difference of equations for players $i$ and $j$, one can obtain the FBSDE
\begin{equation}
\left\{
\begin{aligned}
\;& d(Q_t^i-Q_t^j)  = -\gamma\, a_t\,\big(\bar{\delta}^{i,a}_t-\bar{\delta}^{j,a}_t-(\delta_t^{i,a}-\delta_t^{j,a})\big)dt+\gamma\, b_t\,\big(\bar{\delta}^{i,b}_t-\bar{\delta}^{j,b}_t-(\delta_t^{i,b}-\delta_t^{j,b})\big)dt, \\
& d(\delta_t^{i,a}-\delta_t^{j,a})=d\big(\bar{\delta}^{i,a}_t-\bar{\delta}^{j,a}_t\big)/2+\phi_t\,(Q_t^i-Q_t^j)\,dt-d(M_t^i-M_t^j),\\
& d(\delta_t^{i,b}-\delta_t^{j,b})=d\big(\bar{\delta}^{i,b}_t-\bar{\delta}^{j,b}_t\big)/2-\phi_t\,(Q_t^i-Q_t^j)\,dt+d(M_t^i-M_t^j),\\
& Q_0^i-Q_0^j=q_0^i-q_0^j, \quad \delta_T^{i,a}-\delta_T^{j,a}=\big(\bar{\delta}^{i,a}_T-\bar{\delta}^{j,a}_T\big)/2-A\,(Q_T^i-Q_T^j),\\
&\hspace{3.5cm}\delta_T^{i,b}-\delta_T^{j,b}=\big(\bar{\delta}^{i,b}_T-\bar{\delta}^{j,b}_T\big)/2+A\,(Q_T^i-Q_T^j).\\
\end{aligned}
\nonumber
\right.
\end{equation}
Define $\mathcal{X}_t:=Q_t^i-Q_t^j$, $\mathcal{Y}_t:=\delta_t^{i,a}-\delta_t^{j,a}-(\bar{\delta}^{i,a}_t-\bar{\delta}_t^{j,a})/2$ and $\mathcal{M}_t:=M_t^i-M_t^j$ for any $t\in[0, T]$. In view of \eqref{dim_reduct}, the above FBSDE can be rewritten as
\begin{equation}
\left\{
\begin{aligned}
\;& d\mathcal{X}_t  = -\gamma\, a_t\,\big((\bar{\delta}^{i,a}_t-\bar{\delta}^{j,a}_t)/2-\mathcal{Y}_t\big)dt+\gamma\, b_t\,\big((\bar{\delta}^{i,b}_t-\bar{\delta}^{j,b}_t)/2+\mathcal{Y}_t\big)dt, \\
& d\mathcal{Y}_t=\phi_t\,\mathcal{X}_t\,dt-d\mathcal{M}_t,\\
& \mathcal{X}_0=q_0^i-q_0^j,\quad \mathcal{Y}_T=-A\,\mathcal{X}_T.
\end{aligned}
\label{no_cross_FBSDE}
\right.
\end{equation}
When $\mathcal{Y}_t\neq 0$, note that
\begin{equation}
\begin{aligned}
    \big|\,\frac{\bar{\delta}^{i,a}_t-\bar{\delta}^{j,a}_t}{\mathcal{Y}_t}\,\big|&=\big|\,\frac{\delta_t^{j,a}\wedge\min_{k\neq i,j}\delta_t^{k,a}-\delta_t^{i,a}\wedge\min_{k\neq i,j}\delta_t^{k,a}}{\delta_t^{i,a}-\delta_t^{j,a}-(\bar{\delta}^{i,a}_t-\bar{\delta}_t^{j,a})/2}\,\big|\leq \frac{|\delta_t^{j,a}-\delta_t^{i,a}|}{|\delta_t^{i,a}-\delta_t^{j,a}|}=1,\\
    \big|\,\frac{\bar{\delta}^{i,b}_t-\bar{\delta}^{j,b}_t}{\mathcal{Y}_t}\,\big|&=\big|\,\frac{\delta_t^{j,b}\wedge\min_{k\neq i,j}\delta_t^{k,b}-\delta_t^{i,b}\wedge\min_{k\neq i,j}\delta_t^{k,b}}{\delta_t^{i,a}-\delta_t^{j,a}-(\bar{\delta}^{i,a}_t-\bar{\delta}_t^{j,a})/2}\,\big|\\
    &=\big|\,\frac{\delta_t^{j,b}\wedge\min_{k\neq i,j}\delta_t^{k,b}-\delta_t^{i,b}\wedge\min_{k\neq i,j}\delta_t^{k,b}}{\delta_t^{i,b}-\delta_t^{j,b}-(\bar{\delta}^{i,b}_t-\bar{\delta}_t^{j,b})/2}\,\big|\leq 1,
    \nonumber
\end{aligned}
\end{equation}
where we have applied the property that $\delta_t^{i,a}-\delta_t^{j,a}$ and $-(\bar{\delta}^{i,a}_t-\bar{\delta}^{j,a}_t)$ can not have different signs. While one can observe that $\bar{\delta}^{j,a}_t-\bar{\delta}^{i,a}_t=0$ if $\mathcal{Y}_t=0$, by introducing the process $\mathcal{U}$ as
\begin{equation}
    \mathcal{U}_t=\left\{
\begin{aligned}
& -\gamma\,a_t\,\Big(\frac{\bar{\delta}_t^{i,a}-\bar{\delta}_t^{j,a}}{2\,\mathcal{Y}_t}-1\Big)+\gamma\, b_t\,\Big(\frac{\bar{\delta}_t^{i,b}-\bar{\delta}_t^{j,b}}{2\,\mathcal{Y}_t}+1\Big), \quad \text{when}\; \mathcal{Y}_t\neq 0, \\
& 0, \quad \text{when}\; \mathcal{Y}_t= 0,
\end{aligned}
\right.
\nonumber
\end{equation}
equation \eqref{no_cross_FBSDE} can then be further simplified as
\begin{equation}
\left\{
\begin{aligned}
\;& d\mathcal{X}_t  = \mathcal{U}_t\, \mathcal{Y}_t\,dt, \\
& d\mathcal{Y}_t=\phi_t\,\mathcal{X}_t\,dt-d\mathcal{M}_t,\\
& \mathcal{X}_0=q_0^i-q_0^j,\quad \mathcal{Y}_T=-A\,\mathcal{X}_T.
\end{aligned}
\label{linear_FBSDE}
\right.
\end{equation}
Notice that $\mathcal{U}$ is bounded and also non-negative because
\begin{equation*}
    \frac{\bar{\delta}_t^{i,a}-\bar{\delta}_t^{j,a}}{2\,\mathcal{Y}_t}-1\leq -\frac{1}{2} \text{\quad and \quad} \frac{\bar{\delta}_t^{i,b}-\bar{\delta}_t^{j,b}}{2\,\mathcal{Y}_t}+1\geq \frac{1}{2}.
\end{equation*}
Consequently, the above FBSDE is studied in Lemma 3.6 from \cite{guo2023macroscopic}. We know it has a unique solution and the sign of $\mathcal{Y}$ differs from the initial condition of $\mathcal{X}$. Given $\mathcal{X}_0=q_0^i-q_0^j\geq0$, for $t\in[0,T]$ one has $\mathcal{X}_t\geq0$ and 
\begin{equation}
    0\geq\mathcal{Y}_t=\delta_t^{i,a}-\delta_t^{j,a}-(\bar{\delta}^{i,a}_t-\bar{\delta}_t^{j,a})/2,
    \nonumber
\end{equation}
which finally implies $\delta_t^{i,a}-\delta_t^{j,a} \leq 0$ and simultaneously $\delta_t^{i,b}-\delta_t^{j,b} \geq 0$.
\end{proof}
\end{lemma}

\kong

Following the ordering property, the next lemma looks carefully at a particular four-player game that serves as a fundamental building block for the interaction of the general $N$-player game.

\kong

\begin{lemma}[Four-player framework]
\label{four_game}
Assume $N=4$ and $q_0^1\geq q_0^2\geq q_0^3 \geq q_0^4$. Then, there exists a unique Nash equilibrium for this four-player game and the equilibrium strategy profile $(\boldsymbol{\delta}^i)_{1\leq i\leq 4}$ exhibits the following ordering property:
\begin{equation}
\begin{aligned}
    \delta_t^{1,a}\leq\delta_t^{2,a}&\leq\delta_t^{3,a}\leq\delta_t^{4,a},\\
    \delta_t^{4,b}\leq\delta_t^{3,b}&\leq\delta_t^{2,b}\leq\delta_t^{1,b},
    \label{linear_game_order}
\end{aligned}
\end{equation}
$\mathbb{P}$-a.s. for all $t\in[0, T]$.
\end{lemma}

\begin{proof}
Given Lemma \ref{no_crossing} and that $q_0^1\geq q_0^2\geq q_0^3\geq q_0^4$, we know as a necessary condition that
\begin{equation*}
\begin{aligned}
    \delta_t^{1,a}\leq\delta_t^{2,a}&\leq\delta_t^{3,a}\leq\delta_t^{4,a},\\
    \delta_t^{4,b}\leq\delta_t^{3,b}&\leq\delta_t^{2,b}\leq\delta_t^{1,b}.
\end{aligned}
\end{equation*}
Hence, instead of \eqref{N_linear_fbsdes}, it is equivalent to look at the following system : 
\begin{equation}
\left\{
\begin{aligned}
\;& dQ_t^1  = -a_t\,\big(\zeta+\gamma\delta_t^{2,a}-\gamma\delta_t^{1,a}\big)dt+b_t\,\big(\zeta+\gamma\delta_t^{4,b}-\gamma\delta_t^{1,b}\big)dt, \\
& d\delta_t^{1,a}=d\delta^{2,a}_t/2+\phi_t\,Q_t^1\,dt-dM_t^1,\\
& d\delta_t^{1,b}=d\delta^{4,b}_t/2-\phi_t\,Q_t^1\,dt+dM_t^1,\\
& Q_0^1=q_0^1,\quad \delta_T^{1,a}=\zeta/(2\gamma)+\delta_T^{2,a}/2-A\,Q_T^1,\quad \delta_T^{1,b}=\zeta/(2\gamma)+\delta_T^{4,b}/2+A\,Q_T^1;\\
\end{aligned}
\right.
\label{four_bsde_1}
\end{equation}

\vspace{0.1cm}

\begin{equation}
\left\{
\begin{aligned}
\;& dQ_t^2  = -a_t\,\big(\zeta+\gamma\delta_t^{1,a}-\gamma\delta_t^{2,a}\big)dt+b_t\,\big(\zeta+\gamma\delta_t^{4,b}-\gamma\delta_t^{2,b}\big)dt, \\
& d\delta_t^{2,a}=d\delta^{1,a}_t/2+\phi_t\,Q_t^2\,dt-dM_t^2,\\
& d\delta_t^{2,b}=d\delta^{4,b}_t/2-\phi_t\,Q_t^2\,dt+dM_t^2,\\
& Q_0^2=q_0^2,\quad \delta_T^{2,a}=\zeta/(2\gamma)+\delta_T^{1,a}/2-A\,Q_T^2,\quad \delta_T^{2,b}=\zeta/(2\gamma)+\delta_T^{4,b}/2+A\,Q_T^2;\\
\end{aligned}
\right.
\label{four_bsde_2}
\end{equation}

\vspace{0.1cm}

\begin{equation}
\left\{
\begin{aligned}
\;& dQ_t^3  = -a_t\,\big(\zeta+\gamma\delta_t^{1,a}-\gamma\delta_t^{3,a}\big)dt+b_t\,\big(\zeta+\gamma\delta_t^{4,b}-\gamma\delta_t^{3,b}\big)dt, \\
& d\delta_t^{3,a}=d\delta^{1,a}_t/2+\phi_t\,Q_t^3\,dt-dM_t^3,\\
& d\delta_t^{3,b}=d\delta^{4,b}_t/2-\phi_t\,Q_t^3\,dt+dM_t^3,\\
& Q_0^3=q_0^3,\quad \delta_T^{3,a}=\zeta/(2\gamma)+\delta_T^{1,a}/2-A\,Q_T^3,\quad \delta_T^{3,b}=\zeta/(2\gamma)+\delta_T^{4,b}/2+A\,Q_T^3;\\
\end{aligned}
\right.
\label{four_bsde_3}
\end{equation}

\vspace{0.1cm}

\begin{equation}
\left\{
\begin{aligned}
\;& dQ_t^4  = -a_t\,\big(\zeta+\gamma\delta_t^{1,a}-\gamma\delta_t^{4,a}\big)dt+b_t\,\big(\zeta+\gamma\delta_t^{3,b}-\gamma\delta_t^{4,b}\big)dt, \\
& d\delta_t^{4,a}=d\delta^{1,a}_t/2+\phi_t\,Q_t^4\,dt-dM_t^4,\\
& d\delta_t^{4,b}=d\delta^{3,b}_t/2-\phi_t\,Q_t^4\,dt+dM_t^4,\\
& Q_0^4=q_0^4,\quad \delta_T^{4,a}=\zeta/(2\gamma)+\delta_T^{1,a}/2-A\,Q_T^4,\quad \delta_T^{4,b}=\zeta/(2\gamma)+\delta_T^{3,b}/2+A\,Q_T^4.\\
\end{aligned}
\right.
\label{four_bsde_4}
\end{equation}

\noindent Our goal is to construct a solution to the affine system \eqref{four_bsde_1} - \eqref{four_bsde_4} and then show that the ordering property holds. To solve the affine system, first let us take the difference of \eqref{four_bsde_2} and \eqref{four_bsde_3} to obtain
\begin{equation}
\left\{
\begin{aligned}
\;& d(Q_t^2-Q_t^3)  =-\gamma\,a_t\,(\delta_t^{3,a}-\delta_t^{2,a})\,dt+\gamma\,b_t\,(\delta_t^{3,b}-\delta_t^{2,b})\,dt, \\
& d(\delta_t^{2,a}-\delta_t^{3,a})=\phi_t\,(Q_t^2-Q_t^3)\,dt-d(M_t^2-M_t^3),\\
& d(\delta_t^{2,b}-\delta_t^{3,b})=-\phi_t\,(Q_t^2-Q_t^3)\,dt+d(M_t^2-M_t^3),\\
& Q_0^2-Q_0^3=q_0^2-q_0^3,\quad \delta_T^{2,a}-\delta_T^{3,a}=-A\,(Q_T^2-Q_T^3),\quad \delta_T^{2,b}-\delta_T^{3,b}=A\,(Q_T^2-Q_T^3).\\
\end{aligned}
\right.
\nonumber
\end{equation}
\noindent While it is clear that $\delta_t^{2,a}-\delta_t^{3,a}=\delta_t^{3,b}-\delta_t^{2,b}$, one can then observes that $(Q^2-Q^3, \delta^{2,a}-\delta^{3,a}, M^2-M^3)$ is of the same type as \eqref{linear_FBSDE}. Based on Lemma 3.6 of \cite{guo2023macroscopic}, there exists a non-positive bounded process $(P_{2,3}(t))_{0\leq t\leq T}\in\mathbb{H}^2$ such that
\begin{equation}
    \delta_t^{2,a}-\delta_t^{3,a}=P_{2,3}(t)\cdot(Q_t^2-Q_t^3) \quad\text{and}\quad Q_t^2-Q_t^3=(q_0^2-q_0^3)\cdot e^{\int_0^t\gamma(a_u+b_u)\,P_{2,3}(u)\,du}.
    \nonumber
\end{equation}
In the same fashion, one can also take the difference of \eqref{four_bsde_1} and \eqref{four_bsde_2} to have
\begin{equation}
\left\{
\begin{aligned}
\;& d(Q_t^1-Q_t^2)  =-2\gamma\,a_t\,(\delta_t^{2,a}-\delta_t^{1,a})\,dt+\gamma\,b_t\,(\delta_t^{2,b}-\delta_t^{1,b})\,dt, \\
(3/2)\,& d(\delta_t^{1,a}-\delta_t^{2,a})=\phi_t\,(Q_t^1-Q_t^2)\,dt-d(M_t^1-M_t^2),\\
& d(\delta_t^{1,b}-\delta_t^{2,b})=-\phi_t\,(Q_t^1-Q_t^2)\,dt+d(M_t^1-M_t^2),\\
& Q_0^1-Q_0^2=q_0^1-q_0^2,\quad (3/2)\,(\delta_T^{1,a}-\delta_T^{2,a})=-A\,(Q_T^1-Q_T^2),\quad \delta_T^{1,b}-\delta_T^{2,b}=A\,(Q_T^1-Q_T^2).\\
\end{aligned}
\right.
\nonumber
\end{equation}
\noindent Setting $(3/2)\,(\delta_t^{1,a}-\delta_t^{2,a})=\delta_t^{2,b}-\delta_t^{1,b}$, it turns out that $(Q^1-Q^2,  \delta^{1,a}-\delta^{2,a})$ accepts the representation:
\begin{equation}
    \delta_t^{1,a}-\delta_t^{2,a}=P_{1,2}(t)\cdot(Q_t^1-Q_t^2) \quad\text{and}\quad Q_t^1-Q_t^2=(q_0^1-q_0^2)\cdot e^{\int_0^t\gamma(2a_u+3b_u/2)\,P_{1,2}(u)\,du},
    \nonumber
\end{equation}
for a non-positive bounded process $(P_{1,2}(t))_{0\leq t\leq T}\in\mathbb{H}^2$. Symmetrically, the difference of \eqref{four_bsde_3} and \eqref{four_bsde_4} yields
\begin{equation}
\left\{
\begin{aligned}
\;& d(Q_t^3-Q_t^4)  =-\gamma\,a_t\,(\delta_t^{4,a}-\delta_t^{3,a})\,dt+2\gamma\,b_t\,(\delta_t^{4,b}-\delta_t^{3,b})\,dt, \\
& d(\delta_t^{3,a}-\delta_t^{4,a})=\phi_t\,(Q_t^3-Q_t^4)\,dt-d(M_t^3-M_t^4),\\
(3/2)\,& d(\delta_t^{3,b}-\delta_t^{4,b})=-\phi_t\,(Q_t^3-Q_t^4)\,dt+d(M_t^3-M_t^4),\\
& Q_0^3-Q_0^4=q_0^3-q_0^4,\quad \delta_T^{3,a}-\delta_T^{4,a}=-A\,(Q_T^3-Q_T^4),\quad (3/2)\,(\delta_T^{3,b}-\delta_T^{4,b})=A\,(Q_T^3-Q_T^4).\\
\end{aligned}
\right.
\nonumber
\end{equation}
\noindent Setting $(3/2)\,(\delta_t^{3,b}-\delta_t^{4,b})=\delta_t^{4,a}-\delta_t^{3,a}$, we can represent $(Q_t^3-Q_t^4,  \delta_t^{3,a}-\delta_t^{4,a})$ by:
\begin{equation}
    \delta_t^{3,a}-\delta_t^{4,a}=P_{3,4}(t)\cdot(Q_t^3-Q_t^4) \quad\text{and}\quad Q_t^3-Q_t^4=(q_0^3-q_0^4)\cdot e^{\int_0^t\gamma(a_u+4b_u/3)\,P_{3,4}(u)\,du}.\\
    \nonumber
\end{equation}
for a non-positive bounded process $(P_{3,4}(t))_{0\leq t\leq T}\in\mathbb{H}^2$. Finally, since we have computed $\delta^{j,a}-\delta^{j+1,a}$ (and hence $\delta^{j,b}-\delta^{j+1,b})$ for $j\in\{1,2,3\}$, the right hand side of $dQ_t^1$ are known and $Q^1$ can be computed. Consequently, the equation \eqref{four_bsde_1} is reduced to a simple BSDE
\begin{equation}
\left\{
\begin{aligned}
& \frac{1}{2}\,d\delta_t^{1,a}=\frac{1}{2}\,d(\delta^{2,a}_t-\delta^{1,a}_t)+\phi_t\,Q_t^1\,dt-dM_t^1,\\
& \frac{1}{2}\,d\delta_t^{1,b}=\frac{1}{2}\,d(\delta^{4,b}_t-\delta_t^{1,b})-\phi_t\,Q_t^1\,dt+dM_t^1,\\
& Q_0^1=q_0^1,\quad \frac{1}{2}\,\delta_T^{1,a}=\frac{\zeta}{2\gamma}+\frac{1}{2}\,(\delta_T^{2,a}-\delta_T^{1,a})-A\,Q_T^1,\quad \frac{1}{2}\,\delta_T^{1,b}=\frac{\zeta}{2\gamma}+\frac{1}{2}(\delta_T^{4,b}-\delta_T^{1,b})+A\,Q_T^1.\\
\end{aligned}
\right.
\nonumber
\end{equation}
Its solution is given by
\begin{gather*}
    \frac{1}{2}\,\delta_t^{1,a}=\frac{\zeta}{2\gamma}+\frac{1}{2}\,(\delta_t^{2,a}-\delta_t^{1,a})+\int_0^t\phi_u\, Q_u^1\,du-\mathbb{E}_t\Big[\int_0^T\phi_u\,Q_u^1\,du+A\,Q_T^1\Big],\\
    \frac{1}{2}\,\delta_t^{1,b}=\frac{\zeta}{2\gamma}+\frac{1}{2}\,(\delta_t^{4,b}-\delta_t^{1,b})-\int_0^t\phi_u\,Q_u^1\,du+\mathbb{E}_t\Big[\int_0^T\phi_u\,Q_u^1\,du+A\,Q_T^1\Big],
    \nonumber
\end{gather*}
where all terms on the right hand side are known, and the expression inside of the conditional expectation is almost surely bounded. In conclusion, we have computed $(\boldsymbol{\delta}^1, \boldsymbol{\delta}^1-\boldsymbol{\delta}^2, \boldsymbol{\delta}^2-\boldsymbol{\delta}^3, \boldsymbol{\delta}^3-\boldsymbol{\delta}^4)$, which is an invertible linear transformation of $(\boldsymbol{\delta}^1, \boldsymbol{\delta}^2, \boldsymbol{\delta}^3, \boldsymbol{\delta}^4)$. This guarantees the existence of solutions. Moreover, the uniqueness of the solution can be established by noting that each of the three FBSDEs associated with $\boldsymbol{\delta}^1-\boldsymbol{\delta}^2$, $\boldsymbol{\delta}^2-\boldsymbol{\delta}^3$, and $\boldsymbol{\delta}^3-\boldsymbol{\delta}^4$ has a unique solution in $\mathbb{S}^2 \times \mathbb{S}^2 \times \mathbb{M}$, ensuring the uniqueness of $\boldsymbol{\delta}^1$ finally.

Finally, in order to derive the ordering property, we would like to summarize the above discussion as follows:
\begin{equation}
    \begin{aligned}
        \delta_t^{1,a}-\delta_t^{2,a}=P_{1,2}(t)\cdot(q_0^1-q_0^2)\cdot e^{\int_0^t\gamma(2a_u+3b_u/2)\,P_{1,2}(u)\,du} \quad&\text{and}\quad \delta_t^{2,b}-\delta_t^{1,b}= \frac{3}{2}\,(\delta_t^{1,a}-\delta_t^{2,a});\\
        \delta_t^{2,a}-\delta_t^{3,a}=P_{2,3}(t)\cdot(q_0^2-q_0^3)\cdot e^{\int_0^t\gamma(a_u+b_u)\,P_{2,3}(u)\,du} \quad&\text{and}\quad \delta_t^{2,a}-\delta_t^{3,a}=\delta_t^{3,b}-\delta_t^{2,b};\\
        \delta_t^{3,a}-\delta_t^{4,a}=P_{3,4}(t)\cdot(q_0^3-q_0^4)\cdot e^{\int_0^t\gamma(a_u+4b_u/3)\,P_{3,4}(u)\,du} \quad&\text{and}\quad \delta_t^{4,a}-\delta_t^{3,a}= \frac{3}{2}\,(\delta_t^{3,b}-\delta_t^{4,b}).
    \end{aligned}
    \nonumber
\end{equation}
for all $t\in[0,T]$. The ordering follows from the non-positiveness of $P_{j,j+1}(t)$ and non-negativeness of $(q_0^j-q_0^{j+1})$ for $j\in\{1,2,3\}$.
\end{proof}

\kong

We are now prepared to present the main result of the $N$-player game. Firstly, we solve the four-player game involving the `top' two players and the `bottom' two players. Once the equilibrium of this four-player game is determined, it becomes evident that the remaining players are simply solving an equivalent stochastic control problem.

\kong

\begin{theorem}
\label{N_players_and_4_players}
Let the index of a player be the rank of her initial inventory level, i.e.,
\begin{equation}
    q_0^1\geq q_0^2\geq q_0^3\geq\cdots\geq q_0^{N-1}\geq q_0^N.
    \nonumber   
\end{equation}
Then, there exists a unique Nash equilibrium. Moreover, the equilibrium strategy profile $(\boldsymbol{\delta}^j)_{1\leq j\leq N}$ satisfies
\begin{equation*}
\begin{aligned}
    \delta_t^{1,a} \leq \delta_t^{2,a} &\leq \cdots \leq \delta_t^{N-1,a} \leq \delta_t^{N,a},\\
    \delta_t^{1,b} \geq \delta_t^{2,b} &\geq \cdots \geq \delta_t^{N-1,b} \geq \delta_t^{N,b},
\end{aligned}
\end{equation*}
$\mathbb{P}$-a.s. for all $t\in[0, T]$.
\end{theorem}

\begin{proof}
Based on the ordering property, through a finite number of iterations of the method in the proof of Lemma \ref{four_game}, one can construct a solution $(Q^i, \boldsymbol{\delta}^i, M^i)_{i\in\{1, \dots, N\}}$ to the FBSDE system \eqref{N_linear_fbsdes}. Consequently, strategy profile $(\boldsymbol{\delta}^i)_{i\in\{1, \dots, N\}}$ forms a Nash equilibrium with the prescribed ordering. We then turn to the uniqueness. The ordering property again infers that the best ask and bid prices are determined by the four-dimensional FBSDE in Lemma \ref{four_game}, which admits a unique solution. Denoting by $\delta^{1,a}$ (resp. $\delta^{N,b}$) the obtained best ask (resp. bid) quoting strategy, the agent $i$, with $i\in\{3,\dots,N-2\}$, solves the FBSDE
\begin{equation*}
    \left\{
\begin{aligned}
\;& dQ_t^i  = -a_t\,\big(\zeta+\gamma\delta^{1,a}_t-\gamma\delta_t^{i,a}\big)dt+b_t\,\big(\zeta+\gamma\delta^{N,b}_t-\gamma\delta_t^{i,b}\big)dt, \\
& d\delta_t^{i,a}=d\delta^{1,a}_t/2+\phi_t\,Q_t^i\,dt-dM_t^i,\\
& d\delta_t^{i,b}=d\delta^{N,b}_t/2-\phi_t\,Q_t^i\,dt+dM_t^i,\\
& Q_0^i=q_0^i,\quad \delta_T^{i,a}=\zeta/(2\gamma)+\delta^{1,a}_T/2-A\,Q_T^i,\quad \delta_T^{i,b}=\zeta/(2\gamma)+\delta^{N,b}_T/2+A\,Q_T^i.
\end{aligned}
\right.
\end{equation*}
It suffices to derive its uniqueness. Let $(Q^i, \boldsymbol{\delta}^i, M^i)$, $(\tilde{Q}^i, \tilde{\boldsymbol{\delta}}^i, \tilde{M}^i)$ be two solutions and define $(\Delta Q, \Delta \boldsymbol{\delta}, \Delta M):=(\tilde{Q}^i-Q^i, \tilde{\boldsymbol{\delta}}^i-\boldsymbol{\delta}^i, \tilde{M}^i-M^i)$. It follows $(\Delta Q, \Delta \boldsymbol{\delta}, \Delta M)$ solves
\begin{equation}
    \left\{
\begin{aligned}
\;& d\Delta Q_t  = \gamma \, a_t\,\Delta \delta_t^a\,dt - \gamma \, b_t\,\Delta \delta_t^b\,dt, \\
& d\Delta \delta_t^{a}=\phi_t\,\Delta Q_t\,dt-d\Delta M_t,\\
& d\Delta \delta_t^{b}=-\phi_t\, \Delta Q_t\,dt+d\Delta M_t,\\
& \Delta Q_0=0,\quad \Delta \delta_T^{a}=-A\,\Delta Q_T, \quad \Delta \delta_T^{b}=A\,\Delta Q_T.
\label{temp_1}
\end{aligned}
\right.
\end{equation}
While it is straightforward to see $\Delta \delta^a=-\Delta \delta^b$, the FBSDE \eqref{temp_1} reduces to 
\begin{equation*}
    \left\{
\begin{aligned}
\;& d\Delta Q_t  = \gamma \, (a_t+b_t)\,\Delta \delta_t^a\,dt, \\
& d\Delta \delta_t^{a}=\phi_t\,\Delta Q_t\,dt-d\Delta M_t, \\
& \Delta Q_0=0,\quad \Delta \delta_T^{a}=-A\,\Delta Q_T,
\end{aligned}
\right.
\end{equation*}
the well-posedness of which is derived in \cite{guo2023macroscopic}. The proof is completed by noting that the unique solution to this system is $(0, 0, 0)$.
\end{proof}

\kong

\begin{remark}
(1) The equilibrium profile bears a similar economic interpretation to the optimal strategy in the control problem: players with higher inventory aim to sell more, consequently offering more favorable ask prices but less attractive bid prices. Some other economic interpretations are left to the general case later.

(2) For illustration, here we let $b\equiv0$ and define multi-dimensional processes $\boldsymbol{Q}:=(Q^1, Q^2, \dots)$ and $\boldsymbol{Y}:=(\delta^{1,a}, \delta^{2,a}, \dots)$. Then, it turns out later that the forward equations of both \eqref{N_linear_fbsdes} and \eqref{four_bsde_1}-\eqref{four_bsde_4} can be written neatly by 
\begin{equation}
    d\boldsymbol{Q}_t= \boldsymbol{H}_t\, \boldsymbol{Y}_t\,dt,
    \label{M-matrix remark}
\end{equation}
where $\boldsymbol{H}_t$ is a random $M$-matrix for all $t$. We will later introduce and examine this type of matrices in a more general context, along with investigating the associated FBSDE system.
\end{remark}

\vspace{0.2cm}

\section{General Market Making Game: Lipschitz Formulations}
\label{paper 2 section 3}
\noindent Since the Assumption \ref{linear assumption} is proposed based on the linear intensity function, we extend such assumption to the case of general intensity functions. Moreover, agents are allows to be \textit{heterogeneous} in their penalty parameters. The first goal of this section is to characterize the equilibrium of the stochastic game as the solution of Lipschitz FBSDEs. We adopt the following notations, definitions, and associated results from \cite{gueant2017optimal} regarding general intensity functions, with a slight modification.

\kong

\begin{assumption}
\label{general_inten}
A function $\Lambda:\mathbb{R}\to\mathbb{R_+}$ belongs to the class of intensity functions $\boldsymbol{\Lambda}$ if:

\vspace{0.1cm}

\begin{itemize}
    \item[1.] $\Lambda$ is twice continuously differentiable;\\
    \vspace{-0.2cm}
    
    \item[2.] $\Lambda$ is strictly decreasing and hence $\Lambda'(x)<0$ for any $x\in\mathbb{R}$;\\
    \vspace{-0.2cm}
    
    \item[3.] $\lim_{x\to\infty}\Lambda(x)=0\,$ and $\,-\infty<\inf_{x\in\mathbb{R}}\frac{\Lambda(x)\,\Lambda''(x)}{(\Lambda'(x))^2}\leq\sup_{x\in\mathbb{R}}\frac{\Lambda(x)\,\Lambda''(x)}{(\Lambda'(x))^2}\leq1$.
\end{itemize}
\label{inten_assu}
\end{assumption}

\begin{lemma}[\cite{gueant2017optimal}]
\label{inten_fun}
For any $\Lambda\in\boldsymbol{\Lambda}$, define the function $\mathcal{W}:\mathbb{R}\to\mathbb{R}$ as $\mathcal{W}(p)=\sup_{\delta\in\mathbb{R}}\Lambda(\delta)\,(\delta-p)$. Then, the following holds:

\vspace{0.1cm}

\begin{itemize}
    \item[1.] $\mathcal{W}$ is a decreasing function of class $C^2$;\\
    \vspace{-0.1cm}
    
    \item[2.] The supremum in the definition of $\mathcal{W}$ is attained at a unique $\delta^*(p)$ characterized by 
    \begin{equation}
        \delta^*(p)=\Lambda^{-1}\big(-\mathcal{W}'(p)\big),
        \nonumber
    \end{equation}
    where $\Lambda^{-1}$ denotes the inverse function of $\Lambda$;\\
    \vspace{-0.1cm}
    
    \item[3.] The function $p\mapsto\delta^*(p)$ belongs to $C^1$ and is increasing. Its derivative reads
    \begin{equation}
        (\delta^{*})'(p)=\Big[2-\frac{\Lambda(\delta^*(p))\,\Lambda''(\delta^*(p))}{\Lambda'(\delta^*(p))^2}\Big]^{-1}>0.\\
        \nonumber
    \end{equation}
\end{itemize}
\end{lemma}

\kong

\begin{remark}
The definition remains the same as in \cite{gueant2017optimal}, with the only modification being the replacement of the original inequality $\sup_{x\in\mathbb{R}}\frac{\Lambda(x)\,\Lambda''(x)}{(\Lambda'(x))^2}<2$ with $\sup_{x\in\mathbb{R}}\frac{\Lambda(x)\,\Lambda''(x)}{(\Lambda'(x))^2}\leq1$, and the addition of $-\infty<\inf_{x\in\mathbb{R}}\frac{\Lambda(x)\,\Lambda''(x)}{(\Lambda'(x))^2}$. This adjustment is made for technical reasons. Note that $\Lambda(x)=u^{-x}$ is eligible for any $u > 1$.
\end{remark}

\kong

\indent Similar to the role played by the linear intensity function in Assumption \ref{linear assumption}, we now introduce the nonlinear context characterized by the usage of general intensity functions.

\kong

\begin{assumption}[General intensity]
\label{general assumption}
The quantity of order flow executed by agent $i$ depends on the difference between her offered price and the best price offered by the others in a nonlinear way. Specifically, if we write $v_t^{i,a}$ and $v_t^{i,b}$ as the (passive) selling and buying rates of the agent at time $t$, we can express this nonlinear dependence as 
\begin{equation*}
    v_t^{i,a}=a_t\,\Lambda(\delta_t^{i,a}-\bar{\delta}_t^{i,a}) \quad\text{and}\quad v_t^{i,b}=b_t\,\Lambda(\delta_t^{i,b}-\bar{\delta}_t^{i,b}),
\end{equation*}
for some $\Lambda\in\boldsymbol{\Lambda}$. We have also assumed the bid-ask symmetry for notational convenience.
\end{assumption}

\kong

\noindent Given any admissible strategy $\boldsymbol{\delta}^i\in\mathbb{A}\times\mathbb{A}$ with
\begin{equation*}
    \mathbb{A}:=\{\delta\in\mathbb{H}^2 :\, |\delta_t|\leq \xi \, \text{ for all } t\in[0,T]\,\} \; \text{ for some constant } \, \xi>0,
\end{equation*}
the inventory and cash of the agent $i$ are given by
\begin{gather*}
    X_t^i=\int_0^t(S_u+\delta_u^{i,a})\,a_u\,\Lambda(\delta_u^{i,a}-\bar{\delta}_u^{i,a})\,du-\int_0^t(S_u-\delta_u^{i,b})\,b_u\,\Lambda(\delta_u^{i,b}-\bar{\delta}_u^{i,b})\,du,\\
    Q_t^i=q_0^i-\int_0^ta_u\,\Lambda(\delta_u^{i,a}-\bar{\delta}_u^{i,a})\,du+\int_0^tb_u\,\Lambda(\delta_u^{i,b}-\bar{\delta}_u^{i,b})\,du,
    \nonumber   
\end{gather*}
with $q_0^i\in\mathbb{R}$ representing the initial inventory. The player $i$ aims at maximizing the objective functional
\begin{equation}
\begin{aligned}
    &J^i(\boldsymbol{\delta}^i; \boldsymbol{\delta}^{-i}):=\mathbb{E}\Big[X_T^i+S_T\,Q_T^i-\int_0^T\phi_t^i\,\big(Q_t^i\big)^2\,dt-A^i\,\big(Q_T^i\big)^2\Big]\\
     &= \mathbb{E}\Big[\int_0^T\delta_t^{i,a}\,a_t\,\Lambda(\delta_t^{i,a}-\bar{\delta}_t^{i,a})\,dt+\int_0^T\delta_t^{i,b}\,b_t\,\Lambda(\delta_t^{i,b}-\bar{\delta}_t^{i,b})\,dt-\int_0^T \phi_t^i\, \big(Q_t^i\big)^2\,dt-A^i\,\big(Q_T^i\big)^2 \Big].
    \label{nonlin_game_obj}
\end{aligned}
\end{equation}
Here, for all $i\in\{1, \dots, N\}$, penalties $\phi^i:=(\phi_t^i)_{t\in[0,T]}\in\mathbb{H}^2$ and $A^i\in L^2(\Omega, \mathcal{F}_T)$ are non-negative, satisfying $\phi_t^i\leq\bar{\phi}$ and $A^i\leq\bar{A}$ for some constants $\bar{\phi}, \bar{A}>0$. The goal is to find a Nash equilibrium in the same sense as \eqref{Nash}.

\kong

\begin{definition}
An admissible strategy profile $(\hat{\boldsymbol{\delta}}^j)_{1\leq j\leq N}\in(\mathbb{A}\times\mathbb{A})^{N}$ is called
a Nash equilibrium if, for all $1\leq i\leq N$ and any admissible strategies $\boldsymbol{\delta}^i\in\mathbb{A}\times\mathbb{A}$, it holds that
\begin{equation*}
    J^i(\boldsymbol{\delta}^i; \hat{\boldsymbol{\delta}}^{-i})\leq J^i(\hat{\boldsymbol{\delta}}^i; \hat{\boldsymbol{\delta}}^{-i}).\\
\end{equation*}

\end{definition}

\kong

\begin{remark}
The constant $\xi$ in the definition of $\mathbb{A}$ serves as a regularization parameter, allowing us to formulate some Lipschitz mappings. In the homogeneous case, we will explore how this constant can be eliminated.
\end{remark}

\kong

In view of the Pontryagin stochastic maximum principle, the Hamiltonian of agent $i$ reads
\begin{equation}
\begin{aligned}
    H^i(t, q^i, y^i, \boldsymbol{\delta}^i; \boldsymbol{\delta}^{-i}) = \big[b_t\,\Lambda(\delta^{i,b}& - \bar{\delta}^{i,b}) - a_t \, \Lambda(\delta^{i,a} -\bar{\delta}^{i,a}) \big] \, y^i\\
    & + b_t \, \delta^{i,b} \, \Lambda( \delta^{i,b} -\bar{\delta}^{i,b}) + a_t \, \delta^{i, a} \, \Lambda(\delta^{i,a} - \bar{\delta}^{i,a}) -\phi_t^i\,\big(q^i\big)^2.
    \label{hamilton}
\end{aligned}
\end{equation}
While $H^i$ exhibits concavity in the state variable $Q^i$, its concavity with respect to the control $\boldsymbol{\delta}^i$ is not assured. This lack of concavity violates the typical stochastic maximum principle, outlined in works such as \cite{carmona2016lectures} and \cite{carmona2018probabilistic}. However, due to the separation between the state variable and control, we can still apply the stochastic maximum principle.

\kong

\begin{definition}
Given an admissible strategy profile $(\boldsymbol{\beta}^i)_{i\in \{1, \dots, N\}} \in (\mathbb{A}\times\mathbb{A})^N$ and the corresponding controlled inventories $(Q^1, \dots, Q^N)$, a set of $N$ pairs $(Y^i, M^i)=(Y_t^i, M_t^i)_{t\in[0, T]}$ of processes in $\mathbb{S}^2$ and $\mathbb{M}$, respectively, for $i = 1, \dots, N$, is said
to be a set of \textit{adjoint processes} associated with $(\boldsymbol{\beta}^i)_{i \in \{1, \dots, N\}}$ if they satisfy the BSDEs
\begin{equation*}
    \left\{
    \begin{aligned}
    \, dY_t^i &= 2\phi_t^i\,Q_t^i\,dt+dM_t^i,\\
    Y_T^i &= -2A^i\,Q_T^i,
    \end{aligned}
    \right.
\end{equation*}
for all $i\in\{1, \dots, N\}$.
\end{definition}

\kong

\begin{proposition}
\label{gen_sto_max}
    Consider an admissible strategy profile $(\hat{\boldsymbol{\delta}}^i)_{i=1}^N \in (\mathbb{A}\times\mathbb{A})^N$. Let $(Q^i)_{i=1}^N$ denote the corresponding controlled inventories and $(Y^i, M^i)_{i=1}^N$ represent the adjoint processes. Then $(\hat{\boldsymbol{\delta}}^i)_{i=1}^N$ forms a Nash equilibrium if and only if the generalized min-max Isaacs condition holds along the optimal paths in the following sense:
    \begin{equation}
        H^i(t, Q_t^i, Y_t^i, \hat{\boldsymbol{\delta}}_t^i; \hat{\boldsymbol{\delta}}_t^{-i} ) = \max_{\boldsymbol{\beta}^i\in[-\xi, \xi]^2}H^i(t, Q_t^i, Y_t^i, \boldsymbol{\beta}^i; \hat{\boldsymbol{\delta}}_t^{-i})
        \label{issacs condition}
    \end{equation}
$dt\times d\mathbb{P}$-a.s. for each $i\in \{1, \dots, N\}$.
\end{proposition}
\begin{proof}
See the appendix.
\end{proof}

\kong

\noindent To maximize the Hamiltonian simultaneously for all agents, it must hold for all $i$ that
\begin{equation}
\begin{aligned}
    \delta^{i,b}&=\big[\Bar{\delta}^{i,b}+\delta^*(-y^i-\Bar{\delta}^{i,b})\big]\vee(-\xi)\wedge\xi,\\
    \delta^{i,a}&=\big[\Bar{\delta}^{i,a}+\delta^*(y^i-\Bar{\delta}^{i,a})\big]\vee(-\xi)\wedge\xi;
    \label{max_hamil_1}
\end{aligned}
\end{equation}
see \cite{guo2023macroscopic} for the derivation based on Lemma \ref{inten_fun}. We claim that, for any $\boldsymbol{y}:=(y^1, \dots, y^N)\in\mathbb{R}^N$, there exist $\boldsymbol{\delta}^b:=(\delta^{1,b}, \dots, \delta^{N, b}) \in \mathbb{R}^N$ and $\boldsymbol{\delta}^a \in \mathbb{R}^N$ such that \eqref{max_hamil_1} holds. Indeed, the compactness brought by the truncation $\xi$ enables us to find such $\boldsymbol{\delta}^a$ and $\boldsymbol{\delta}^b$ through the Schauder fixed-point theorem. Define functions $\Psi^a, \Psi^b:\mathbb{R}^N \times \mathbb{R}^N \to \mathbb{R}^N$ as
\begin{equation}
\begin{aligned}
    \Psi^{i,b}(\boldsymbol{\delta}^b, \boldsymbol{y})&=\delta^{i,b}-\big[\Bar{\delta}^{i,b}+\delta^*(-y^i-\Bar{\delta}^{i,b})\big]\vee(-\xi)\wedge\xi,\\
    \Psi^{i,a}(\boldsymbol{\delta}^a, \boldsymbol{y})&=\delta^{i,a}-\big[\Bar{\delta}^{i,a}+\delta^*(y^i-\Bar{\delta}^{i,a})\big]\vee(-\xi)\wedge\xi,
\end{aligned}
\label{def_mapping}
\end{equation}
for all $i\in\{1,\dots, N\}$. Here, the additional superscript $i$ represents the $i$-th entry of the vector. Consequently, the Issacs condition can be represented as
\begin{equation}
    \Psi^{i,b}(\boldsymbol{\delta}^b, \boldsymbol{y})=0 \quad \text{and} \quad 
    \Psi^{i,a}(\boldsymbol{\delta}^a, \boldsymbol{y})=0.
    \label{equiv issac}
\end{equation}
We intend to find Lipschitz functions $\psi^a, \psi^b:\mathbb{R}^N\to\mathbb{R}^N$ such that
\begin{equation*}
    \Psi^b(\psi^b(\boldsymbol{y}), \boldsymbol{y})=0 \quad \text{and} \quad \Psi^a(\psi^a(\boldsymbol{y}), \boldsymbol{y})=0.
\end{equation*}
Functions $\psi^a$ and $\psi^b$ are known as implicit functions, and their existence, uniqueness, and regularity are the main focus of the (global) implicit function theorem. While the original theorem primarily dealt with the case when $\Psi^a$ and $\Psi^b$ are smooth, a local implicit function theorem was first introduced in \cite{clarke1990optimization} to handle locally Lipschitz non-smooth mappings. Subsequently, \cite{galewski2018global} investigated a global implicit function theorem for the same type of mappings, where the resulting implicit function was locally Lipschitz. Building upon these works, we propose a global implicit function theorem for non-smooth mappings with Lipschitz implicit functions.

\kong

\begin{proposition}
\label{general_implicit}
Assume that $F:\mathbb{R}^n\times\mathbb{R}^m\to \mathbb{R}^n$ is a locally Lipschitz mapping such that:

\kong

\begin{itemize}
    \item [1.] For every $y\in\mathbb{R}^m$, the functional $\varphi_y:\mathbb{R}^n\to\mathbb{R}$ given by the formula
    \begin{equation*}
        \varphi_y(x):=\frac{1}{2}\,|F(x,y)|^2
    \end{equation*}
    is coercive, i.e., $\lim_{|x|\to\infty}\varphi_y(x)=\infty$;\\
    
    \item[2.] The set $\partial_xF(x,y)$ is of maximal rank for all $(x,y)\in\mathbb{R}^n\times\mathbb{R}^m$;\\
    
    \item[3.] Define the function $\tilde{F}:\mathbb{R}^{m+n}\to\mathbb{R}^{m+n}$ by
    \begin{equation*}
        \tilde{F}(x,y)=\big(y, \;F(x,y)\big)
    \end{equation*}
    and let $\tilde{S}$ denote the unit sphere in $\mathbb{R}^{m+n}$. There exists a constant $\upsilon>0$ such that the distance between $\partial \tilde{F}(x,y)\,\tilde{S}$ and $0$ is at least $\upsilon$, for all $(x,y)\in\mathbb{R}^n\times\mathbb{R}^m$. Here, set $\partial \tilde{F}(x,y)\,\tilde{S}$ is defined by
    \begin{equation*}
        \partial \tilde{F}(x,y) \, \tilde{S} := \big\{ U \, v \, : \, U \in \tilde{F}(x,y) \text{\, and \,} v \in \Tilde{S} \big\}
    \end{equation*}
    and the distance refers to $\inf\{|U\,v| \, : \, U \in \tilde{F}(x,y) \text{\, and \,} v \in \Tilde{S} \}$.
\end{itemize}
\kong

\noindent Then, there exists a unique Lipschitz function $f:\mathbb{R}^m\to\mathbb{R}^n$ such that $F(x,y)=0$ and $x=f(y)$ are equivalent in the set $\mathbb{R}^n\times\mathbb{R}^m$. 
\end{proposition}
\begin{proof}
See the \hyperref[maximum prin]{appendix}.
\end{proof}

\kong

\noindent We intend to use Proposition \ref{general_implicit} to solve the Issacs condition with Lipschitz mappings. The following result in \cite{varah1975lower} turns out to be helpful.

\kong

\begin{theorem}[\cite{varah1975lower}]
\label{inverse_norm_matrix}
Assume $\boldsymbol{A}\in\mathbb{R}^{n\times n}$ is strictly diagonally dominant (by rows) matrix and set the `gap' $\alpha=\min_{1\leq k\leq n}\{\boldsymbol{A}_{kk}-\sum_{j\neq k}|\boldsymbol{A}_{kj}|\}$. Then, $\|\boldsymbol{A}^{-1}\|_{\infty}\leq 1\,/\,\alpha$, where $\|\cdot\|_\infty$ is the matrix norm induced by vector $\infty$-norm.
\end{theorem}

\kong

\begin{theorem}
\label{Verification}
(1) There exist unique Lipschitz functions $\psi^a, \psi^b:\mathbb{R}^N\to\mathbb{R}^N$ such that $\Psi^b(\boldsymbol{\delta}^b, \boldsymbol{y})=0$ and $\boldsymbol{\delta}^b=\psi^b(\boldsymbol{y})$ are equivalent in the set $\mathbb{R}^N\times\mathbb{R}^N$. The same is true for $\Psi^a(\boldsymbol{\delta}^a, \boldsymbol{y})=0$ and $\boldsymbol{\delta}^a=\psi^a(\boldsymbol{y})$.

(2) An admissible strategy profile $(\hat{\boldsymbol{\delta}}^i)_{i\in {1, \dots, N}} \in (\mathbb{A}\times\mathbb{A})^N$ forms a Nash equilibrium if and only if the profile, together with the adjoint processes $(Y^i, M^i)_{i\in {1, \dots, N}}$, satisfies the system of FBSDEs
\begin{equation}
    \left\{
    \begin{aligned}
     dQ_t^i &= b_t\,\Lambda\big(\psi^{i,b}(\boldsymbol{Y}_t)-\bar{\psi}^{i,b}(\boldsymbol{Y}_t)\big)\,dt-a_t\,\Lambda\big(\psi^{i,a}(\boldsymbol{Y}_t)-\bar{\psi}^{i,a}(\boldsymbol{Y}_t)\big)\,dt,\\
    \, dY_t^i &= 2\phi_t^i\,Q_t^i\,dt+dM_t^i,\\
    Q_0^i &= q_0^i, \quad Y_T^i = -2A^i\,Q_T^i,
    \end{aligned}
    \right.
    \label{general FBSDE}
\end{equation}
for all $i$. Here, the additional superscript $i$ of $\psi^b$ indicates the $i$-th entry of the vector and $\Bar{\psi}^{i,b}(\boldsymbol{y}):=\min_{j\neq i}\psi^{j,b}(\boldsymbol{y})$.
\end{theorem}
\begin{proof}
(1) We look at the bid side, where Proposition \ref{general_implicit} is applied for this proof. First, let us show the Lipschitz property of $\Psi^b$ and verify the first two conditions stated in Proposition \ref{general_implicit}. To show that $\Psi^b$ is Lipschitz, it suffices to prove the Lipschitz property of $\bar{\delta}^{i,b}$ for all $i$, since $\delta^*$ has a bounded derivative. Indeed, the third condition of Assumption \ref{inten_assu} along with Lemma \ref{inten_fun} infers such boundedness. Given any $\boldsymbol{\alpha}^b$, $\boldsymbol{\beta}^b\in\mathbb{R}^N$, let us set $j:=\argmin_{l\neq i}\alpha^{l,b}$, $k:=\argmin_{l\neq i}\beta^{l,b}$ and then observe the following:
\begin{equation*}
\begin{aligned}
    \text{if \;} \alpha^{j,b}\geq\beta^{k,b}, &\text{\; then \;} 0\leq\alpha^{j,b}-\beta^{k,b}\leq \alpha^{k,b}-\beta^{k,b};\\
    \text{if \;} \alpha^{j,b}\leq\beta^{k,b}, &\text{\; then \;} 0\geq\alpha^{j,b}-\beta^{k,b}\geq \alpha^{j,b}-\beta^{j,b}.\\
\end{aligned}
\end{equation*}
This observation helps us deduce that
\begin{equation*}
|\Bar{\alpha}^{i,b}-\Bar{\beta}^{i,b}|=|\alpha^{j,b}-\beta^{k,b}|\leq\max\big(|\alpha^{j,b}-\beta^{j,b}|,\,|\alpha^{k,b}-\beta^{k,b}|\big)\leq |\alpha^{j,b}-\beta^{j,b}|+|\alpha^{k,b}-\beta^{k,b}|,
\end{equation*}
and hence function $\Psi^b$ is Lipschitz. The coercive property is clear because the second term of 
\begin{equation*}
     \Psi_i^b(\boldsymbol{\delta}^b, \boldsymbol{y})=\delta^{i,b}-\big[\Bar{\delta}^{i,b}+\delta^*(-y^i-\Bar{\delta}^{i,b})\big]\vee(-\xi)\wedge\xi
\end{equation*}
is bounded by $\xi$ for all $(\boldsymbol{\delta}^b, \boldsymbol{y})$. Finally, cases when $\Psi_i^b$ is differentiable consist of
\begin{equation*}
\begin{aligned}
    \Psi_i^b(\boldsymbol{\delta}^b, \boldsymbol{y})&=\delta^{i,b}-\xi,\\
    \Psi_i^b(\boldsymbol{\delta}^b, \boldsymbol{y})&=\delta^{i,b}-\big[\delta^{j,b}+\delta^*(-y^i-\delta^{j,b})\big],\\
    \Psi_i^b(\boldsymbol{\delta}^b, \boldsymbol{y})&=\delta^{i,b}+\xi,
\end{aligned} 
\end{equation*}
for a unique index $j \neq i$. Consequently, whenever $\Psi^b$ is differentiable with respect to $\boldsymbol{\delta}^b$, the $i$-th row of the Jacobian matrix $\nabla_{\boldsymbol{\delta}^b}\Psi^b$ is a vector with $1$ in the $i$-th coordinate, and $0$ or $-1+(\delta^*)'(-y^i-\delta^{j,b})$ in the $j$-th coordinate. Since
\begin{equation*}
0<\inf_{x\in\mathbb{R}}(\delta^*)'(x)\leq\sup_{x\in\mathbb{R}}(\delta^*)'(x)\leq 1
\end{equation*}
due to Lemma \ref{inten_fun}, it follows $\nabla_{\boldsymbol{\delta}^b}\Psi^b$ is always strictly row diagonally dominant and thus is of maximal rank. Defined as the convex hull of selected limits in $\nabla_{\boldsymbol{\delta}^b} \Psi^b$, each matrix in $\partial_{\boldsymbol{\delta}^b} \Psi^b$ is also strictly diagonally dominant, hence has a maximal rank.

Define the function $\tilde{F}:\mathbb{R}^{2N}\to\mathbb{R}^{2N}$ by $\tilde{F}(\boldsymbol{y}, \boldsymbol{\delta}^b) = \big(\boldsymbol{y}, \,\Psi^b(\boldsymbol{\delta}^b, \boldsymbol{y})\big)$. Whenever $\Tilde{F}$ is differentiable, the Jacobian matrix $\nabla\Tilde{F}$ has the block form
\begin{equation}
\nabla\Tilde{F}(\boldsymbol{\delta}^b, \boldsymbol{y})=
\begin{pmatrix}
I & 0\\
\nabla_{\boldsymbol{y}}\Psi^b(\boldsymbol{\delta}^b, \boldsymbol{y}) & \nabla_{\boldsymbol{\delta}^b}\Psi^b(\boldsymbol{\delta}^b, \boldsymbol{y})
\end{pmatrix}
,
\label{block_matrix}
\end{equation}
where $I\in\mathbb{R}^{N\times N}$ is the identity matrix and $\nabla_{\boldsymbol{y}}\Psi^b(\boldsymbol{\delta}^b, \boldsymbol{y})$ is a diagonal matrix with 
\begin{equation*}
    \big[\nabla_{\boldsymbol{y}}\Psi^b(\boldsymbol{\delta}^b, \boldsymbol{y})\big]_{ii} \text{\; being \;} 0 \text{\; or \;} (\delta^*)'(-y^i-\delta^{j,b})
\end{equation*}
for all $i$ and some $j\neq i$. Taking the advantage of the singular value decomposition (SVD), we now show the smallest singular value of $\nabla\Tilde{F}(\boldsymbol{\delta}^b, \boldsymbol{y})$ is bounded away from $0$, in order to meet the last condition of Proposition \ref{general_implicit}. First, recall that the smallest singular value of a matrix is the reciprocal of the $2$-norm of its inverse, i.e.,
\begin{equation*}
    \underline{\sigma}(\boldsymbol{\delta}^b, \boldsymbol{y}) =\frac{1}{\|\nabla\Tilde{F}(\boldsymbol{\delta}^b, \boldsymbol{y})^{-1}\|_2}.
\end{equation*}
Here, we denote by $\underline{\sigma}(\boldsymbol{\delta}^b, \boldsymbol{y})$ the smallest singular value of $\nabla\Tilde{F}(\boldsymbol{\delta}^b, \boldsymbol{y})$ and use $\|\cdot\|_p$ to indicate the matrix norm induced by the vector $p$-norm. Since
\begin{equation*}
    \nabla\Tilde{F}(\boldsymbol{\delta}^b, \boldsymbol{y})^{-1}=
    \begin{pmatrix}
I & 0\\
-\nabla_{\boldsymbol{\delta}^b}\Psi^b(\boldsymbol{\delta}^b, \boldsymbol{y})^{-1}\,\nabla_{\boldsymbol{y}}\Psi^b(\boldsymbol{\delta}^b, \boldsymbol{y}) & \nabla_{\boldsymbol{\delta}^b}\Psi^b(\boldsymbol{\delta}^b, \boldsymbol{y})^{-1}
\end{pmatrix}
,
\end{equation*}
an application of the triangle inequality and the matrix norm inequality $\|\cdot\|_2\leq \sqrt{N}\,\|\cdot\|_\infty$ implies
\begin{equation*}
\begin{aligned}
    \|\nabla\Tilde{F}(\boldsymbol{\delta}^b, \boldsymbol{y})^{-1}\|_2&\leq \|I\|_2+ \|\nabla_{\boldsymbol{\delta}^b}\Psi^b(\boldsymbol{\delta}^b, \boldsymbol{y})^{-1}\,\nabla_{\boldsymbol{y}}\Psi^b(\boldsymbol{\delta}^b, \boldsymbol{y})\|_2 +\|\nabla_{\boldsymbol{\delta}^b}\Psi^b(\boldsymbol{\delta}^b, \boldsymbol{y})^{-1}\|_2\\
    &\leq 1 + \|\nabla_{\boldsymbol{\delta}^b}\Psi^b(\boldsymbol{\delta}^b, \boldsymbol{y})^{-1}\|_2 \, \|\nabla_{\boldsymbol{y}}\Psi^b(\boldsymbol{\delta}^b, \boldsymbol{y})\|_2 +\|\nabla_{\boldsymbol{\delta}^b}\Psi^b(\boldsymbol{\delta}^b, \boldsymbol{y})^{-1}\|_2\\
    &\leq 1+\big(1+\sup_{x\in\mathbb{R}}\,(\delta^*)'(x)\big)\;\|\nabla_{\boldsymbol{\delta}^b}\Psi^b(\boldsymbol{\delta}^b, \boldsymbol{y})^{-1}\|_2\\
    &\leq 1+\sqrt{N}\,\big(1+\sup_{x\in\mathbb{R}}\,(\delta^*)'(x)\big)\,\|\nabla_{\boldsymbol{\delta}^b}\Psi^b(\boldsymbol{\delta}^b, \boldsymbol{y})^{-1}\|_\infty.
\end{aligned}
\end{equation*}
In view of Theorem \ref{inverse_norm_matrix}, one can obtain a uniform lower bounded of the singular value 
\begin{equation}
\begin{aligned}
    \underline{\sigma}(\boldsymbol{\delta}^b, \boldsymbol{y}) &=\frac{1}{\|\nabla\Tilde{F}(\boldsymbol{\delta}^b, \boldsymbol{y})^{-1}\|_2}\\
    &\geq \Bigg\{\, 1 + \sqrt{N} \; \big[1+\sup_{x\in\mathbb{R}}\,(\delta^*)'(x)\big] \; \frac{1}{\inf_{x\in\mathbb{R}}\,(\delta^*)'(x)} \; \Bigg\}^{-1},
\label{small_singular_val}
\end{aligned}
\end{equation}
where we remark that the right hand side is independent of $(\boldsymbol{\delta}^b, \boldsymbol{y})$. According to the SVD, matrix $\nabla\Tilde{F}$ can be represented as
\begin{equation}
    \nabla\Tilde{F}(\boldsymbol{\delta}^b, \boldsymbol{y})=\mathscr{U}(\boldsymbol{\delta}^b, \boldsymbol{y})\,\Sigma(\boldsymbol{\delta}^b, \boldsymbol{y})\,\mathscr{V}^{*}(\boldsymbol{\delta}^b, \boldsymbol{y}),
    \label{svd}
\end{equation}
where $\mathscr{U}(\boldsymbol{\delta}^b, \boldsymbol{y})$ is a square orthogonal matrix, $\Sigma(\boldsymbol{\delta}^b, \boldsymbol{y})$ is square diagonal matrix with non-negative singular values on the diagonal, $\mathscr{V}(\boldsymbol{\delta}^b, \boldsymbol{y})$ is a square orthogonal matrix, and $\mathscr{V}^{*}(\boldsymbol{\delta}^b, \boldsymbol{y})$ is its transpose, for every $\boldsymbol{\delta}^b$ and $\boldsymbol{y}$. Because any orthogonal matrix is also an isometry, the only matrix in \eqref{svd} that will change the norm of a vector is $\Sigma(\boldsymbol{\delta}^b, \boldsymbol{y})$. While the smallest singular value of $\nabla\Tilde{F}(\boldsymbol{\delta}^b, \boldsymbol{y})$ is uniformly bounded away from $0$, the third condition of Proposition \ref{general_implicit} now holds whenever differentiable. Finally, it suffices to notice that every matrix $D \in \partial \Tilde{F}(\boldsymbol{\delta}^b, \boldsymbol{y})$ has the same partitioned structure as in \eqref{block_matrix}:
\begin{equation*}
    D =
    \begin{pmatrix}
I & 0\\
D_1 & D_2
\end{pmatrix}
,
\end{equation*}
where $D_1\in\mathbb{R}^{N\times N}$ is some diagonal matrix with entries being non-negative and bounded by $\sup_{x\in\mathbb{R}}\,(\delta^*)'(x)$,
and $D_2\in\mathbb{R}^{N\times N}$ is a $Z$-matrix (see Definition \ref{matrix_type}) with $1$ on the diagonal and positive row sums. The diagonal dominance `gap' of the $D_2$ is still $\inf_{x\in\mathbb{R}}\,(\delta^*)'(x)$. Notice that
\begin{equation*}
    D^{-1} =
    \begin{pmatrix}
I & 0\\
(D_2)^{-1}\,D_1 & (D_2)^{-1}
\end{pmatrix}
,
\end{equation*}
and a similar computation gives
\begin{equation*}
\begin{aligned}
    \|D^{-1}\|_2
    &\leq 1 + \|(D_2)^{-1}\|_2 \, \|D_1\|_2 +\|(D_2)^{-1}\|_2\\
    &\leq 1+\sqrt{N}\,\big(1+\sup_{x\in\mathbb{R}}\,(\delta^*)'(x)\big)\,\|(D_2)^{-1}\|_\infty.
\end{aligned}
\end{equation*}
Its smallest singular
value $\underline{\sigma}_D$ is bounded away from $0$ by the same constant as in \eqref{small_singular_val}:
\begin{equation*}
    \underline{\sigma}_D=\frac{1}{\|D^{-1}\|_2}\geq \Bigg\{\, 1 + \sqrt{N} \; \Big[1+\sup_{x\in\mathbb{R}}\,(\delta^*)'(x)\Big] \; \frac{1}{\inf_{x\in\mathbb{R}}\,(\delta^*)'(x)} \; \Bigg\}^{-1}.
\end{equation*}
As all conditions of Proposition \ref{general_implicit} have been checked, the first part of the proof is complete.

(2) By the stochastic maximum principle \ref{gen_sto_max}, an admissible strategy profile $(\boldsymbol{\delta}^i)_{i\in {1, \dots, N}}$ forms a Nash equilibrium if and only if the profile, together with the adjoint processes $(Y^i, M^i)_{i\in {1, \dots, N}}$, satisfies the Issacs condition \eqref{equiv issac}. Subsequently, to fulfil the Issacs condition, the implicit function theorem \ref{general_implicit} infers that $\boldsymbol{\delta}^a_t=\psi^a(\boldsymbol{Y}_t)$ and $\boldsymbol{\delta}^b_t=\psi^b(\boldsymbol{Y}_t)$, where $\psi^a, \psi^b$ are Lipschitz functions defined in the previous part of this theorem. The FBSDE system is thus the outcome of these two procedures.
\end{proof}

\kong

\begin{example}
\label{exp game}
Let $\Lambda(\delta)= \exp(-\gamma \, \delta)$ for some $\gamma>0$. By maximizing the Hamiltonian \eqref{hamilton}, the optimal control in feedback form reads
\begin{equation*}
    \delta^{i, b}= \Big( \frac{1}{\gamma}-y^i \Big)\vee(-\xi)\wedge\xi \text{ \; and \; } \delta^{i, a}= \Big( \frac{1}{\gamma}+y^i \Big) \vee(-\xi)\wedge\xi.
\end{equation*}
If we assume here the adjoint processes $(Y^i)_{i\in\{1, \dots, N\}}$ are all bounded so that the truncation does not influence if large enough, the FBSDE suggested by the stochastic maximum principle can be neatly written as
\begin{equation}
    \left\{
    \begin{aligned}
     dQ_t^i &= b_t\,\exp\Big(-\gamma\,\big(\max_{j\neq i}Y_t^j -Y_t^i \big)\Big)\,dt-a_t\,\exp\Big(-\gamma\,\big( Y_t^i - \min_{j\neq i} Y_t^j \big) \Big)\,dt,\\
    \, dY_t^i &= 2\phi_t^i\,Q_t^i\,dt+dM_t^i,\\
    Q_0^i &= q_0^i, \quad Y_T^i = -2A^i\,Q_T^i,
    \end{aligned}
    \right.
    \label{exponential FBSDE}
\end{equation}
for all $i$. Although the absence of the truncation leads to a non-Lipschitz FBSDE, we will find a unique bounded solution.
\end{example}

\vspace{0.2cm}

\section{General Game: $Z$-matrix and $M$-matrix}
\label{paper 2 section 4}
\noindent Thanks to the Lipschitz property of $\psi^a, \psi^b$ and the uniform boundedness of admissible strategy $\mathbb{A}$, the FBSDE system \eqref{general FBSDE} is therefore Lipschitz and a local well-posedness result is well-known; see \cite{carmona2018probabilistic}. To expand the analysis to global well-posedness, we need more information on the generalized derivative of the forward equations
\begin{equation*}
\begin{aligned}    
    \rho(t,\boldsymbol{Y}_t)&:=\rho^b(t,\boldsymbol{Y}_t)+\rho^a(t,\boldsymbol{Y}_t),\\
    \rho^b(t,\boldsymbol{Y}_t) &:=\Big\{ b_t\,\Lambda\big(\psi^{i,b}(\boldsymbol{Y}_t)-\bar{\psi}^{i,b}(\boldsymbol{Y}_t)\big) \Big\}_{i\in\{1,\cdots, N\}},\\
    \rho^a(t,\boldsymbol{Y}_t) &:= \Big\{ -a_t\,\Lambda\big(\psi^{i,a}(\boldsymbol{Y}_t)-\bar{\psi}^{i,a}(\boldsymbol{Y}_t)\big) \Big\}_{i\in\{1,\cdots, N\}},
\end{aligned}
\end{equation*}
with respect to $\boldsymbol{Y}$, where function $\rho:[0,T]\times\Omega\times \mathbb{R}^N\to \mathbb{R}^N$ is Lipschitz in the space variable. In this section, we will see its Jacobian matrix is a $Z_+$-matrix. Moreover, the forward equation of the corresponding variational FBSDE system will take the form of \eqref{M-matrix remark}. We first introduce essential concepts in matrix algebra that are crucial for our subsequent analysis.

\kong

\begin{definition}[\cite{horn1994topics}]
    A square matrix $\bold{A}$ is \textit{(strictly) row diagonally dominant} if $|\bold{A}_{ii}|>\sum_{j \neq i} |\bold{A}_{ij}|$ for all $i$. The \textit{column diagonal dominance} is similarly defined. A square matrix $\bold{A}$ is \textit{(strictly) diagonally dominant of row entries} if $|\bold{A}_{ii}| > |\bold{A}_{ij}|$ for any $i,j$. The \textit{diagonal dominance of column entries} is similarly defined. A matrix $\bold{A}$ is said \textit{non-negative} if $\bold{A}_{ij}\geq 0$ for all $i,j$. We remark the difference between non-negativeness and \textit{positive semi-definiteness}. 
\end{definition}

\kong

\begin{definition}
\label{matrix_type}
The class of \textit{$Z$-matrices} are those matrices whose off-diagonal entries are less than or equal to zero. Denote by \textit{$Z_+$-matrices} the class of $Z$-matrices with non-negative diagonal entries. A matrix is called \textit{non-negative stable} if all its eigenvalues have non-negative real parts. Finally, an \textit{$M$-matrix} is a $Z$-matrix that is also non-negative stable.
\end{definition}

\kong

\begin{remark}
It is worth noting that in some literature, such as \cite{horn1994topics}, an $M$-matrix is defined as a $Z$-matrix with eigenvalues whose real parts are strictly positive. In our definition, we include both the $M$-matrix and the singular $M$-matrix with respect to their terminology.
\end{remark}

\kong

\noindent $M$-matrices play a crucial role in various areas of mathematics, including matrix theory, numerical analysis, and optimization. We conclude this paragraph with a simple result as follows:

\kong

\begin{lemma}
\label{Z-matrix}
    A $Z$-matrix with non-negative row sums is an $M$-matrix. Especially, a $Z$-matrix with zero row sums, defined as a $M_0$-matrix, is an $M$-matrix. The same is true for the case of column sums.
\end{lemma}
\begin{proof}
    According to Theorem 2.5.3 of \cite{horn1994topics}, a $Z$-matrix $\boldsymbol{A}$ is a non-singular $M$-matrix if and only if $\boldsymbol{A} + t I$ is non-singular for all $t\geq0$. Let $\boldsymbol{B}$ be a $Z$-matrix with non-negative row sums. From that theorem, we know that $\boldsymbol{B} + \epsilon I$ is an $M$-matrix for any $\epsilon>0$ due to strict row dominance, and thus eigenvalues of $\boldsymbol{B} + \epsilon I$ have positive real parts. Finally, eigenvalues of $\boldsymbol{B}$ have non-negative real parts due to the continuity.
\end{proof}

\kong

Recall that implicit functions $\psi^a$ and $\psi^b$ satisfy 
\begin{equation*}
    \Psi^b(\psi^b(\boldsymbol{y}), \boldsymbol{y})=0 \quad \text{and} \quad \Psi^a(\psi^a(\boldsymbol{y}), \boldsymbol{y})=0.
\end{equation*}
We take the total derivative with respect to $\boldsymbol{y}$ on both sides to see
\begin{equation*}
    \nabla\Psi^b(\psi^b(\boldsymbol{y}), \boldsymbol{y})=\boldsymbol{0} \quad \text{and} \quad \nabla\Psi^a(\psi^a(\boldsymbol{y}), \boldsymbol{y})=\boldsymbol{0},
\end{equation*}
where $\boldsymbol{0}$ here is an zero matrix. Although  both the implicit functions $\psi^a, \psi^b$, and mappings $\Psi^a, \Psi^b$ are Lipschitz, in general we can not further write
\begin{equation*}
\begin{aligned}
    \partial\psi^b(\boldsymbol{y}) &= - \, \big[\partial_{\boldsymbol{\delta}}\Psi^b(\psi^b(\boldsymbol{y}), \boldsymbol{y})\big]^{-1} \, \partial_{\boldsymbol{y}}(\psi^b(\boldsymbol{y}), \boldsymbol{y}),\\
    \partial\psi^a(\boldsymbol{y}) &= - \, \big[\partial_{\boldsymbol{\delta}}\Psi^a(\psi^a(\boldsymbol{y}), \boldsymbol{y})\big]^{-1} \, \partial_{\boldsymbol{y}}(\psi^a(\boldsymbol{y}), \boldsymbol{y}),
\end{aligned}
\end{equation*}
which is true in the smooth case; see Theorem \ref{smooth implicit}. Regarding the calculus of the generalized derivative (see \cite{clarke2013functional} and \cite{clason2017nonsmooth} for references), while several results on the sum rule and chain rule are discussed, the subset relations presented are insufficient for our analysis of the derivative of the implicit function. Since additional regularity is required to obtain the equality relation, we can not use them directly but instead introduce a `region-by-region' method. According to the definition and the Lipschitz property, functions $\Psi^a$ and  $\Psi^b$ are non-differentiable in some closed zero-measure sets $\mathscr{D}^a$ and $\mathscr{D}^b\subset \mathbb{R}^N\times \mathbb{R}^N$ accordingly. On the other hand, the graph $(\psi^b(\boldsymbol{y}), \boldsymbol{y})$ resides in the same space $\mathbb{R}^N\times \mathbb{R}^N$. To analyze the gradient, the subsequent result divides the examination into two scenarios: when $\Psi^b$ is differentiable at $(\psi^b(\boldsymbol{y}), \boldsymbol{y})$ and when it is not at that point. The following smooth version of the implicit function theorem turns out to be helpful. Although we only state the global result, one should note that it also holds locally.

\kong

\begin{theorem}[\cite{idczak2014global}, \cite{galewski2018global}]
\label{smooth implicit}
Assume that $F : \mathbb{R}^n \times \mathbb{R}^m \to \mathbb{R}^n$ is a $C^1$ mapping such that:

\kong

\begin{itemize}
    \item [1.] For every $y\in\mathbb{R}^m$, the functional $\varphi_y:\mathbb{R}^n\to\mathbb{R}$ given by the formula
    \begin{equation*}
        \varphi_y(x):=\frac{1}{2}\,|F(x,y)|^2
    \end{equation*}
    is coercive, i.e., $\lim_{|x|\to\infty}\varphi_y(x)=\infty$;\\

    \item[2.] The Jacobian matrix $\nabla_xF(x,y)$ is of maximal rank for all $(x,y)\in\mathbb{R}^n\times\mathbb{R}^m$.
\end{itemize}
\kong
Then, there exists a unique function $f : \mathbb{R}^m\to \mathbb{R}^n$ such that equations $F(x, y) = 0$ and $x = f(y)$ are equivalent in the set $\mathbb{R}^n\times \mathbb{R}^m$. Moreover, function $f$ is also continuously differentiable with
\begin{equation*}
    \nabla f(y) = -\,\big[\nabla_x F(f(y), y)\big]^{-1} \, \nabla_y F(f(y), y).\\
\end{equation*}
\end{theorem}

\kong

\noindent The following lemma analyzes the derivative of $\rho$ in the differentiable region, utilizing the smooth implicit function theorem provided.

\kong

\begin{lemma}
\label{grad at smooth region}
For any $\boldsymbol{y}$ such that $(\psi^b(\boldsymbol{y}), \boldsymbol{y})\notin\mathscr{D}^b$ and $(\psi^a(\boldsymbol{y}), \boldsymbol{y})\notin\mathscr{D}^a$, the following statements hold:
\kong

\begin{itemize}
    \item[(i)] matrix $\nabla_{\boldsymbol{y}} \psi^b$ is element-wise uniformly bounded by some constant independent of $\xi$. It has non-positive entries and is diagonally dominant of column entries;\\
    \vspace{-0.2cm}
    
    \item[(ii)] matrix $\nabla_{\boldsymbol{y}} \psi^a$ is element-wise uniformly bounded by some constant independent of $\xi$. It has non-negative entries and is diagonally dominant of column entries;\\
    \vspace{-0.2cm}

    \item[(iii)] the Jacobian $\nabla_{\boldsymbol{y}} \rho(t,\boldsymbol{y})$ is an $M$-matrix. Notably, it is an $M_0$-matrix in the absence of $\xi$.\\
    \vspace{-0.2cm}
\end{itemize}
\end{lemma}

\begin{proof}
We start with the bid side and $\rho^b$. Since $\mathscr{D}^b$ is closed, if $(\psi^b(\boldsymbol{y}), \boldsymbol{y})\notin\mathscr{D}^b$, then there exists a neighbourhood $\mathscr{B}_{\boldsymbol{y}}$ of $\boldsymbol{y}$ such that $(\psi^b(\boldsymbol{x}), \boldsymbol{x})\notin\mathscr{D}^b$ for any $\boldsymbol{x}\in \mathscr{B}_{\boldsymbol{y}}$. Recalling the definition \eqref{def_mapping} of $\Psi^b$, the non-smoothness is caused by the min function and the truncation by $\pm\xi$. Let us first look at the case when the truncation has no effect and consider the possible ordering: 
\begin{equation}
    \psi^{1,b}(\boldsymbol{y}) < \psi^{2,b}(\boldsymbol{y}) < \cdots < \psi^{N,b}(\boldsymbol{y}).
    \label{example_ordering_bid}
\end{equation}
Given above relation, in the neighbourhood $\mathscr{B}_{\boldsymbol{y}}$, function $\Psi^b$ is equal to $\Tilde{\Psi}^b$ defined as
\begin{equation*}
\begin{aligned}
    \tilde{\Psi}^{1,b}(\boldsymbol{\delta}^b, \boldsymbol{y})&=\delta^{1,b}-\big[\delta^{2,b}+\delta^*(-y^1-\delta^{2,b})\big],\\
    \tilde{\Psi}^{i,b}(\boldsymbol{\delta}^b, \boldsymbol{y})&=\delta^{i,b}-\big[\delta^{1,b}+\delta^*(-y^i-\delta^{1,b})\big],
\end{aligned}
\end{equation*}
for $i\in\{2,\dots, N\}$. Noting that $\Tilde{\Psi}^b$ is differentiable, we then check the conditions in Theorem \ref{smooth implicit}. The property of maximal rank can be verified similarly as in Theorem \ref{Verification}. Fixing any $\boldsymbol{y}\in\mathbb{R}^N$, the triangle inequality yields
\begin{equation*}
    |\Tilde{\Psi}^{1,b}(\boldsymbol{\delta}^b, \boldsymbol{y})|+|\Tilde{\Psi}^{2,b}(\boldsymbol{\delta}^b, \boldsymbol{y})|\geq |\Tilde{\Psi}^{1,b}(\boldsymbol{\delta}^b, \boldsymbol{y})+\Tilde{\Psi}^{2,b}(\boldsymbol{\delta}^b, \boldsymbol{y})|=|\delta^*(-y^2-\delta^{1,b})+\delta^*(-y^1-\delta^{2,b})|.
\end{equation*}
If the right hand side of above equation does not explode when $|\boldsymbol{\delta}^b|\to\infty$, only the following three cases can happen because the derivative of $\delta^*$ is non-negative and bounded away from $0$:\\
\vspace{-0.2cm}

\begin{itemize}
    \item[(i)] both $\delta^{1,b}$ and $\delta^{2,b}$ are bounded;\\
    \vspace{-0.2cm}

    \item[(ii)] $\delta^{1,b}\to-\infty$ and $\delta^{2,b}\to\infty$;\\
    \vspace{-0.2cm}

    \item[(iii)] $\delta^{1,b}\to\infty$ and $\delta^{2,b}\to-\infty$.\\
    \vspace{-0.2cm}
\end{itemize}
If case (i) is true, since $|\boldsymbol{\delta}^b|\to\infty$, there exists an index $j\notin \{1, 2\}$ such that $|\delta^{j,b}|\to\infty$ and thus $|\Tilde{\Psi}^{j, b}(\boldsymbol{\delta}^b, \boldsymbol{y})|=|\delta^{j,b}-[\delta^{1,b}+\delta^*(-y^j-\delta^{1,b})]|\to\infty$. Turning to case (ii), in view of $(\delta^*)'\in(0, 1)$, one can see $\delta^{2,b}+\delta^*(-y^1-\delta^{2,b})$ is increasing with respect to $\delta^{2,b}$. This gives
\begin{equation*}
    \Tilde{\Psi}^{1,b}(\boldsymbol{\delta}^b, \boldsymbol{y}) = \delta^{1,b}-\big[\delta^{2,b}+\delta^*(-y^1-\delta^{2,b})\big]\to-\infty.
\end{equation*}
Finally, since case (iii) is symmetric to case (ii), we see the $\ell_1$-norm of $\Tilde{\Psi}^b(\boldsymbol{\delta}^b, \boldsymbol{y})$ will always explode when $|\boldsymbol{\delta}^b|\to\infty$. The coercive property follows from the equivalence of norms in the finite-dimensional space. Therefore, the local version of Theorem \ref{smooth implicit} guarantees an implicit function $\tilde{\psi}^b$ and the uniqueness further implies $\tilde{\psi}^b = \psi^b$ in the neighborhood $\mathscr{B}_{\boldsymbol{y}}$. We then learn that $\psi^b$ is differentiable in $\mathscr{B}_{\boldsymbol{y}}$ with
\begin{equation*}
    \nabla \psi^b(\boldsymbol{y}) = -\,\big[\nabla_{\boldsymbol{\delta}} \tilde{\Psi}^b(\psi^b(\boldsymbol{y}), \boldsymbol{y})\big]^{-1}\, \nabla_{\boldsymbol{y}} \tilde{\Psi}^b(\psi^b(\boldsymbol{y}), \boldsymbol{y}),
\end{equation*}
for any $\boldsymbol{y}\in\mathscr{B}_{\boldsymbol{y}}$. Note that $\rho^b$ is thus also differentiable at such $\boldsymbol{y}$.

The Jacobian matrix $\nabla_{\boldsymbol{\delta}} \tilde{\Psi}^b(\psi^b(\boldsymbol{y}), \boldsymbol{y})$ reads
\begin{equation*}
\begin{bmatrix}
1 & (-1)\,\big[1-(\delta^*)'(-\psi^{2,b}(\boldsymbol{y})-y^1)\big] & 0 & \cdots & 0\\
(-1)\,\big[1-(\delta^*)'(-\psi^{1,b}(\boldsymbol{y})-y^2)\big] & 1 & 0 & \cdots & 0 \\
(-1)\,\big[1-(\delta^*)'(-\psi^{1,b}(\boldsymbol{y})-y^3)\big] & 0 & 1 & \cdots & 0 \\
\vdots & \vdots & \vdots & \ddots & \vdots\\
(-1)\,\big[1-(\delta^*)'(-\psi^{1,b}(\boldsymbol{y})-y^N)\big] & 0 & 0 & \cdots & 1 
\end{bmatrix}.
\end{equation*}
Define
\begin{equation*}
    \mathfrak{a}_1=(-1)\,\big[ 1 - (\delta^*)'(-\psi^{2,b}(\boldsymbol{y})-y^1) \big]  \text{\; and \;} \mathfrak{a}_i=(-1) \, \big[ 1 - (\delta^*)'(-\psi^{1,b}(\boldsymbol{y})-y^i) \big]
\end{equation*}
for $i\in\{2,\dots,N\}$ and observe that each $a_i\in(-1, 0]$. Direct calculations yield
\begin{equation*}
    \big[\nabla_{\boldsymbol{\delta}} \tilde{\Psi}^b(\psi^b(\boldsymbol{y}), \boldsymbol{y})\big]^{-1}=\frac{1}{1-\mathfrak{a}_1\mathfrak{a}_2}\,
    \begin{bmatrix}
1 & - \mathfrak{a}_1 & 0 & \cdots & 0\\
- \mathfrak{a}_2 & 1 & 0 & \cdots & 0 \\
- \mathfrak{a}_3 & \mathfrak{a}_1 \mathfrak{a}_3 & 1-\mathfrak{a}_1 \mathfrak{a}_2 & \cdots & 0 \\
\vdots & \vdots & \vdots & \ddots & \vdots\\
-\mathfrak{a}_N& \mathfrak{a}_1 \mathfrak{a}_N & 0 & \cdots & 1-\mathfrak{a}_1 \mathfrak{a}_2
\end{bmatrix}.
\end{equation*}
Since the Jacobian matrix $\nabla_{\boldsymbol{y}} \tilde{\Psi}^b(\psi^b(\boldsymbol{y}), \boldsymbol{y})$ reads
\begin{equation*}
\nabla_{\boldsymbol{y}} \tilde{\Psi}^b(\psi^b(\boldsymbol{y}), \boldsymbol{y})=\,
\begin{bmatrix}
    \mathfrak{a}_1+1 & 0 & \cdots & 0\\
    0 & \mathfrak{a}_2+1 & \cdots & 0\\
    \vdots & \vdots & \ddots & \vdots\\
    0 & 0 & \cdots & \mathfrak{a}_N+1
\end{bmatrix},
\end{equation*}
we can compute the Jacobian matrix $\nabla \psi^b(\boldsymbol{y})$ via
\begin{equation}
\begin{aligned}
    \nabla\psi^b(\boldsymbol{y}) &= \frac{1}{1-\mathfrak{a}_1\mathfrak{a}_2} \cdot\\
&    
\begin{bmatrix}
    -\mathfrak{a}_1-1 & \mathfrak{a}_1(\mathfrak{a}_2+1) & 0 & \cdots & 0\\
    \mathfrak{a}_2(\mathfrak{a}_1+1) & -\mathfrak{a}_2-1 & 0 & \cdots & 0\\
    \mathfrak{a}_3(\mathfrak{a}_1+1) & -\mathfrak{a}_1\mathfrak{a}_3(\mathfrak{a}_2+1) & (\mathfrak{a}_1\mathfrak{a}_2-1)(\mathfrak{a}_3+1) & \cdots & 0\\
    \vdots & \vdots & \vdots & \ddots & \vdots\\
    \mathfrak{a}_N(\mathfrak{a}_1+1) & -\mathfrak{a}_1\mathfrak{a}_N(\mathfrak{a}_2+1) & 0 & \cdots &(\mathfrak{a}_1\mathfrak{a}_2-1)(\mathfrak{a}_N+1)
\end{bmatrix}.
\label{Jacobian of control 1}
\end{aligned}
\end{equation}
It is evident that $\nabla\psi^b(\boldsymbol{y})$ has non-positive entries and is diagonally dominant of column entries, for which an element-wise bound can be
\begin{equation*}
    \frac{2}{ 1 - \big[ 1 - \inf_{x\in\mathbb{R}}(\delta^*)'(x) \big]^2 }.
\end{equation*}
Because we are considering the scenario \eqref{example_ordering_bid}, each agent is affected by the best offer (or equivalently the smallest gap $\delta$) from the others. The interaction between agents now is
\begin{equation*}
    \Xi^b(\boldsymbol{y}):=\Big(\psi^{1,b}(\boldsymbol{y})-\psi^{2,b}(\boldsymbol{y}),\, \psi^{2,b}(\boldsymbol{y})-\psi^{1,b}(\boldsymbol{y}),\, \psi^{3,b}(\boldsymbol{y})-\psi^{1,b}(\boldsymbol{y}),\, \dots,\, \psi^{N,b}(\boldsymbol{y})-\psi^{1,b}(\boldsymbol{y})\Big),
\end{equation*}
and the Jacobian $\nabla \Xi^b(\boldsymbol{y})$ is given by
\begin{equation}
\begin{aligned}
        \nabla &\Xi^b(\boldsymbol{y})=\frac{1}{1-\mathfrak{a}_1\mathfrak{a}_2}\cdot\\
&
\begin{bmatrix}
    -(\mathfrak{a}_1+1)(\mathfrak{a}_2+1) & (\mathfrak{a}_1+1)(\mathfrak{a}_2+1) & 0 & \cdots & 0\\
    
    (\mathfrak{a}_1+1)(\mathfrak{a}_2+1) & -(\mathfrak{a}_1+1)(\mathfrak{a}_2+1) & 0 & \cdots & 0\\
    
   (\mathfrak{a}_1+1)(\mathfrak{a}_3+1) & -\mathfrak{a}_1(\mathfrak{a}_2+1)(\mathfrak{a}_3+1) & (\mathfrak{a}_1\mathfrak{a}_2-1)(\mathfrak{a}_3+1) & \cdots & 0\\
   
    \vdots & \vdots & \vdots & \ddots & \vdots\\
    
    (\mathfrak{a}_1+1)(\mathfrak{a}_N+1) & -\mathfrak{a}_1(\mathfrak{a}_2+1)(\mathfrak{a}_N+1) & 0 & \cdots &(\mathfrak{a}_1\mathfrak{a}_2-1)(\mathfrak{a}_N+1)\\
\end{bmatrix}.
\label{jacobian of interaction}
\end{aligned}
\end{equation}

\noindent It is straightforward to observe that the matrix \eqref{jacobian of interaction} has a zero row sum. Recalling $\mathfrak{a}_k\in(-1,0]$ for all $k\in\{1, \dots, N\}$, we can further determine
\begin{equation*}
    (\mathfrak{a}_1+1)(\mathfrak{a}_i+1) \geq 0, \quad
    -\mathfrak{a}_1(\mathfrak{a}_1+1)(\mathfrak{a}_i+1) \geq 0, \quad \mathfrak{a}_1\mathfrak{a}_2-1 \leq 0.
\end{equation*}
Consequently, the matrix \eqref{jacobian of interaction} has non-positive diagonal entries and non-negative off-diagonal entries, with zero row sums. Because $\Lambda$ is a decreasing function, the Jacobian $\nabla_{\boldsymbol{y}}\rho^b(t,\boldsymbol{y})$ equals to the multiplication of a non-positive diagonal matrix with $\nabla \Xi^b(\boldsymbol{y})$. Therefore, the matrix $\nabla_{\boldsymbol{y}}\rho^b(t,\boldsymbol{y})$ exhibits non-negative diagonal entries and non-positive off-diagonal entries, with zero row sums. By Lemma \ref{Z-matrix}, it is an $M_0$-matrix and thus an $M$-matrix. Since $\Psi^b$ is symmetric with respect to the index, any permutation of the ordering \eqref{example_ordering_bid} will result in a matrix with the same property as above.

Then, let us turn to the case when the truncation by $\xi$ is in effect. That is to say, given $\boldsymbol{y}$, there exists some $k$ such that
\begin{equation*}
    \Psi^{k,b}(\psi^b(\boldsymbol{y}), \boldsymbol{y})=\psi^{k,b}(\boldsymbol{y})-\big[\Bar{\psi}^{k,b}(\boldsymbol{y})+\delta^*(-y^k-\Bar{\psi}^{k,b}(\boldsymbol{y}))\big]\vee(-\xi)\wedge\xi=\psi^{k,b}(\boldsymbol{y})-\xi.
\end{equation*}
Suppose this is true for all $k$, condition \eqref{equiv issac} yields $\psi^{k,b}(\boldsymbol{y})=\xi$ for any $k$ and hence $\nabla_{\boldsymbol{y}}\rho^b(t,\boldsymbol{y})$ is a zero matrix, which is an $M$-matrix. Next, suppose there exists only one index, denoted by $1$, such that the truncation has no influence:
\begin{equation*}
\begin{aligned}
    \Psi^{1,b}(\psi^b(\boldsymbol{y}), \boldsymbol{y})&=\psi^{1,b}(\boldsymbol{y})-\big[\Bar{\psi}^{1,b}(\boldsymbol{y})+\delta^*(-y^1-\Bar{\psi}^{1,b}(\boldsymbol{y}))\big]\vee(-\xi)\wedge\xi\\
    &=\psi^{1,b}(\boldsymbol{y})-\big[\xi+\delta^*(-y^1-\xi)\big].
\end{aligned}
\end{equation*}
One can observe that $\psi^{1,b}(\boldsymbol{y})=\psi^{1,b}(y^1)$ is decreasing with respect to $y^1$. The interaction between agents is
\begin{equation*}
    \Xi^b(\boldsymbol{y}):=\Big(\psi^{1,b}(\boldsymbol{y})-\xi,\, \xi-\psi^{1,b}(\boldsymbol{y}),\, \xi-\psi^{1,b}(\boldsymbol{y}),\, \dots,\, \xi-\psi^{1,b}(\boldsymbol{y})\Big).
\end{equation*}
Because $\Lambda$ is decreasing, for constants $C_i\geq 0$, it can be deduced that
\begin{equation*}
\nabla_{\boldsymbol{y}} \rho^b(t, \boldsymbol{y})=\,
\begin{bmatrix}
    \hspace{0.28cm} C_1 & 0 & \cdots & 0\\
    -C_2 & 0 & \cdots & 0\\
    \vdots & \vdots & \ddots & \vdots\\
    -C_N & 0 & \cdots & 0
\end{bmatrix},
\end{equation*}
which is an $M$-matrix by the method in the proof of Lemma \ref{Z-matrix}. Finally, if there exist more than one index such that the truncation has no influence. In this case, let us suppose that truncation is only applied to index $1$ but not the others, that is,
\begin{equation*}
\begin{aligned}
    \Psi^{1,b}(\psi^b(\boldsymbol{y}), \boldsymbol{y})&=\psi^{1,b}(\boldsymbol{y})-\big[\Bar{\psi}^{1,b}(\boldsymbol{y})+\delta^*(-y^1-\Bar{\psi}^{1,b}(\boldsymbol{y}))\big]\vee(-\xi)\wedge\xi=\psi^{1,b}(\boldsymbol{y})-\xi,\\
    \Psi^{k,b}(\psi^b(\boldsymbol{y}), \boldsymbol{y})&=\psi^{k,b}(\boldsymbol{y})-\big[\Bar{\psi}^{k,b}(\boldsymbol{y})+\delta^*(-y^k-\Bar{\psi}^{k,b}(\boldsymbol{y}))\big]\vee(-\xi)\wedge\xi\\
    &=\psi^{k,b}(\boldsymbol{y})-\big[\Bar{\psi}^{k,b}(\boldsymbol{y})+\delta^*(-y^k-\Bar{\psi}^{k,b}(\boldsymbol{y}))\big],
\end{aligned}
\end{equation*}
for $k \neq 1$. Note that index $1$ has no impact on the remaining $N-1$ indices. If we further assume the ordering $\psi^{2,b}(\boldsymbol{y}) < \psi^{3,b}(\boldsymbol{y}) < \cdots < \psi^{N,b}(\boldsymbol{y})$, the result from the previous discussion yields
\begin{equation*}
    \nabla\psi^b(\boldsymbol{y})=\frac{1}{1-\mathfrak{c}_2\mathfrak{c}_3} \begin{bmatrix}
        0 & 0 & 0 & \cdots & 0\\
     0 & -\mathfrak{c}_2-1 & \mathfrak{c}_2(\mathfrak{c}_3+1) & \cdots & 0\\
    0 & \mathfrak{c}_3(\mathfrak{c}_2+1) & -\mathfrak{c}_3-1 & \cdots & 0\\
    \vdots & \vdots & \vdots & \ddots & \vdots\\
    0 & \mathfrak{c}_N(\mathfrak{c}_2+1) & -\mathfrak{c}_2\mathfrak{c}_N(\mathfrak{c}_3+1) & \cdots &(\mathfrak{c}_2\mathfrak{c}_3-1)(\mathfrak{c}_N+1)
    \end{bmatrix},
\end{equation*}
where
\begin{equation*} 
    \mathfrak{c}_2=(-1)\,\big[1-(\delta^*)'(-\psi^{3,b}(\boldsymbol{y})-y^2)\big] \in (-1, 0], \; \mathfrak{c}_i=(-1)\,\big[1-(\delta^*)'(-\psi^{2,b}(\boldsymbol{y})-y^i)\big] \in (-1,0]
\end{equation*}
for $i\in\{3, \cdots, N\}$. Since the interaction between agents now is
\begin{equation*}
    \Big(\psi^{1,b}(\boldsymbol{y})-\psi^{2,b}(\boldsymbol{y}),\, \psi^{2,b}(\boldsymbol{y})-\psi^{3,b}(\boldsymbol{y}),\, \psi^{3,b}(\boldsymbol{y})-\psi^{2,b}(\boldsymbol{y}),\, \dots,\, \psi^{N,b}(\boldsymbol{y})-\psi^{2,b}(\boldsymbol{y})\Big),
\end{equation*}
the Jacobian of $\rho^b$ admits the representation as the following block matrix:
\begin{equation*}
\nabla_{\boldsymbol{y}} \rho^b(t, \boldsymbol{y})=\,
\begin{bmatrix}
    0 & -\Vec{C}\\
    \Vec{0} & \hspace{0.3cm} D
\end{bmatrix},
\end{equation*}
where $\Vec{0}\in\mathbb{R}^{N-1}$ is the zero vector, $\Vec{C}\in\mathbb{R}^{N-1}$ is a vector with non-negative entries, and $D\in\mathbb{R}^{(N-1)\times(N-1)}$ is an $M_0$-matrix. Through an application of the Laplace expansion, one can see $\nabla_{\boldsymbol{y}} \rho^b(t, \boldsymbol{y})$ is an $M$-matrix. Since $\Psi^b$ is symmetric with respect to the index, above situations with different indices will not change the desired property of the Jacobian matrix.

Finally, we examine the case when the truncation by $-\xi$ happens. Especially, given $\boldsymbol{y}$, there exists some $k$ such that
\begin{equation*}
    \Psi^{k,b}(\psi^b(\boldsymbol{y}), \boldsymbol{y})=\psi^{k,b}(\boldsymbol{y})-\big[\Bar{\psi}^{k,b}(\boldsymbol{y})+\delta^*(-y^k-\Bar{\psi}^{k,b}(\boldsymbol{y}))\big]\vee(-\xi)\wedge\xi=\psi^{k,b}(\boldsymbol{y})+\xi.
\end{equation*}
Suppose this is true for multiple indices; for example $1$ and $2$. Since $-\xi$ is the smallest value admissible, the smallest and second smallest value of the vector $\psi^b(\boldsymbol{y})$ are then both $-\xi$. As a result, for $j\notin \{1,2\}$ it holds that
\begin{equation*}
    \psi^{j,b}(\boldsymbol{y})=\Bar{\psi}^{j,b}(\boldsymbol{y})+\delta^*(-y^j-\Bar{\psi}^{j,b}(\boldsymbol{y}))=-\xi+\delta^*(-y^j+\xi),
\end{equation*}
which infers $\psi^{j,b}(\boldsymbol{y})=\psi^{j,b}(y^j)$ is decreasing in $y^j$. Because the interaction between agents now reads
\begin{equation*}
    \Big(\psi^{1,b}(\boldsymbol{y})-\psi^{2,b}(\boldsymbol{y})=0,\, \psi^{2,b}(\boldsymbol{y})-\psi^{1,b}(\boldsymbol{y})=0,\, \psi^{3,b}(\boldsymbol{y})+\xi,\, \dots,\, \psi^{N,b}(\boldsymbol{y})+\xi\Big),
\end{equation*}
the Jacobian of $\rho^b$ admits the following block representation:
\begin{equation*}
\nabla_{\boldsymbol{y}} \rho^b(t, \boldsymbol{y})=\,
\begin{bmatrix}
    O & O\\
    O & D
\end{bmatrix},
\end{equation*}
where $O\in\mathbb{R}^{2\times2}$ is the zero matrix and $D\in\mathbb{R}^{2\times2}$ is a diagonal matrix with non-negative entries. It is now straightforward to see that $\nabla_{\boldsymbol{y}} \rho^b(t, \boldsymbol{y})$ is an $M$-matrix. Next, suppose there is only index $1$ with truncation being applied. The definition of implicit function then yields
\begin{equation*}
    \psi^{1,b}(\boldsymbol{y})=-\xi \text{ \; and \; } \psi^{j,b}(\boldsymbol{y})=-\xi+\delta^*(-y^j+\xi),
\end{equation*}
for $j \neq 1$. Let $m=\argmin_{j \neq 1} \psi^{j,b}(\boldsymbol{y})$. With the interaction being
\begin{equation*}
    \Big(\psi^{1,b}(\boldsymbol{y}) - \psi^{m,b}(\boldsymbol{y}),\, \psi^{2,b}(\boldsymbol{y}) + \xi,\, \psi^{3,b}(\boldsymbol{y}) + \xi,\, \dots,\, \psi^{N,b}(\boldsymbol{y}) + \xi \Big),
\end{equation*}
the Jacobian of $\rho^b$ admits the representation as the following block matrix:
\begin{equation*}
\nabla_{\boldsymbol{y}} \rho^b(t, \boldsymbol{y})=\,
\begin{bmatrix}
    0 & -\Vec{C}\\
    \Vec{0} & \hspace{0.3cm} D
\end{bmatrix}.
\end{equation*}
Here, again $\Vec{C}$ is a non-negative vector, $D$ is non-negative diagonal matrix, and thus $\nabla_{\boldsymbol{y}} \rho^b(t, \boldsymbol{y})$ is an $M$-matrix. We remark again that change of index will not alter the targeted property of the matrix, and also the composition of different truncation can be discussed likewise.

The ask side can be verified in like manner, while we briefly summarize the calculation results when all truncation is of no effect. Consider the possible ordering 
\begin{equation}
    \psi^a_1(\boldsymbol{y})>\psi^a_2(\boldsymbol{y})>\cdots>\psi^a_N(\boldsymbol{y}).
    \label{example_ordering_ask}
\end{equation}
The interaction between agents now is
\begin{equation*}
    \Xi^a(\boldsymbol{y}):=\Big(\psi_1^a(\boldsymbol{y})-\psi_N^a(\boldsymbol{y}),\, \psi_2^a(\boldsymbol{y})-\psi_N^a(\boldsymbol{y}),\, \dots,\, \psi_{N-1}^a(\boldsymbol{y})-\psi_N^a(\boldsymbol{y}),\, \psi_N^a(\boldsymbol{y})-\psi_{N-1}^a(\boldsymbol{y})\Big).
\end{equation*}
Let us define
\begin{equation}
    \mathfrak{b}_N=(-1)\,\big[1-(\delta^*)'(y^N-\psi^{N-1,a}(\boldsymbol{y}))\big]  \text{\; and \;} \mathfrak{b}_i=(-1)\,\big[1-(\delta^*)'(y^i-\psi^{N,a}(\boldsymbol{y}))\big]
    \label{jacobian of interaction_2}
\end{equation}
for $i\in\{2,\dots,N\}$ and notice that $\mathfrak{b}_k\in(-1,0]$ for all $k\in\{1,\dots,N\}$. The Jacobian of the control $\psi^a(\boldsymbol{y})$ is given by
\begin{equation}
\begin{aligned}
    &\nabla  \psi^a(\boldsymbol{y})=\frac{1}{1-\mathfrak{b}_{N-1}\mathfrak{b}_N}\cdot\\
    &
\begin{bmatrix}
    (1-\mathfrak{b}_{N-1}\mathfrak{b}_N)(\mathfrak{b}_1+1) & 0 & \cdots & \mathfrak{b}_1\mathfrak{b}_N(\mathfrak{b}_{N-1}+1)& -\mathfrak{b}_1(\mathfrak{b}_{N}+1)\\
    
    0 & (1-\mathfrak{b}_{N-1}\mathfrak{b}_N)(\mathfrak{b}_2+1) & \cdots & \mathfrak{b}_2\mathfrak{b}_N(\mathfrak{b}_{N-1}+1)& -\mathfrak{b}_2(\mathfrak{b}_{N}+1)\\
   
    \vdots & \vdots & \ddots & \vdots & \vdots\\
    
    0 & 0 & \cdots & \mathfrak{b}_{N-1}+1& -\mathfrak{b}_{N-1}(\mathfrak{b}_N+1)\\

    0 & 0 & \cdots & -(\mathfrak{b}_{N-1}+1)\mathfrak{b}_N & \mathfrak{b}_N+1
\end{bmatrix}.
\label{Jacobian of control 2}
\end{aligned}
\end{equation}
and resulting $\nabla \Xi^a(\boldsymbol{y})$ reads
\begin{equation}
\begin{aligned}
    &\nabla \Xi^a(\boldsymbol{y})=\frac{1}{1-\mathfrak{b}_{N-1}\mathfrak{b}_N}\cdot\\
&\begin{bmatrix}
    (1-\mathfrak{b}_{N-1}\mathfrak{b}_N)(\mathfrak{b}_1+1) & 0 & \cdots & \mathfrak{b}_N(\mathfrak{b}_1+1)(\mathfrak{b}_{N-1}+1)& -(\mathfrak{b}_1+1)(\mathfrak{b}_{N}+1)\\
    
    0 & (1-\mathfrak{b}_{N-1}\mathfrak{b}_N)(\mathfrak{b}_2+1) & \cdots & \mathfrak{b}_N(\mathfrak{b}_2+1)(\mathfrak{b}_{N-1}+1)& -(\mathfrak{b}_2+1)(\mathfrak{b}_{N}+1)\\

    \vdots & \vdots & \ddots & \vdots & \vdots\\
    
    0 & 0 & \cdots & (\mathfrak{b}_{N-1}+1)(\mathfrak{b}_N+1)& -(\mathfrak{b}_{N-1}+1)(\mathfrak{b}_N+1)\\

    0 & 0 & \cdots & -(\mathfrak{b}_{N-1}+1)\,(\mathfrak{b}_N+1) & (\mathfrak{b}_{N-1}+1)(\mathfrak{b}_N+1)
\end{bmatrix}.
\label{jacobian of ask interaction}
\end{aligned}
\end{equation}
Since $-\Lambda$ is increasing, the Jacobian $\nabla_{\boldsymbol{y}}\rho^a(t,\boldsymbol{y})$ is the multiplication of a non-negative diagonal matrix with $\nabla \Xi^a(\boldsymbol{y})$. Therefore, the matrix $\nabla_{\boldsymbol{y}}\rho^a(t,\boldsymbol{y})$ exhibits non-negative diagonal entries and non-positive off-diagonal entries, with zero row sums. By Lemma \ref{Z-matrix}, it is an $M_0$-matrix and thus an $M$-matrix.
\end{proof}

\kong

\noindent We then seek the generalization of the above result to the whole space.

\kong

\begin{theorem}
\label{global_Jacobian}
    For all $\boldsymbol{y}\in \mathbb{R}^N$, the following statements hold:
    \kong
    
\begin{itemize}
    \item[(i)] any matrix in $\partial_{\boldsymbol{y}} \psi^b(\boldsymbol{y})$ is element-wise uniformly bounded by some constant independent of $\xi$. The matrix is non-positive and is diagonally dominant of column entries;\\
    \vspace{-0.2cm}
    
    \item[(ii)] any matrix in $\partial_{\boldsymbol{y}} \psi^a(\boldsymbol{y})$ is element-wise uniformly bounded by some constant independent of $\xi$. The matrix is non-negative and is diagonally dominant of column entries;\\
    \vspace{-0.2cm}

    \item[(iii)] any matrix in $\partial_{\boldsymbol{y}}\rho(t,\boldsymbol{y})$ is an $Z_+$-matrix, and notably an $M_0$-matrix in the absence of $\xi$.\\
    \vspace{-0.2cm}
\end{itemize}
\end{theorem}

\begin{proof}

We begin with the bid side. Fix any $\boldsymbol{y}\in \mathbb{R}^N$. Since $\rho^b$ is Lipschitz with respect to $\boldsymbol{y}$, the set $\partial_{\boldsymbol{y}}\rho^b(t,\boldsymbol{y})$ is the convex hull generated by matrices of the kind $\lim_{\boldsymbol{y}_n\to \boldsymbol{y}} \nabla_{\boldsymbol{y}}\rho^b(t,\boldsymbol{y}_n)$, where $\rho^b$ is differentiable at $\boldsymbol{y}_n$ for each $n$. Considering that Lemma \ref{grad at smooth region} has studied $\boldsymbol{y}$ such that $(\psi^b(\boldsymbol{y}), \boldsymbol{y})\notin\mathscr{D}^b$, it suffices to explore $\boldsymbol{y}$ such that $\Psi^b$ is not differentiable at $(\psi^b(\boldsymbol{y}), \boldsymbol{y})$. Scenarios in $\mathscr{D}^b$ can be categorized into three regions:
\kong

\begin{itemize}
    \item[(i)] there are multiple indices achieving the minimum, i.e., for some indices $j$ and $k$, it holds
    \begin{equation}
    \psi^{j,b}(\boldsymbol{y})=\psi^{k,b}(\boldsymbol{y})=\min_l \psi^{l,b}(\boldsymbol{y});
        \label{region i}
    \end{equation}
    
    \vspace{-0.1cm}
    
    \item[(ii)] there are multiple indices achieving the second minimum, i.e., for some indices $j$ and $k$, we set $l:=\argmin_i \psi^{i,b}(\boldsymbol{y})$ and then the following holds
    \begin{equation}
        \min_{i\neq j, k, l} \psi^{i,b}(\boldsymbol{y})\geq\psi^{j,b}(\boldsymbol{y})=\psi^{k,b}(\boldsymbol{y})>\psi^{l,b}(\boldsymbol{y});
        \label{region ii}
    \end{equation}
    
    \vspace{-0.1cm}

    \item[(iii)] there exists some index $k$ such that it touches the truncation level:
    \begin{equation*}
        \Bar{\psi}^{k,b}(\boldsymbol{y})+\delta^*(-y^k-\Bar{\psi}^{k,b}(\boldsymbol{y}))=\pm\xi.
    \end{equation*}
\end{itemize}

\kong

\noindent Region (i) refers to the area where $\Psi^b$ is not differentiable for the `non-minimizing' indices. Given any indices $j, k$ with property \eqref{region i}, first we can observe that
\begin{equation*}
    \rho^{j,b}(t,\boldsymbol{y}) = \rho^{k,b}(t,\boldsymbol{y})=b_t\,\Lambda(0).
\end{equation*}
Therefore, the $j$-th and $k$-th rows in the (generalized) Jacobian matrix are $N$-dimensional zero vectors. Since $\psi^{j,b}$, $\psi^{k,b}$ are implicit functions, it further implies
\begin{equation}
\begin{aligned}
    \psi^{j,b}(\boldsymbol{y})&=\psi^{k,b}(\boldsymbol{y})+\delta^*(-y^j-\psi^{k,b}(\boldsymbol{y})),\\
    \psi^{k,b}(\boldsymbol{y})&=\psi^{j,b}(\boldsymbol{y})+\delta^*(-y^k-\psi^{j,b}(\boldsymbol{y})),
\end{aligned}
\label{region_i_}
\end{equation}
for any $\boldsymbol{y}$ in region (i). Due to the fact that $(\delta^*)'\in(\epsilon, 1)$ for some $\epsilon>0$, we can deduce from \eqref{region_i_} that
\begin{equation}
    y^j=y^k \text{ \; and \; } \psi^{j,b}(\boldsymbol{y})=-y^j-\check{\delta}^*(0)=-y^k-\check{\delta}^*(0)=\psi^{k,b}(\boldsymbol{y}),
    \label{region_i__}
\end{equation}
for any $\boldsymbol{y}$ in region (i), where $\check{\delta}^*$ is the inverse function of $\delta^*$. Based on \eqref{region_i__}, the implicit functions of the others now possess explicit formulae as 
\begin{equation}
\begin{aligned}
    \psi^{l,b}(\boldsymbol{y})&=\psi^{j,b}(\boldsymbol{y})+\delta^*(-y^l-\psi^{j,b}(\boldsymbol{y}))\\
    &=-y^j-\check{\delta}^*(0)+\delta^*\big(-y^l+y^j+\check{\delta}^*(0)\big)=-y^k-\check{\delta}^*(0)+\delta^*\big(-y^l+y^k+\check{\delta}^*(0)\big)
\end{aligned}
\end{equation}
for $l\neq j, k$. It suffices to consider $y^j$ only. We infer that $\psi^{l,b}$ is decreasing in both $y^l$ and $y^k$. The partial derivative with respect to $y^k$ is no smaller than $-1$, together with \eqref{region_i__} implying the diagonal dominance of column entries. Bearing in mind that $j$ achieves the minimum, we then calculate the gradient of interactions
\begin{equation*}
    \nabla \big(\psi^{l,b}-\psi^{j,b}\big)(\boldsymbol{y})=\Big( \cdots, \underbrace{(\delta^*)'\big(-y^l+y^j+\check{\delta}^*(0)\big)}_{\text{$j$-th entry}}, \cdots, \underbrace{-(\delta^*)'\big(-y^l+y^j+\check{\delta}^*(0)\big)}_{\text{$l$-th entry}}, \cdots\Big).
\end{equation*}
The $l$-th entry of $\nabla \big(\psi^{l,b}-\psi^{j,b}\big)$ is negative and its `off-diagonal' entries are non-negative. Because $\Lambda$ is decreasing, we conclude that $\nabla_{\boldsymbol{y}}\rho^b(t,\boldsymbol{y})$ is an $M_0$-matrix in region (i). The case that more than two indices satisfy property \eqref{region i} can be discussed in an analogous manner.

Region (ii) refers to the area where $\Psi^b$ is not differentiable for the `minimizing’ indices. Since the overlapping at the minimum value has been discussed, we now assume that $l$ is the unique minimizing index and consider the property
\begin{equation}
        \min_{i\neq j, k, l} \psi^{i,b}(\boldsymbol{y})>\psi^{j,b}(\boldsymbol{y})=\psi^{k,b}(\boldsymbol{y})>\psi^{l,b}(\boldsymbol{y}).
        \label{region_ii_}
    \end{equation}
Since $\psi^{j,b}, \psi^{k,b}$ are implicit functions, we similarly have 
\begin{equation*}
\begin{aligned}
    \psi^{j,b}(\boldsymbol{y})&=\psi^{l,b}(\boldsymbol{y})+\delta^*(-y^j-\psi^{l,b}(\boldsymbol{y})),\\
    \psi^{k,b}(\boldsymbol{y})&=\psi^{l,b}(\boldsymbol{y})+\delta^*(-y^k-\psi^{l,b}(\boldsymbol{y})),
\end{aligned}
\end{equation*}
and subsequently $y^j=y^k$, for any $\boldsymbol{y}$ in region (ii). Let $\mathcal{I}$ be the index set defined as $\mathcal{I}=\{1, \dots, N\}\backslash k$. According to \eqref{region_ii_} and the continuity of implicit functions, there exists a neighborhood $\mathscr{B}\subseteq\mathbb{R}^{N-1}$ of $(y^i)_{i\in\mathcal{I}}$ such that
\begin{equation*}
        \min_{i\in\mathcal{I},\; i\neq j,l} \psi^{i,b}(\boldsymbol{y})>\psi^{j,b}(\boldsymbol{y})>\psi^{l,b}(\boldsymbol{y})
\end{equation*}
for any $(y^i)_{i\in\mathcal{I}}\in\mathscr{B}$. Notice that $(\Psi^{i,b})_{i\in\mathcal{I}}$ is differentiable in $\mathscr{B}$. While $(\psi^{i,b})_{i\in\mathcal{I}}$ is the implicit functions for $(\Psi^{i,b})_{i\in\mathcal{I}}$ in $\mathscr{B}$, the uniqueness implies that $(\psi^{i,b})_{i\in\mathcal{I}}$ satisfies the property stated in Lemma \ref{grad at smooth region}. After applying the same argument for $k$, we can deduce in region (ii) that functions $\psi^{j,b}$ and $\psi^{k,b}$ are equivalent. Consequently, the resulting Jacobian $\nabla_{\boldsymbol{y}}\rho^b(t,\boldsymbol{y})$ is still an $M_0$-matrix. The case that more than two indices satisfy property \eqref{region_ii_} can be again discussed analogously.

The analysis of region (iii) is already included in the proof of Lemma \ref{grad at smooth region}. We can then conclude that $\nabla_{\boldsymbol{y}}\rho^b(t,\boldsymbol{y})$ is an $M$-matrix whenever it is well-defined. Recall that
\begin{equation*}
        \partial_{\boldsymbol{y}}\rho^b(t,\boldsymbol{y})= \text{co }\big\{ \lim_{n\to\infty} \nabla_{\boldsymbol{y}}\rho^b(t,\boldsymbol{y}_n) \,:\, \boldsymbol{y_n}\to \boldsymbol{y}, \, \nabla_{\boldsymbol{y}}\rho^b(t,\boldsymbol{y}_n) \text{\, is well-defined} \big\}.
\end{equation*}
Although the convex combination of $M$-matrices may lead to a non-$M$-matrix, any matrix in $\partial_{\boldsymbol{y}}\rho^b(t,\boldsymbol{y})$
must be a $Z_+$-matrix, considering any $M$-matrix has non-negative diagonal elements and non-positive off-diagonal elements. A similar discussion infers that any matrix in $\partial_{\boldsymbol{y}}\rho^a(t,\boldsymbol{y})$ is a $Z_+$-matrix. The statement that any matrix in $\partial_{\boldsymbol{y}}\rho(t,\boldsymbol{y})$ is an $Z_+$-matrix then follows from the sum rule in Theorem \ref{sum rule}.
\end{proof}

\vspace{0.2cm}

\section{General Game: Ordering Property and Decoupling Field}
\label{paper 2 section 5}
\noindent We will consider the homogeneous case for the global well-posedness. Similar to the linear scenario, the ordering property will be derived first, with an immediate application in removing the constraint $\xi$. Moreover, the property simplifies the analysis from $Z_+$-matrices to $M_0$-matrices. Afterward, we will present the connection between the FBSDE and the characteristic BSDE. Let us begin with the following assumption:

\kong

\begin{assumption}
    (1) The agents are homogeneous in penalty coefficients. Specifically, there are bounded $\phi:=(\phi_t)_{t\in[0,T]}\in\mathbb{H}^2$ and $A\in L^2(\Omega, \mathcal{F}_T)$, such that $\phi^i=\phi$ and $A^i=A$ for all $i$.
    
    (2) The index of the agent indicates the rank of her initial inventory as follows:
    \begin{equation*}
        q_0^1\leq q_0^2 \leq \cdots \leq q_0^{N-1} \leq q_0^{N}.
    \end{equation*}
\end{assumption}

\kong

\noindent In analogy to the linear case, the ordering property under the general intensity function is presented in the following lemma, utilizing matrix properties in Theorem \ref{global_Jacobian}.

\kong

\begin{lemma}
\label{ordering in general}
Suppose the FBSDE \eqref{general FBSDE} has a solution on $[0,T]$, then the corresponding equilibrium satisfies
\begin{equation*}
\begin{aligned}
    \psi^{1,b}(\boldsymbol{Y}_t)\leq\psi^{2,b}(\boldsymbol{Y}_t)\leq &\cdots \leq \psi^{N-1,b}(\boldsymbol{Y}_t)\leq \psi^{N,b}(\boldsymbol{Y}_t),\\
    \psi^{1,a}(\boldsymbol{Y}_t)\geq\psi^{2,a}(\boldsymbol{Y}_t)\geq &\cdots \geq \psi^{N-1,a}(\boldsymbol{Y}_t)\geq \psi^{N,a}(\boldsymbol{Y}_t)
\end{aligned}
\end{equation*}
almost surely for all $t\in[0,T]$.
\end{lemma}

\begin{proof}
Denote by $(\boldsymbol{Q}, \boldsymbol{Y}, \boldsymbol{M})$ the solution of the FBSDE \eqref{general FBSDE}. For any $i\in\{1,\dots, N\}$, we first define the function $\varrho^i$ as
\begin{equation*}
\begin{aligned}
    b_t\,\Lambda\big(\psi^{i,b}(\boldsymbol{y})-&\bar{\psi}^{i,b}(\boldsymbol{y})\big) - a_t\,\Lambda\big(\psi^{i,a}(\boldsymbol{y})-\bar{\psi}^{i,a}(\boldsymbol{y})\big)\\
    &=b_t\,\Lambda\big(\psi^{i,b}(\dots, \underbrace{y^i-y^j+y^j}_{\text{i-th entry}}, \dots)-\bar{\psi}^{i,b}(\dots, y^i-y^j+y^j, \dots)\big)\\
    &\hspace{2cm} -a_t\,\Lambda\big(\psi^{i,a}(\dots, y^i-y^j+y^j, \dots)-\bar{\psi}^{i,a}(\dots, y^i-y^j+y^j, \dots)\big)\\
    &=\varrho^i(t,y^i-y^j),
\end{aligned}
\end{equation*}
where function $\varrho^i$ is defined as
\begin{equation*}
\begin{aligned}
    \varrho^i(t, x) &= b_t\,\Lambda\big(\psi^{i,b}(\dots, \underbrace{x+y^j}_{\text{i-th entry}}, \dots)-\bar{\psi}^{i,b}(\dots, x + y^j, \dots)\big)\\
    &\hspace{2cm} - a_t\,\Lambda\big(\psi^{i,a}(\dots, x+y^j, \dots)-\bar{\psi}^{i,a}(\dots, x+y^j, \dots)\big)
\end{aligned}
\end{equation*}
for a given vector $\boldsymbol{y}$. According to Theorem \ref{global_Jacobian}, function $\varrho^i$ is non-decreasing with respect to the second variable. We pick $j \geq i$ as some other index and define function $\varrho^j$ through the same trick on the $j$-th entry. Note that $\varrho^j$ is also non-decreasing in the space variable. If one further set $(\Delta Q, \Delta
 Y, \Delta M):=(Q^i-Q^j, Y^i-Y^j, M^i-M^j)$, the following can then be obtained:
\begin{equation*}
    d\Delta Q_t = \big( \varrho^i(t,Y_t^i-Y_t^j) - \varrho^j(t, -Y_t^i+Y_t^j) \big) \, dt= \varrho(t,Y_t^i-Y_t^j) \, dt.
\end{equation*}
Here, function $\varrho$ is defined as $\varrho(t,x) = \varrho^i(t, x) - \varrho^j(t, -x)$ and it is non-decreasing with respect to $x$. Think of the case when $\Delta Y_t = Y_t^i - Y_t^j = 0$. Since $\psi^{i,b}$, $\psi^{j,b}$ are implicit functions, it holds that
\begin{equation}
\label{implicit in ordering}
\begin{aligned}
    \psi^{i,b}(\boldsymbol{y})&= \big[ \bar{\psi}^{i,b}(\boldsymbol{y})+\delta^*(-y^i-\bar{\psi}^{i,b}(\boldsymbol{y})) \big] \vee(-\xi)\wedge\xi, \\
    \psi^{j,b}(\boldsymbol{y})&= \big[ \bar{\psi}^{j,b}(\boldsymbol{y})+\delta^*(-y^j-\bar{\psi}^{j,b}(\boldsymbol{y})) \big] \vee(-\xi)\wedge\xi.
\end{aligned}
\end{equation}
Given $y^i=y^j$, let us assume that $\psi^{i,b}(\boldsymbol{y})>\psi^{j,b}(\boldsymbol{y})$. The definition first yields 
\begin{equation*}
    \bar{\psi}^{i,b}(\boldsymbol{y})= \psi^{j,b}(\boldsymbol{y}) \wedge \min_{k\neq i, j} \psi^{k,b}(\boldsymbol{y}) \leq \psi^{i,b}(\boldsymbol{y}) \wedge \min_{k\neq i, j} \psi^{k,b}(\boldsymbol{y}) = \bar{\psi}^{j,b}(\boldsymbol{y}).
\end{equation*}
Since $(\delta^*)'\in(0,1)$, the function defined by
\begin{equation*}
    z \mapsto z + \delta^*(-y^i-z)
\end{equation*}
is increasing. Consequently, this leads to a contradiction that
\begin{equation*}
    \bar{\psi}^{i,b}(\boldsymbol{y})+\delta^*(-y^i-\bar{\psi}^{i,b}(\boldsymbol{y})) \leq \bar{\psi}^{j,b}(\boldsymbol{y})+\delta^*(-y^j-\bar{\psi}^{j,b}(\boldsymbol{y})).
\end{equation*}
The contradiction when $\psi^{i,b}(\boldsymbol{y})<\psi^{j,b}(\boldsymbol{y})$ can be found symmetrically. Since the relation for $\psi^{i,a}$ and $\psi^{j,a}$ can be analyzed similarly, it is now evident that $\Delta Y_t=Y_t^i-Y_t^j=0$ yields $\psi^{i,b}(\boldsymbol{Y}_t)=\psi^{j,b}(\boldsymbol{Y}_t)$, $\psi^{i,a}(\boldsymbol{Y}_t)=\psi^{j,a}(\boldsymbol{Y}_t)$, and thus $\varrho(t,0)=0$. Literally, agent $i$ and $j$ execute the same control. Consequently, we can write
\begin{equation*}
    d\Delta Q_t = \varrho(t,\Delta Y_t)\,dt = \big(\varrho(t,\Delta Y_t) - \varrho(t,0)\big)\,dt=\iota_t\,\Delta Y_t \, dt,
\end{equation*}
where $\iota\in\mathbb{H}^2$ is non-negative and bounded because of the Lipschitz property of $\varrho$. One can now see $(\Delta Q, \Delta
 Y, \Delta M)$ solves the FBSDE
\begin{equation*}
    \left\{
    \begin{aligned}
     \, d\Delta Q_t &= \iota_t\,\Delta Y_t \, dt,\\
     d\Delta Y_t &= 2\phi_t\,\Delta Q_t\,dt+d\Delta M_t,\\
    \Delta Q_0 &= q_0^i-q_0^j, \quad \Delta Y_T = -2A\,\Delta Q_T.
    \end{aligned}
    \right.
\end{equation*}
In accordance with \cite{guo2023macroscopic}, not only the above FBSDE is well-posed, but also we know $\Delta Q$ always has the same sign with $q_0^i-q_0^j$ and $\Delta Y$ has a different sign with $\Delta Q$. Hence, when $q_0^i-q_0^j\leq0$, it follows that $\Delta Q_t\leq 0$ and $\Delta Y_t\geq 0$.

It suffices to study how $\Delta y$ affects the controls of agent $i$ and $j$. For convenience, here we write $ \psi^{b}(y^i, y^j) = \psi^{b}(\dots, y^i, \dots, y^j, \dots)$. By the mean value theorem \ref{mean value and calculus}, we can have
\begin{equation*}
\begin{pmatrix}
        \, \psi^{i,b}(y^i, y^j) \, \\
        \, \psi^{j,b}(y^i, y^j) \,
\end{pmatrix} - 
\begin{pmatrix}
        \, \psi^{i,b}(y^i, y^i) \, \\
        \, \psi^{j,b}(y^i, y^i) \,
\end{pmatrix} = \boldsymbol{K} \,
\begin{pmatrix}
        \, 0 \, \\
        \, -\Delta y \,
\end{pmatrix} = 
\begin{pmatrix}
    k_{11} & k_{12} \\
    k_{21} & k_{22}
\end{pmatrix} \,
\begin{pmatrix}
        \, 0 \, \\
        \, -\Delta y \,
\end{pmatrix}
\end{equation*}
where $\boldsymbol{K}\in\mathbb{R}^{2 \times 2}$ is the convex combination of matrices with properties described in Theorem \ref{global_Jacobian}. Therefore, matrix $\boldsymbol{K}$ has non-positive entries and is diagonally dominant of column entries. We can further deduce 
\begin{equation*}
\begin{aligned}
    \psi^{i,b}(y^i, y^j) - \psi^{i,b}(y^i, y^i) - \psi^{j,b}(y^i, y^j) + \psi^{j,b}(y^i, y^i) &= - k_{12} \, \Delta y + k_{22} \, \Delta y\\
    \psi^{i,b}(y^i, y^j) - \psi^{j,b}(y^i, y^j) &= (k_{22} - k_{12}) \, \Delta y.
\end{aligned}
\end{equation*}
Given $\Delta y \geq 0$, the fact that $k_{22} \leq k_{12} \leq 0$ yields $\psi^{i,b}(y^i, y^j) - \psi^{j,b}(y^i, y^j) \leq 0$. The ask side can be discussed similarly.
\end{proof}

\kong

\noindent As the first application of the ordering property, we use it in the following result to remove the regularization term $\xi$ in the $\mathbb{A}$. This is achieved by establishing a bound for the solution that is independent of $\xi$. 

\kong

\begin{proposition}
\label{remove xi}
Suppose that the FBSDE \eqref{general FBSDE} has a solution $(\boldsymbol{Q}, \boldsymbol{Y}, \boldsymbol{M})$ on $[0,T]$ and $\xi \geq |\delta^*(0)|$, then it holds almost surely that $|\boldsymbol{Y}_t^i| \leq C$ for any $i$ and $t$, where $C>0$ is some constant independent of $\xi$.
\end{proposition}

\begin{proof}
Given the ordering property in Lemma \ref{ordering in general}, the forward equations for $i\in\{2, \dots, N-1\}$ now becomes
\begin{equation*}
\begin{aligned}
    dQ_t^i&=b_t \, \Lambda(\psi^{i,b}(\boldsymbol{Y}_t)-\bar{\psi}^{i,b}(\boldsymbol{Y}_t))\,dt-a_t \, \Lambda(\psi^{i,a}(\boldsymbol{Y}_t)-\bar{\psi}^{i,a}(\boldsymbol{Y}_t))\,dt\\
    &=b_t \, \Lambda(\psi^{i,b}(\boldsymbol{Y}_t)-\psi^{1,b}(\boldsymbol{Y}_t))\,dt - a_t \, \Lambda(\psi^{i,a}(\boldsymbol{Y}_t)-\psi^{N,a}(\boldsymbol{Y}_t))\,dt.
\end{aligned}
\end{equation*}
Since $\psi^{i,b}(\boldsymbol{Y}_t)\geq\psi^{1,b}(\boldsymbol{Y}_t)$ and $\psi^{i,a}(\boldsymbol{Y}_t)\geq\psi^{N,a}(\boldsymbol{Y}_t)$, it infers $|Q_t^i|\leq |q_0^i|+(\bar{a}+\bar{b})\,\Lambda(0)\,T$ 
for all $t\in[0,T]$. Due to the fact that
\begin{equation}
    Y_t^i=\mathbb{E}_t\Big[-2A\,Q_T^i-2\int_t^T \phi_s\,Q_s^i\,ds\Big],
    \label{relation of backward variable}
\end{equation}
we can then conclude $(Y^i)_{i\in\{2, \dots, N-1\}}$ are uniformly bounded by some constant independent of the constraint $\xi$. For agent $1$, we first observe
\begin{equation*}
    Q_t^1\geq q_0^1-\int_0^t a_s\,\Lambda(\psi^{1,a}(\boldsymbol{Y}_s)-\psi^{N,a}(\boldsymbol{Y}_s))\,ds \geq q_0^1-\Bar{a}\,\Lambda(0)\,T.
\end{equation*}
By the expression in \eqref{relation of backward variable}, a lower bound of $Q^1$ provides an upper bound for $Y^1$ that is also independent of $\xi$. One can obtain a uniform lower bound for $Y^N$ from a symmetric argument. Via the $Z_+$-matrix property in Theorem \ref{global_Jacobian}, the function
\begin{equation}
    \psi^{1,b}(Y_t^1, Y_t^2, \dots, Y_t^{N-1}, Y_t^N)-\psi^{2,b}(Y_t^1, Y_t^2, \dots, Y_t^{N-1}, Y_t^N)
    \label{temp expr in non lip}
\end{equation}
is non-increasing with respect to $Y_t^1$ and non-decreasing in $\{Y_t^2, \dots, Y_t^N\}$, the generalized derivatives of which are all bounded element-wise by some constant independent of $\xi$. In addition, considering the following facts:
\kong

\begin{itemize}
    \item[(1)] $Y^1$ is upper bounded and $Y^N$ is lower bounded;\\
    \vspace{-0.2cm}
    
    \item[(2)] $(Y^j)_{j\in\{2, \dots, N-1\}}$ are bounded;\\ 
    \vspace{-0.2cm}
    
    \item[(3)] all above bounds are independent of $\xi$;\\ 
    \vspace{-0.2cm}

    \item[(4)] $ \psi^{1,b}(0, 0, \dots, 0)-\psi^{2,b}(0, 0, \dots, 0) = 0$,
\end{itemize}
\kong
\noindent we can derive a lower bound for expression \eqref{temp expr in non lip} and it leads to an upper bound for $Q^1$, both of which are again independent of $\xi$. It follows $Y^1$ are bounded by some constant independent of $\xi$; the same is true for $Y^N$ through a similar argument. On the other hand, it can be directly checked that $\psi^{b}(\boldsymbol{0})=\boldsymbol{0}$ provided $\xi \geq |\delta^*(0)|$. Together with the boundedness of $\boldsymbol{Y}$, the Lipschitz continuity of $\psi^b$ yields $\psi^{i,b}(\boldsymbol{Y}_t)\leq C$ for all $i$ and $t$, where $C$ is a constant independent of $\xi$.
\end{proof}
\kong

\noindent Hence, we force $\xi$ to be large enough.

\kong
\begin{assumption}
    The constraint $\xi$ is chosen to be larger than the constant specified in Theorem \ref{remove xi}. As a result, the truncation has no effect. 
\end{assumption}

\begin{remark}
    The absence of the truncation $\xi$ renders the FBSDE  \eqref{general FBSDE} non-Lipschitz. We will still regard \eqref{general FBSDE} as a Lipschitz FBSDE due to the boundedness of its solution. From the perspective of the non-Lipschitz FBSDE, what we will find is the unique bounded solution.
\end{remark}

\kong

\noindent The following statement is then an immediate consequence of Theorem \ref{global_Jacobian}.
\kong

\begin{theorem}
\label{global_Jacobian_M0}
    For any $\boldsymbol{y}\in \mathbb{R}^N$, any matrix in $\partial_{\boldsymbol{y}}\rho(t,\boldsymbol{y})$ is an $M_0$-matrix.
\end{theorem}
\begin{proof}
When $\xi$ has no effect, matrices $\nabla_{\boldsymbol{y}}\rho^b(t,\boldsymbol{y})$ and $\nabla_{\boldsymbol{y}}\rho^a(t,\boldsymbol{y})$ are always $M_0$-matrix whenever differentiable, according to the proof of Theorem \ref{global_Jacobian}. It suffices to observe that $M_0$-matrices are stable under convex combinations.
\end{proof}

\kong

Thanks to the Lipschitz property of the FBSDE \eqref{general FBSDE}, it is well-known that there exists $\Delta>0$ being small enough, such that \eqref{general FBSDE} is well-posed on the horizon $[T-\Delta, T]$ via the contraction mapping principle; i.e., the original initial time $0$ is replaced by $T-\Delta$. Given any $t\in[T-\Delta, T]$ and $\boldsymbol{q}\in\mathbb{R}^N$ as the initial time and condition, there exists a unique solution $(\boldsymbol{Q}^{t,\boldsymbol{q}}, \boldsymbol{Y}^{t,\boldsymbol{q}}, \boldsymbol{M}^{t,\boldsymbol{q}})$ to \eqref{general FBSDE}, where the superscript denotes the dependence on the initial data. We can then define a function $u:[T-\Delta, T]\times \Omega \times \mathbb{R}^N\to\mathbb{R}^N$ through
\begin{equation}
    u(t, \boldsymbol{q}):=\boldsymbol{Y}^{t,\boldsymbol{q}},
    \label{local_contract}
\end{equation}
which is known as the decoupling field. 

\kong

\begin{definition}
    Let $t \in [0, T]$. A function $u: [t, T] \times \Omega  \times \mathbb{R}^N \to \mathbb{R}^N$, with $u(T,\boldsymbol{q})=-2A\,\boldsymbol{q}$ a.e., is called a \textit{decoupling field} for the FBSDE on $[t, T]$ if, for all $t_1, t_2 \in [t, T]$ with $t_1 < t_2$ and
    any $\mathcal{F}_{t_1}$-measurable $\boldsymbol{\eta}: \Omega\to \mathbb{R}^N$, there exist progressively measurable processes $(\boldsymbol{Q},\boldsymbol{Y},\boldsymbol{Z})$ on $[t_1,t_2]$ such that
    \begin{equation*}
        \begin{aligned}
            \boldsymbol{Q}_s &= \boldsymbol{\eta} + \int_{t_1}^s \rho(r, \boldsymbol{Y}_r)\,dr,\\
            \boldsymbol{Y}_s &= \boldsymbol{Y}_{t_2} - \int_s^{t_2} 2\phi_r\,\boldsymbol{Q}_r\,dr -\int_s^{t_2} \boldsymbol{Z}_r\,dW_r,\\
            \boldsymbol{Y}_s &= u(s,\boldsymbol{Q}_s),
        \end{aligned}
    \end{equation*}
    for all $s \in [t_1, t_2]$. In particular, we want all integrals to be well defined.
\end{definition}

\kong

\noindent The theory of decoupling fields, originally introduced by \cite{ma2015well} for one-dimensional equations, has been extended to multi-dimensional equations through subsequent works such as \cite{fromm2013existence}, \cite{fromm2015theory}, and \cite{ankirchner2020optimal}. The fundamental idea is that, if the decoupling field can be regularly extended over the whole prescribed time horizon, it then ensures the well-posedness of the FBSDE on that horizon. Some essential results are provided below, but we refer the reader to \cite{ankirchner2020optimal} for an excellent short summary.

\kong

\begin{definition}[\cite{ankirchner2020optimal}]
Denote by $\mathscr{L}_{u(s,\cdot)}$ the Lipschitz coefficient of $u$ with respect to the space variable at time $s$:
\begin{equation*}
    \mathscr{L}_{u(s,\cdot)}:=\inf\Big\{L>0 \, : \, |u(s,\boldsymbol{q}')-u(s, \boldsymbol{q})|\leq L\,|\boldsymbol{q}'-\boldsymbol{q}| \, \text{ almost surely for all } \boldsymbol{q}', \boldsymbol{q}\in\mathbb{R}^N\Big\}.
\end{equation*}
A decoupling field $u: [t, T] \times \Omega \times \mathbb{R}^N \to \mathbb{R}^N$ is called \textit{(weakly) regular} if
\begin{equation*}
    \sup_{s\in[t, T]}\mathscr{L}_{u(s,\cdot)}<\infty \text{ \; and \; } \sup_{s\in[t, T]} \|u(s,\cdot,0)\|_\infty < \infty,
\end{equation*}
where $\|\cdot\|_\infty$ denotes the $L^\infty$-norm of random variables.
\end{definition}

\kong

\noindent The decoupling field \eqref{local_contract} constructed by the contraction mapping principle is indeed regular. Therefore, one can apply the fixed point method again, extending the definition of the decoupling field to a longer horizon. The first part of the following theorem provides a generalized result of this constructive procedure, while the second part reveals the connection between the decoupling field and the well-posedness of FBSDEs.

\kong

\begin{theorem}[\cite{ankirchner2020optimal}]
\label{local_time}
(1) There exists a time $t \in [0, T]$ such that the FBSDE has a unique (up to modification) decoupling
field $u$ on $[t, T]$ that is also regular.

(2) If there exists a regular decoupling field $u$ of the corresponding FBSDE on some interval $[t, T]$, then for any initial condition $\boldsymbol{Q}_t =\eta \in \mathbb{R}^N$ there is a unique solution $(\boldsymbol{Q},\boldsymbol{Y},\boldsymbol{Z})$
of the FBSDE on $[t, T]$ satisfying
\begin{equation*}
    \sup_{s\in[t,T]}\mathbb{E}\big[|\boldsymbol{Q}_s|^2\big] + \sup_{s\in[t,T]}\mathbb{E}\big[|\boldsymbol{Y}_s|^2\big]+ \mathbb{E}\Big[\int_t^T|\boldsymbol{Z}_s|^2\,ds\Big]<\infty.
\end{equation*}
\end{theorem}

\kong

\noindent Since the extension of the decoupling field is a key consideration, several important questions arise: (1) How far can the decoupling field be extended? (2) What are the implications if it cannot be extended further? Here are some answers to these questions:

\kong

\begin{theorem}[\cite{ankirchner2020optimal}]
    Define the maximal interval $I_{\textnormal{max}} \subseteq [0, T]$ of the FBSDE as the union of all intervals $[t, T] \subseteq [0, T]$, such that there exists a regular decoupling field $u$ on $[t, T]$. Then, there exists a unique regular decoupling field $u$ on $I_{\textnormal{max}}$. Furthermore, either $I_{\textnormal{max}} = [0, T]$ or $I_{\textnormal{max}} = (t_{\textnormal{min}}, T]$ with $t_{\textnormal{min}} \in [0, T)$. In the latter case, we have $\lim_{t\searrow t_{\textnormal{min}}}
    \mathscr{L}_{u(t,\cdot)} = \infty$.
\end{theorem}

\kong

\noindent The regularity of the decoupling field, which ensures the well-posedness of FBSDEs, is studied in the latter theorem through the analysis of the corresponding characteristic BSDE. The ideas of variational FBSDEs and characteristic BSDEs are initially introduced for one-dimensional FBSDEs in \cite{ma2015well}. While \cite{hu2022path} presents a one-dimensional characteristic BSDE for multi-dimensional equations, the construction method suffers from the loss of critical information from the original FBSDE. The resulting BSDE is hence difficult to analyze. To overcome this challenge, we introduce a multi-dimensional characteristic BSDE that preserves more information from the original FBSDE, albeit at the expense of increased dimension. A mean value theorem in non-smooth analysis is presented below.

\kong

\begin{theorem}[\cite{clarke1990optimization}]
\label{mean value and calculus}
Let $F$ be Lipschitz on an open convex set $U$ in $\mathbb{R}^n$, and let $x$ and $y$ be two points in $U$. Then, it holds that
\begin{equation*}
F(y) - F(x) \in \textnormal{co} \Big\{ k \cdot (y - x) \,|\, k \in \partial F(z) \textnormal{\, and \,} z \textnormal{\, is in the line segment between $x$ and $y$ } \Big \}.\\
\end{equation*}
\end{theorem}

\kong

\noindent Such mean value theorem is utilized in the proof of the following theorem, which establishes a connection between the well-posedness of the FBSDE and that of the backward stochastic Riccati equation (BSRE).

\kong

\begin{theorem}
\label{mult_char_bsre}
Consider the BSRE of the following type:
\begin{equation}
    d\mathcal{X}_t= (2\phi_t\, I - \mathcal{X}_t \, \mathcal{B}_t \, \mathcal{X}_t)\,dt+\mathcal{Z}_t\,dW_t, \quad \mathcal{X}_T=-2A\,I.
    \label{char_BSRE}
\end{equation}
Note that $\mathcal{X}_t$ and $\mathcal{B}_t$ are $N \times N$ matrices, and $I$ represents the identity matrix. Additionally, the process $\mathcal{B}$ satisfies the following: (1) each entry is a bounded process in $\mathbb{H}^2$, and (2) the matrix $\mathcal{B}_t$ is of $M_0$-type for every $t$.

Suppose the BSRE \eqref{char_BSRE} accepts a unique solution $\mathcal{X}$ such that each entry is a bounded process in $\mathbb{H}^2$, then the FBSDE \eqref{general FBSDE} has a unique solution $(\boldsymbol{Q}, \boldsymbol{Y}, \boldsymbol{M})$ in $(\mathbb{S}^2 \times \mathbb{H}^2 \times \mathbb{M})^N$.
\end{theorem}

\begin{proof}
Based the Lipschitz nature of the FBSDE, there exists some $s\in(0,T)$ such that: (1) FBSDE \eqref{general FBSDE} has a unique regular decoupling
field $u$ on the horizon $[s, T]$; (2) it also accepts a unique solution on this horizon with any deterministic initial condition. Fixing any $\boldsymbol{\eta}, \tilde{\boldsymbol{\eta}}\in\mathbb{R}^N$ as two initial data, let us denote by $(\boldsymbol{Q}, \boldsymbol{Y}, \boldsymbol{M})$ and $(\tilde{\boldsymbol{Q}}, \tilde{\boldsymbol{Y}}, \tilde{\boldsymbol{M}})$ the two corresponding solutions. It then holds:
\begin{equation*}
        \begin{aligned}
            \boldsymbol{Q}_t &= \boldsymbol{\eta} + \int_{s}^t \rho(r, \boldsymbol{Y}_r)\,dr,\\
            \boldsymbol{Y}_t &= -2A\,\boldsymbol{Q}_T - \int_t^{T} 2\phi_r\,\boldsymbol{Q}_r\,dr -\int_t^{T} d\boldsymbol{M}_r,\\
            \boldsymbol{Y}_t &= u(t, \boldsymbol{Q}_t)
        \end{aligned}
    \end{equation*}
for any $t\in[s,T]$; the case of $(\tilde{\boldsymbol{Q}}, \tilde{\boldsymbol{Y}}, \tilde{\boldsymbol{M}})$ is similar. If we define $(\boldsymbol{\mathscr{Q}}, \boldsymbol{\mathscr{Y}}, \boldsymbol{\mathscr{M}}):=(\tilde{\boldsymbol{Q}}-\boldsymbol{Q}, \tilde{\boldsymbol{Y}}-\boldsymbol{Y}, \tilde{\boldsymbol{M}}-\boldsymbol{M})$, by taking the difference of two FBSDEs, it is straightforward to see that the backward equation becomes
\begin{equation*}
    \boldsymbol{\mathscr{Y}}_t = -2A\,\boldsymbol{\mathscr{Q}}_T - \int_t^{T} 2\phi_r\,\boldsymbol{\mathscr{Q}}_r\,dr -\int_t^{T} d\boldsymbol{\mathscr{M}}_r.
\end{equation*}
According to Theorem \ref{mean value and calculus}, the forward equation can be obtained by
\begin{equation*}
    d\boldsymbol{\mathscr{Q}}_t = \big[ \rho(t, \tilde{\boldsymbol{Y}}_t)-\rho(t, \boldsymbol{Y}_t) \big]\, dt = \mathcal{B}_t \, \boldsymbol{\mathscr{Y}}_t \, dt,
\end{equation*}
where $\mathcal{B}:=(\mathcal{B}_t)_{t\in[0,T]}$ is an $M_0$-matrix for any $t$ by Theorem \ref{global_Jacobian_M0}. Indeed, $M_0$-matrices are closed under convex combinations. It is also not hard to see that $\mathcal{B}$ is continuous with respect to time and element-wise bounded. We can then conclude that $(\boldsymbol{\mathscr{Q}}, \boldsymbol{\mathscr{Y}}, \boldsymbol{\mathscr{M}})$ solves the FBSDE
\begin{equation}
    \left\{
    \begin{aligned}
     d\boldsymbol{\mathscr{Q}}_t &= \mathcal{B}_t \, \boldsymbol{\mathscr{Y}}_t \, dt,\\
    \, d\boldsymbol{\mathscr{Y}}_t &= 2\phi_t \,  \boldsymbol{\mathscr{Q}}_t \, dt + d\boldsymbol{\mathscr{M}}_t,\\
    \boldsymbol{\mathscr{Q}}_s &= \Tilde{\boldsymbol{\eta}} - \boldsymbol{\eta}, \quad \boldsymbol{\mathscr{Y}}_T = -2A\,\boldsymbol{\mathscr{Q}}_T.
    \end{aligned}
    \right.
    \label{variat FBSDE}
\end{equation}
Equation \eqref{variat FBSDE} is the multi-dimensional version of the \textit{variational FBSDE} in \cite{ma2015well}.

To examine the variational FBSDE, its linear structure suggests the affine ansatz
\begin{equation}
    \boldsymbol{\mathscr{Y}}_t= \upsilon(t,\boldsymbol{\mathscr{Q}}_t):=\mathcal{X}_t\,\boldsymbol{\mathscr{Q}}_t
    \label{linear ansatz}
\end{equation}
for some matrix-valued process $\mathcal{X}$ to be specified. Via matching the coefficients in
\begin{equation*}
    d\boldsymbol{\mathscr{Y}}_t = \big(d\mathcal{X}_t+\mathcal{X}_t \, \mathcal{B}_t \, \mathcal{X}_t\,dt\big)\,\boldsymbol{\mathscr{Q}}_t = 2\phi_t \,
 \boldsymbol{\mathscr{Q}}_t\,dt + d\boldsymbol{\mathscr{M}}_t,
\end{equation*}
it turns out that $\mathcal{X}$ solves the BSRE
\begin{equation}
     d\mathcal{X}_t = (2\phi_t\, I - \mathcal{X}_t \, \mathcal{B}_t \, \mathcal{X}_t)\,dt + \mathcal{Z}_t \, dW_t, \quad \mathcal{X}_T=-2A\,I.
     \label{characteristic BSDE}
\end{equation}
Equation \eqref{characteristic BSDE} is the multi-dimensional generalization of the \textit{characteristic BSDE} in \cite{ma2015well}. The existence of a unique bounded $\mathcal{X}$ is ensured by the assumption. Upon examining the definition, we know the function $\upsilon$ defined in \eqref{linear ansatz} serves as the unique regular decoupling field for the FBSDE \eqref{variat FBSDE} due to the boundedness of $\mathcal{X}$. Here, we only need to focus on the bounded solution $\mathcal{X}$. Indeed, the unboundedness will render the decoupling field non-Lipschitz, contradicting the existence of a regular decoupling field on $[s, T]$. Given the variational FBSDE \eqref{variat FBSDE}, we can infer
\begin{equation}
    |u(s,\tilde{\eta}) - u(s,\eta)| = |\tilde{\boldsymbol{Y}}_s-\boldsymbol{Y}_s| = |\boldsymbol{\mathscr{Y}}_s| = |\mathcal{X}_s\,\boldsymbol{\mathscr{Q}}_s| \leq  \|\mathcal{X}_s\|_2 \cdot |\tilde{\eta}-\eta|,
    \label{lip_estimate}
\end{equation}
where $\|\mathcal{X}_s\|_2$ denotes the spectral norm of matrices. Since the BSRE is well-posed on the entire horizon and the solution $\mathcal{X}$ is bounded, the estimate \eqref{lip_estimate} infers that there is no $t_{\min}\geq 0$, such that $\lim_{s\searrow t_{\textnormal{min}}}
    \mathscr{L}_{u(s,\cdot)} = \infty$. Consequently, the decoupling field $u$ can be extended to the entire horizon $[0,T]$ and the FBSDE is then globally well-posed.
\end{proof}

\vspace{0.2cm}

\section{General Game: Well-posedness and Properties}
\label{paper 2 section 6}
\noindent Revisiting the linear case, it becomes evident that the ordering property plays a pivotal role in simplifying the $N$-player game down to a four-player scenario. This quartet comprises agents numbered $1, 2, N-1$, and $N$, specifically those occupying the highest and lowest inventory levels. We can even condense this four-player dimension into just two when the competition is \textit{complete at the boundary}, that is to say,
    \begin{equation}
    \psi^{1,b}(\boldsymbol{Y}_t) = \inf_{i \neq 1}\psi^{i,b}(\boldsymbol{Y}_t) \text{ \; and \; } \psi^{N,a}(\boldsymbol{Y}_t) = \inf_{i\neq N}\psi^{i,a}(\boldsymbol{Y}_t),
    \label{perfect compet}
    \end{equation}
for all $t$. For a finite number of players, case \eqref{perfect compet} happens if and only if $q_0^1=q_0^2$ and $q_0^{N-1}=q_0^N$ because
    \begin{equation*}
    \psi^{1,b}(\boldsymbol{Y}_t) = \psi^{2,b}(\boldsymbol{Y}_t) = \min_{i}\psi^{i,b}(\boldsymbol{Y}_t) \text{ \; and \; } \psi^{N,a}(\boldsymbol{Y}_t) = \psi^{N-1,a}(\boldsymbol{Y}_t) = \min_{i}\psi^{i,a}(\boldsymbol{Y}_t),
    \end{equation*}
as a direct consequence of Lemma \eqref{ordering in general}. In this scenario, it suffices to solve the two-dimensional BSRE \eqref{char_BSRE} since agent $2$ (resp. $N-1$) is a `copy' of agent $1$ (resp. $N$). Utilizing this finite-player characterization, we investigate another scenario in an infinite-player setting, where identical players are no longer necessary.

\kong

\begin{proposition}
\label{paper 2 mean field game}
Assume the BSRE \eqref{char_BSRE} is well-posed when $N=2$. Consider an infinite number of players indexed by $\mathcal{I}$ with initial inventories satisfying $\sup_{i\in\mathcal{I}}|q_0^i|<\infty$. Suppose that there exist no-duplicate sequences $(k_n)_{n\in\mathbb{N}}, (l_n)_{n\in\mathbb{N}}\subseteq \mathcal{I}$ such that
    \begin{equation*}
        \lim_{n\to\infty}q_0^{k_n}=\inf_{i\in\mathcal{I}} q_0^i \text{\quad and \quad} \lim_{n\to\infty} q_0^{l_n}=\sup_{i\in\mathcal{I}} q_0^i,
    \end{equation*}
    then \eqref{perfect compet} holds and there exists a Nash equilibrium.
\end{proposition}
\begin{proof}
Let $\tilde{1}, \tilde{2}, \Tilde{3}$ and $\Tilde{4}$ be four artificial players with initial inventories given by
\begin{equation*}
    q_0^{\tilde{1}} = q_0^{\tilde{2}} = \inf_{i\in\mathcal{I}} q_0^i \text{ \; and \; } q_0^{\tilde{3}} = q_0^{\tilde{4}} = \sup_{i\in\mathcal{I}} q_0^i.
\end{equation*}
We look at the four-player game of $(\tilde{1}, \tilde{2}, \Tilde{3}, \Tilde{4})$, the equilibrium of which exists by the well-posedness of two-dimensional BSRE \eqref{char_BSRE} discussed above. Denote by $(\boldsymbol{\beta}^{\Tilde{1}}, \boldsymbol{\beta}^{\Tilde{2}}, \boldsymbol{\beta}^{\Tilde{3}}, \boldsymbol{\beta}^{\Tilde{4}})$ the equilibrium profile and it holds for all $t$ that
\begin{equation*}
\begin{aligned}
    \beta_t^{\Tilde{1}, b} &= \beta_t^{\Tilde{2}, b} \leq \beta_t^{\Tilde{3}, b} = \beta_t^{\Tilde{4}, b},\\
    \beta_t^{\Tilde{4}, a} &= \beta_t^{\Tilde{3}, a} \leq \beta_t^{\Tilde{2}, a} = \beta_t^{\Tilde{4}, a}.
\end{aligned}
\end{equation*}
Note that $(\boldsymbol{\beta}^{\Tilde{1}}, \boldsymbol{\beta}^{\Tilde{2}}, \boldsymbol{\beta}^{\Tilde{3}}, \boldsymbol{\beta}^{\Tilde{4}})$ are all bounded due to Proposition \ref{remove xi}. If processes $\beta^{\Tilde{4}, a}$ and $\beta^{\Tilde{1}, b}$ are regarded as the pseudo-best ask and bid strategies, for agent $i\in\mathcal{I}$ let us consider the stochastic optimal control problem, where the inventory controlled by $\boldsymbol{\delta}^i\in\mathbb{H}^2\times\mathbb{H}^2$ reads
\begin{equation*}
    dQ_t^i = b_t\,\Lambda(\delta_t^{i, b}-\beta_t^{\tilde{1}, b})\,dt - a_t\,\Lambda(\delta_t^{i, a}-\beta_t^{\Tilde{4}, a})\,dt.
\end{equation*}
The agent $i$ aims at maximizing the associated control-version objective functional
\begin{equation}
\mathbb{E}\Big[\int_0^T\delta_t^{i,a}\,a_t\,\Lambda(\delta_t^{i,a}-\beta_t^{\tilde{4},a})\,dt+\int_0^T\delta_t^{i,b}\,b_t\,\Lambda(\delta_t^{i,b}-\beta_t^{\tilde{1},b})\,dt-\int_0^T \phi_t\, \big(Q_t^i\big)^2\,dt-A\,\big(Q_T^i\big)^2 \Big].
\label{aux_control}
\end{equation}
The control problem above is slightly more general than the one studied in \cite{guo2023macroscopic}, where the original best ask and bid strategies $(0,0)$ is replaced by $(\beta^{\Tilde{4}, a}, \beta^{\Tilde{1}, b})$. However, the stochastic maximum principle can still be applied to obtain the optimal feedback control 
\begin{equation}
    \hat{\delta}_t^{i,a} = \beta_t^{\tilde{4}, a} + \delta^*(Y_t^i - \beta_t^{\tilde{4}, a}) \text{ \; and  \; } \hat{\delta}_t^{i,b} = \beta_t^{\tilde{1}, b} + \delta^*(-Y_t^i - \beta_t^{\tilde{1}, b}),
    \label{aux_feedback_control}
\end{equation}
where the adjoint process $Y^i$ solves the FBSDE
\begin{equation}
      \left\{
    \begin{aligned}
     dQ_t^i &= b_t\,\Lambda\big(\delta^*(-Y_t^i - \beta_t^{\tilde{1}, b})\big)\,dt-a_t\,\Lambda\big(\delta^*(Y_t^i - \beta_t^{\tilde{4}, a})\big)\,dt,\\
    \, dY_t^i &= 2\phi_t\,Q_t^i\,dt+dM_t^i,\\
    Q_0^i &= q_0^i, \quad Y_T^i = -2A\,Q_T^i.
    \end{aligned}
    \right.
    \label{aux_control_fbsde}
\end{equation}
Similar to the techniques in \cite{guo2023macroscopic}, to solve the non-Lipschitz FBSDE \eqref{aux_control_fbsde} we first impose the regularizer $\xi>0$ to the forward equation so that it becomes
\begin{equation*}
    dQ_t^i = b_t\,\Lambda\big(\delta^*(-Y_t^i - \beta_t^{\tilde{1}, b})\vee(-\xi)\wedge\xi\big)\,dt-a_t\,\Lambda\big(\delta^*(Y_t^i - \beta_t^{\tilde{4}, a})\vee(-\xi)\wedge\xi\big)\,dt.
\end{equation*}
Since the equation is then Lipschitz, we prove the well-posedness the regularized FBSDE. Consequently, the solution $Y^i$ turns out to be bounded by some constant independent of $\xi$, which helps us remove the regularizer $\xi$ and obtain a solution solving the original equation. Finally, one can see the solution $(Q^i, Y^i, M^i)\in \mathbb{S}^2 \times \mathbb{S}^2 \times \mathbb{M}$ obtained for FBSDE \eqref{aux_control_fbsde} is also unique by a continuation argument as in \cite{peng1999fully}. Note that the control problem \eqref{aux_feedback_control} and \eqref{aux_control_fbsde} are also solved by artificial agent $\tilde{2}$ and $\Tilde{3}$; for illustration, we have 
\begin{equation*}
        \beta_t^{\Tilde{2}, a} = \beta_t^{\tilde{4}, a} + \delta^*(Y_t^{\tilde{2}} - \beta_t^{\tilde{4}, a}) \text{ \; and  \; } \beta_t^{\tilde{2}, b} = \beta_t^{\tilde{1}, b} + \delta^* (-Y_t^{\Tilde{2}} - \beta_t^{\tilde{1}, b}),
\end{equation*}
where $Y^{\tilde{2}}$ solves
\begin{equation*}
    \left\{
    \begin{aligned}
     dQ_t^{\Tilde{2}} &= b_t\,\Lambda\big(\delta^*(-Y_t^{\Tilde{2}} - \beta_t^{\tilde{1}, b})\big)\,dt-a_t\,\Lambda\big(\delta^*(Y_t^{\Tilde{2}} - \beta_t^{\tilde{4}, a})\big)\,dt,\\
    \, dY_t^{\Tilde{2}} &= 2\phi_t\,Q_t^{\Tilde{2}}\,dt+dM_t^{\Tilde{2}},\\
    Q_0^{\Tilde{2}} &= q_0^{\Tilde{2}}, \quad Y_T^{\Tilde{2}} = -2A\,Q_T^{\Tilde{2}}.
    \end{aligned}
    \right.
\end{equation*}
For any $i\in\mathcal{I}$, referring to \cite{guo2023macroscopic}, the monotonicity of the control with respect to the initial inventory yields

\begin{itemize}
    \item[(1)] $\hat{\delta}_t^{i, a} \geq \beta_t^{\Tilde{3}, a} = \beta_t^{\Tilde{4}, a}$,\\
    \vspace{-0.2cm}

    \item[(2)] $\hat{\delta}_t^{i, b} \geq \beta_t^{\Tilde{2}, b} = \beta_t^{\Tilde{1}, a}$,\\
    \vspace{-0.2cm}

    \item[(3)] $\hat{\delta}_t^{i, a} - \beta_t^{\Tilde{3}, a}\leq C\,(q_0^{\tilde{3}}-q_0^i)$ \,and\, $\hat{\delta}_t^{i, b} - \beta_t^{\Tilde{2}, b}\leq C\,(q_0^i-q_0^{\tilde{2}})$
\end{itemize}

\noindent for all $t$, where $C>0$ is some constant. Considering the assumption on sequences $(k_n)_{n\in\mathbb{N}}$ and $(l_n)_{n\in\mathbb{N}}$, we conclude that 
\begin{equation*}
    \inf_{i\in\mathcal{I}} \hat{\delta}_t^{i, b} = \beta_t^{\Tilde{2}, b} = \lim_{n \to \infty} \hat{\delta}_t^{k_n, b} \text{ \; and \; } \inf_{i\in\mathcal{I}} \hat{\delta}_t^{i, a} = \beta_t^{\Tilde{3}, a} = \lim_{n \to \infty} \hat{\delta}_t^{l_n, a}.
\end{equation*}
Finally, recalling that both $(k_n)_{n\in\mathbb{N}}$ and $(l_n)_{n\in\mathbb{N}}$ are non-duplicate, the strategy profile $(Q^i, Y^i, M^i)_{i\in\mathcal{I}}$ is a Nash equilibrium since
\begin{equation*}
\begin{aligned}
    dQ_t^i &= b_t\,\Lambda\big(\delta^*(-Y_t^i - \beta_t^{\tilde{1}, b})\big)\,dt - a_t\,\Lambda\big(\delta^*(Y_t^i - \beta_t^{\tilde{4}, a})\big)\,dt\\
    &=  b_t\,\Lambda\big(\delta^*(-Y_t^i - \inf_{\mathcal{I}\ni j \neq i} \hat{\delta}_t^{j, b})\big)\,dt - a_t\,\Lambda\big(\delta^*(Y_t^i - \inf_{\mathcal{I}\ni j \neq i} \hat{\delta}_t^{j, a})\big)\,dt.
\end{aligned}
\end{equation*}
In other words, processes $\beta^{\Tilde{4}, a}$ and $\beta^{\Tilde{1}, b}$ are the genuine best ask and bid strategies.
\end{proof}

\kong

The remainder of this section is dedicated to the well-posedness of BSRE \eqref{char_BSRE}. When the coefficients $a, b, \phi$, and $A$ are all deterministic, then $\mathcal{B}$ is also deterministic and thus the term $\mathcal{Z}$ becomes zero, simplifying BSRE \eqref{char_BSRE} into a matrix Riccati equation. Due to its extensive applications, the well-posedness of the matrix Riccati equation has attracted significant attention, with an early exploration credited to \cite{wonham1968matrix}. However, existing literature has predominantly focused on equations where the coefficients are either positive definite or symmetric. When $a, b, \phi$ and $A$ are random, the term $\mathcal{Z}$ must be non-zero, rendering \eqref{char_BSRE} genuinely stochastic. While \cite{peng1990general} introduced a stochastic adaptation of Bellman's quasi-linearization method in \cite{wonham1968matrix}, similar to the deterministic scenario, cases when coefficients are neither symmetric nor positive definite are rarely studied. To circumvent such difficulties, we first study the deterministic \eqref{char_BSRE} using the well-known Radon's lemma, followed by an analysis of the stochastic version, building upon insights gained in the previous step.

Assume $a, b, \phi$, and $A$ are all deterministic and consider the following linear matrix differential equations:
\begin{equation}
\begin{pmatrix}
    V'(t) \\
    U'(t)
\end{pmatrix}
=
\begin{pmatrix}
    0 & \mathcal{B}_t\\
    2\phi_t\, I & 0
\end{pmatrix} 
\begin{pmatrix}
    V(t) \\
    U(t)
\end{pmatrix}
; \quad
\begin{pmatrix}
    V(T) \\
    U(T)
\end{pmatrix}
=
\begin{pmatrix}
    I \\
    -2A\,I
\end{pmatrix}.
    \label{radon system}
\end{equation}
The purpose of Radon's lemma is to establish a connection between the matrix Riccati equation \eqref{char_BSRE} and the linear equation \eqref{radon system}:

\kong

\begin{lemma}[\cite{freiling2002survey}]
\label{radon lemma}
(1) Let $\mathcal{X}$ be a solution
of the deterministic Riccati equation \eqref{char_BSRE} on some interval $\mathcal{J}\subseteq [0,T]$ such that $T\in\mathcal{J}$. Denote by $V$ the unique solution of the linear equation
\begin{equation*}
    V'(t)=\mathcal{B}_t \, \mathcal{X}_t \, V(t); \quad V(T)=I,
\end{equation*}
for $t\in \mathcal{J}$ and set $U(t)=\mathcal{X}_t\,V(t)$. Then, the matrix $\begin{pmatrix}
    V(t) \\ U(t)
\end{pmatrix}$
defines for $t \in \mathcal{J}$ the solution of the linear differential equation \eqref{radon system};

(2) Suppose $\begin{pmatrix}
    V(t) \\ U(t)
\end{pmatrix}$
is on some interval $\mathcal{J} \subseteq [0,T]$ a solution of the linear differential equation \eqref{radon system} such that $\det V(t)\neq 0$ for all $t \in \mathcal{J}$, then
\begin{equation}
    \mathcal{X}: \mathcal{J}\to \mathbb{R}^{N\times N}, \quad t \mapsto U(t)\, V(t)^{-1}=:\mathcal{X}_t
    \label{radon bridge}
\end{equation}
is a solution to the deterministic Riccati equation \eqref{char_BSRE} on $\mathcal{J}$.
\end{lemma}

\kong

\noindent Radon’s lemma reveals that the matrix Riccati equation \eqref{char_BSRE} is locally equivalent to the linear differential equation \eqref{radon system}. This equivalence persists until a potentially finite blow-up time, as indicated by \eqref{radon bridge}. Furthermore, it is evident from \eqref{radon bridge} that the solution $\mathcal{X}$ of the matrix Riccati equation experiences blow-ups at moments when $\det V(t)$ vanishes. Hence, we shift our focus to the linear differential equation \eqref{radon system} and verify the non-singularity of $V$.

\kong

\begin{theorem}
\label{deterministic game}
    Assume the coefficients $a, b, \phi$, and $A$ are all deterministic. Then the matrix Riccati equation \eqref{char_BSRE} accepts a unique bounded solution $\mathcal{X}$ for cases:
    
    \begin{itemize}
        \item[(i)] when $\phi_t=0$ for all $t$;\\
        \vspace{-0.2cm}

        \item[(ii)] when $N=2$.
        
    \end{itemize}
    
    \noindent Moreover, in both cases, the solution $\mathcal{X}_t$ has row sum $-2A - 2\int_t^T \phi_s\,ds$ for all $t$.
\end{theorem}

\begin{proof}
    If $\phi \equiv 0$ as in case (i), then $U$ and $V$ can be solved one-by-one in the linear system \eqref{radon system}. This allows us to write the solution explicitly as
    \begin{equation*}
        U(t) = -2A \, I \text{ \; and \; } V(t) = I + 2A\,\int_t^T \mathcal{B}_s\,ds.
    \end{equation*}
    Recalling that $\mathcal{B}_s$ is an $M_0$-matrix for all $s$, then $\int_t^T \mathcal{B}_s\,ds$ is also an $M_0$-matrix and $V(t)$ is then a $Z_+$-matrix with row sum $1$. By Lemma \ref{Z-matrix}, we know $V(t)$ is an $M$-matrix, the non-singularity of which can be deduced from its strict diagonal dominance. According to Lemma \ref{radon lemma}, the unique solution $\mathcal{X}$ of matrix Riccati equation \eqref{char_BSRE} can be represented by $\mathcal{X}_t = U(t) \, V(t)^{-1}$. While $V(t)$ is a strictly diagonally dominant $M$-matrix, it further implies $V(t)^{-1}$ is a non-negative matrix and is strictly diagonally dominant of column entries. Defining $\Vec{1}=(1, \dots, 1)$, the fact $V(t)\cdot \Vec{1} = \vec{1}$ yields
    \begin{equation*}
        \Vec{1} = I \cdot \Vec{1} = V(t)^{-1} \cdot V(t) \cdot \Vec{1} = V(t)^{-1} \cdot \Vec{1}.
    \end{equation*}
    This tells us that $V(t)^{-1}$ has row sum $1$. Denote by $C([0,T];\mathbb{R}^{2\times 2})$ the set of matrix-valued functions on $[0,T]$ that is element-wise continuous. Let $C_1([0,T];\mathbb{R}^{2\times 2})\subset C([0,T];\mathbb{R}^{2\times 2})$ be the class such that, for any fixed $t$, the associated matrix is a $Z_+$-matrix with row sum $1$. In the set $C([0,T];\mathbb{R}^{2\times 2})$, consider the norm
    \begin{equation*}
        \|P\|_{\ell}:=\sup_{t\in[0,T]} e^{-\ell \, (T-t)} \cdot \|P(t)\|_F,
    \end{equation*}
    where $\ell$ is a positive constant to be specified, and $\|\cdot\|_F$ represents the Frobenius norm. The supremum in the $\|\cdot\|_\ell$ ensures that both $C([0,T];\mathbb{R}^{2\times 2})$ and $C_1([0,T];\mathbb{R}^{2\times 2})$ are complete with respect to $\|\cdot\|_\ell$. Picking any $P \in C_1([0,T];\mathbb{R}^{2\times 2})$, let us solve the linear equation
    \begin{equation*}
        U'(t) = 2\phi_t \, P(t); \quad U(T) = -2A\, I,
    \end{equation*}
    and denote by $U$ its solution. This yields 
    \begin{equation*}
        -U(t)= 2A\,I + 2\int_t^T \phi_s\,P(s)\,ds
    \end{equation*}
    and consequently $-U(t)$ is a $Z_+$-matrix with row sum $2A + 2\int_t^T \phi_s\,ds$ for any $t$. Given such $U$, we move on to the other linear equation
    \begin{equation*}
        V'(t)=\mathcal{B}_t \, U(t); \quad V(T)=I
    \end{equation*}
    and its solution $V$ reads
    \begin{equation*}
        V(t)= I+\int_t^T \mathcal{B}_s \big[-U(s)\big]\,ds.
    \end{equation*}
    Considering the following facts: (1) $\mathcal{B}_s$ is an $M_0$-matrix; (2) $-U(s)$ is $Z_+$-matrix with a uniform row sum $2A + 2\int_t^T \phi_s\,ds$; (3) both of them are $2$-by-$2$ matrices, some direct calculations infer $\mathcal{B}_s [ - U(s) ]$ is an $M_0$-matrix. Indeed, we apply the $2$-dimensional condition only for the property that $M$-matrices are closed under multiplication in $2$-dimension. The matrix $V(t)$ is thus a $Z_+$-matrix with row sum $1$. In summary, the above procedure defines a map
    \begin{equation*}
        C_1([0,T];\mathbb{R}^{2\times 2}) \ni P \hookrightarrow \mathfrak{T}(P) = V \in C_1([0,T];\mathbb{R}^{2\times 2}).
    \end{equation*}
    We proceed to show $\mathfrak{T}$ is a contraction with respect to $\| \cdot \|_\ell$ for some large $\ell$.

    For any $P, \Tilde{P} \in C_1([0,T];\mathbb{R}^{2\times 2})$, write $V = \mathfrak{T}(P)$ and $\Tilde{V} = \mathfrak{T}(\Tilde{P})$. Utilizing the Fubini's theorem, the map $\mathfrak{T}$ can be briefly written as
    \begin{equation*}
    \begin{aligned}
        V(t) &= I + 2\int_t^T \mathcal{B}_s \Big(A\,I + \int_s^T \phi_u\,P(u)\,du\Big)\,ds\\
        &=I + 2A\int_t^T\mathcal{B}_s\,ds + 2\int_t^T \Big(\int_t^u \mathcal{B}_s\,ds\Big)\, \phi_u \, P(u)\, du.
    \end{aligned}
    \end{equation*}
    We can then see it holds for all $t$ that
    \begin{equation*}
    \begin{aligned}
        \|\Tilde{V}(t)-V(t)\|_F & \leq 2\int_t^T \big\| \phi_u \, \Big(\int_t^u \mathcal{B}_s\,ds\Big) \big\|_F \cdot \|\Tilde{P}_u-P_u\|_F \, du\\
        &\leq C \int_t^T e^{\ell\,(T-u)} \cdot e^{-\ell\,(T-u)} \, \|\Tilde{P}_u-P_u\|_F \, du\\
        &\leq C \, \frac{e^{\ell\,(T-t)}-1}{\ell} \, \|\Tilde{P}-P\|_\ell,
    \end{aligned}
    \end{equation*}
    which immediately yields
    \begin{equation*}
    \begin{aligned}
        e^{-\ell \, (T-t)} \, \|\Tilde{V}(t)-V(t)\|_F &\leq C \, \frac{e^{\ell\,(T-t)} - 1}{\ell \, e^{\ell\,(T-t)}} \, \|\Tilde{P}-P\|_\ell,\\
        \|\tilde{V}-V\|_\ell &\leq \frac{C}{\ell} \, \|\Tilde{P}-P\|_\ell.
    \end{aligned}
    \end{equation*}
    It suffices to pick $\ell > C$ for $\mathfrak{T}$ to be a contraction. By the Banach fixed-point theorem, the map $\mathfrak{T}$ accepts a unique fixed point, corresponding to the unique solution of the linear differential equation \eqref{radon system}. Therefore, we can conclude for all $t$ that:
    
    \begin{itemize}
        \item[(1)] $V(t)$ as the solution of \eqref{radon system} is a $Z_+$-matrix with row sum $1$ and thus is non-singular;\\
        \vspace{-0.2cm}
        
        \item[(2)] $-U(t)$ is a $Z_+$-matrix with row sum $2A + 2\int_t^T \phi_s\,ds$.
    \end{itemize}

    \noindent Discussed above, the matrix $V(t)^{-1}$ is then a non-negative matrix with row sum $1$. We can now see $-U(t) \, V(t)^{-1}$ has the row sum $-2(A+\int_t^T \phi_s\,ds)$.
\end{proof}

\kong

\begin{remark}
    Case (i) in Theorem \ref{deterministic game} pertains to the differential game of $N$ risk-neutral players, while case (ii) is concerned with the two-player differential game. The following Theorem \ref{2_dim bsre} addresses the two-player stochastic differential game.
\end{remark}

\kong

\noindent One crucial characteristic of the solution, shown in Theorem \ref{deterministic game}, is the uniform row sum of the solution $\mathcal{X}_t$ across all $t$. This attribute helps solve the stochastic Riccati equation and further enables the derivation of additional properties of the solution.

\kong

\begin{theorem}
\label{2_dim bsre}
    If $N=2$, then the BSRE \eqref{char_BSRE} accepts a unique bounded solution $\mathcal{X}$. Furthermore, for all $t$ the matrix $\mathcal{X}_t$ is non-positive with row sum $-2 \, \mathbb{E}_t[A+\int_t^T \phi_s \, ds]$, and is strictly diagonally dominant of  column entries.
\end{theorem}

\begin{proof}
    Since the solution exhibits a uniform row sum of $-2A-2\int_t^T \phi_s\,ds$ for all $t$ in the deterministic case, we consider the ansatz that the row sum becomes $\mathcal{S}_t:=\mathbb{E}_t[-2A-2\int_t^T \phi_s\,ds]$ in the stochastic context. Specifically, we propose the form
    \begin{equation}
    \label{bsre_ansatz}
        \mathcal{X}_t = \begin{bmatrix}
           \mathcal{X}_{11}(t) & \mathcal{X}_{12}(t)\\
            \mathcal{X}_{21}(t) & \mathcal{X}_{22}(t)
        \end{bmatrix} =
        \begin{bmatrix}
            \chi_1(t) & \mathcal{S}_t - \chi_1(t)\\
            \mathcal{S}_t - \chi_2(t) & \chi_2(t)
        \end{bmatrix}
    \end{equation}
    where $\chi_1, \chi_2$ are some processes to be determined. Noting that $\mathcal{B}$ can be represented as
    \begin{equation*}
        \mathcal{B}_t =
        \begin{bmatrix}
            \mathfrak{b}_1(t)&  -\mathfrak{b}_1(t)\\
            -\mathfrak{b}_2(t)& \mathfrak{b}_2(t)
        \end{bmatrix}
    \end{equation*}
    for some bounded non-negative processes $\mathfrak{b}_1$ and $\mathfrak{b_2}$, we compute
    \begin{equation*}
    \begin{aligned}
        &\mathcal{X}_t \, \mathcal{B}_t \, \mathcal{X}_t\\
        & \hspace{0.6cm} = 
        \begin{bmatrix}
             (\mathfrak{b}_1\chi_1 - \mathfrak{b}_2\mathcal{S} + \mathfrak{b}_2\chi_1)\cdot(\chi_1 - \mathcal{S} + \chi_2) & -(\mathfrak{b}_1\chi_1 -\mathfrak{b}_2\mathcal{S} + \mathfrak{b}_2\chi_1)\cdot(\chi_1 - \mathcal{S} + \chi_2)\\
            -(\mathfrak{b}_1\chi_2 - \mathfrak{b}_1\mathcal{S} + \mathfrak{b}_2\chi_2)\cdot(\chi_1 - \mathcal{S} + \chi_2) & (\mathfrak{b}_1\chi_2 - \mathfrak{b}_1\mathcal{S} + \mathfrak{b}_2\chi_2)\cdot(\chi_1 - \mathcal{S} + \chi_2)
        \end{bmatrix},
    \end{aligned}
    \end{equation*}
where the $t$ variable is omitted for conciseness whenever appropriate. Since $\mathbb{E}_t[-2A - 2\int_0^T \phi_s\,ds]$ is a martingale, we define $\mathcal{H}\in\mathbb{H}^2$ such that
\begin{equation*}
    \mathbb{E}_t\big[ -2A - 2\int_0^T \phi_s\,ds \big] = \int_0^t \mathcal{H}_s \, dW_s.
\end{equation*}
Plugging \eqref{bsre_ansatz} back to \eqref{char_BSRE}, suppose the BSDE
\begin{equation*}
    d\chi_1(t) =  \Big[ -\big( \mathfrak{b}_1\chi_1 - \mathfrak{b}_2\mathcal{S} + \mathfrak{b}_2\chi_1 \big) \big(\chi_1 - \mathcal{S} + \chi_2 \big) + 2\phi_t \Big]\, dt + \mathcal{Z}_{11}(t) \, dW_t
\end{equation*} 
with $\chi_1(T) = -2A$ is well-posed, we can have
\begin{equation*}
\begin{aligned}
    d\mathcal{X}_{12}(t) &= d\mathcal{S}_t - d\chi_1(t)\\
    &= 2\phi_t\,dt + \mathcal{H}_t\,dW_t - \Big[ - \big( \mathfrak{b}_1\chi_1 - \mathfrak{b}_2\mathcal{S} + \mathfrak{b}_2\chi_1 \big) \big(\chi_1 - \mathcal{S} + \chi_2 \big) + 2\phi_t \Big]\, dt - \mathcal{Z}_{11}(t) \, dW_t\\
    &= \big( \mathfrak{b}_1\chi_1 - \mathfrak{b}_2\mathcal{S} + \mathfrak{b}_2\chi_1 \big) \big(\chi_1 - \mathcal{S} + \chi_2 \big) \, dt + \big( \mathcal{H}_t - \mathcal{Z}_{11}(t) \big) \, dW_t
\end{aligned}
\end{equation*}
and $\mathcal{X}_{12}(T) = -2A + 2A = 0$. Since it suffices to let $\mathcal{Z}_{12}(t) = \mathcal{H}_t - \mathcal{Z}_{11}(t)$, the hypothesis on the position of $\mathcal{X}_{12}(t)$ has been checked; the one of $\mathcal{X}_{21}(t)$ can be justified via similar calculations. We can now shift our focus to the BSDEs for $\chi_1$ and $\chi_2$:
\begin{equation}
\begin{aligned}
    d\chi_1(t) &=  \Big[ - \big( \mathfrak{b}_1\chi_1 - \mathfrak{b}_2\mathcal{S} + \mathfrak{b}_2\chi_1 \big) \cdot \big(\chi_1 - \mathcal{S} + \chi_2 \big) + 2\phi_t \Big]\, dt + \mathcal{Z}_{11} (t) \, dW_t; \quad \chi_1(T) = -2A,\\
    d\chi_2(t) &=  \Big[ - \big(\mathfrak{b}_1\chi_2 - \mathfrak{b}_1\mathcal{S} + \mathfrak{b}_2\chi_2 \big) \cdot \big(\chi_1 - \mathcal{S} + \chi_2\big) + 2\phi_t \Big]\, dt + \mathcal{Z}_{22} (t) \, dW_t; \quad \chi_2(T) = -2A.
    \label{bsde_chi}
\end{aligned}
\end{equation}
Setting $\vartheta_1:=\chi_1+\chi_2-\mathcal{S}$ and $\vartheta_2:=\chi_1-\mathcal{S}$, the system \eqref{bsde_chi} can be equivalently transformed to
\begin{equation}
\begin{aligned}
    d\vartheta_1(t) &= \Big[ - \big(\mathfrak{b}_1(t)+\mathfrak{b}_2(t) \big) \, \vartheta_1(t)^2 + 2\phi_t \Big]\, dt + \tilde{\mathcal{Z}}_{1}(t)\,dW_t; \quad \vartheta_1(T) = -2A,\\
    d\vartheta_2(t) &=   - \vartheta_1(t) \, \Big[  \big( \mathfrak{b}_1(t) + \mathfrak{b}_2(t) \big) \, \vartheta_2(t) +\mathfrak{b}_1(t)\,\mathcal{S}_t \Big] \, dt + \tilde{\mathcal{Z}}_{2}(t)\,dW_t; \quad \vartheta_2(T) = 0.
    \label{bsde_chi_equiv}
\end{aligned}
\end{equation}
Notice the first equation of \eqref{bsde_chi_equiv} is decoupled and its well-posedness is studied in \cite{guo2023macroscopic}. We thus know there exists a unique bounded solution $\vartheta_1$, satisfying $\vartheta_1(t) \in [ - 2( \Bar{A} + \bar{\phi} \, T), 0 ]$ for all $t$. Whereas the second equation in \eqref{bsde_chi_equiv} is linear, we can write the solution explicitly as
\begin{equation*}
    \vartheta_2(t)=\mathbb{E}_t \Big [ \int_t^T \vartheta_1(u) \, \mathfrak{b}_1(u) \, \mathcal{S}_u \, e^{\int_t^u \vartheta_1(s) \, (\mathfrak{b}_1(s) + \mathfrak{b}_2(s) )\,ds} \, du \Big] \in[0,  4 \Bar{\mathfrak{b}} \, (\Bar{A} + \Bar{\phi} \, T)^2 \, T ],
\end{equation*}
where $\Bar{\mathfrak{b}}$ is the upper bound for the process $\mathfrak{b}_1$.  Bounded solutions $\chi_1$ and $\chi_2$ can then be constructed based on $\vartheta_1$ and $\vartheta_2$. On the other hand, given the process $\vartheta_1$, the process $\chi_1$ also satisfies the linear BSDE
\begin{equation*}
    d\chi_1(t)=\Big[ - \vartheta_1 \, ( \mathfrak{b}_1 + \mathfrak{b}_2 ) \, \chi_1 + \big( 2\phi +  \mathfrak{b}_2 \, \mathcal{S} \, \vartheta_1 \big) \Big] \, dt + \mathcal{Z}_{11}(t) \, dW_t;  \quad \chi_1(T) = -2A,
\end{equation*}
the unique solution of which reads
\begin{equation*}
\begin{aligned}
    \chi_1(t) &= -\, \mathbb{E} \Big[ 2A \, e^{\int_t^T \vartheta_1(u) \, (\mathfrak{b}_1(u) + \mathfrak{b}_2(u)) \, du} + \int_t^T \big( 2\phi_u +  \mathfrak{b}_2(u) \, \mathcal{S}(u) \, \vartheta_1(u) \big) \, e^{\int_t^u \vartheta_1(s) \, (\mathfrak{b}_1(s) + \mathfrak{b}_2(s)) \, ds} \, du\Big]\\
    &\leq 0.
\end{aligned}
\end{equation*}
Combining the results on $\chi_1$, $\vartheta_1$, and $\vartheta_2$, in the matrix form \eqref{bsre_ansatz} we can conclude that: (i) entries of the first row are non-positive; (2) the difference $\mathcal{X}_{11}-\mathcal{X}_{21} \leq 0$. With a symmetric discussion for the $\chi_2$, the following can be similarly achieved: (iii) entries of the second row are non-positive; (iv) the difference $\mathcal{X}_{22}-\mathcal{X}_{12} \leq 0$. For the uniqueness, it suffices to see that any two bounded solutions solve the same Lipschitz BSDE, which only admits a unique solution.
\end{proof}

\kong

\noindent Since we have solved the two-dimensional case completely, the infinite-player game studied in the beginning of this section then has a Nash equilibrium.

\kong

\begin{proposition}
    The infinite-player game described in Proposition \ref{paper 2 mean field game} has a Nash equilibrium.
\end{proposition}

\kong

The remaining of this section is devoted to the economic interpretation. We regard the price impact as the price movement triggered by order imbalances. Given some $q_0\in\mathbb{R}$, consider the following condition:
\begin{equation}
    q_0^i = q_0 \text{\, for all \,} i.
    \label{weak comp}
\end{equation}
Since \eqref{weak comp} indicates that all agents share the same initial inventory, the ordering property \ref{ordering in general} tells us that they also share the same strategy in the equilibrium. Hence, we treat the price competition under \eqref{weak comp} as \textit{weak}, and it holds for all $i$ that
\begin{equation*}
    dQ_t^i = b_t\,\Lambda(0)\,dt-a_t\,\Lambda(0)\,dt
\end{equation*}
and $Y^i$ can be subsequently computed explicitly. For simplicity, here we let $q_0=0$ and $\phi_t = \phi$, with $\phi$ and $A$ being non-negative constants. Also, we let $\Lambda$ be the exponential function as in Example \ref{exp game}. The ask price at time $t$ then reads
\begin{equation}
\begin{aligned}
    S_t + \frac{1}{\gamma} &+ 2\big[A + \phi(T-t)\big] \int_0^t (a_u - b_u) \, du \\
    & + 2A \, \mathbb{E}_t\Big[ \int_t^T (a_u - b_u) \, du \Big] + 2\phi \, \mathbb{E}_t\Big[ \int_t^T (T - u)(a_u - b_u) \, du \Big],
    \label{explicit price}
\end{aligned}
\end{equation}
where integral terms are considered as the price impact. The impact is further linear if $\phi = 0$. We regard the integral in the first line of \eqref{explicit price} as the \textit{ex post} impact since it encapsulates the influence of previous orders. Additionally, integrals in the second line of \eqref{explicit price}---indicating the adjustment in correspondence to the expected future imbalances---are considered as the \textit{ex ante} impact.

We then look at the two-player game for convenience. First, note that the decoupling field can be seen as a generalization of the derivative of the terminal payoff. More specifically, fixing any $s\in [0,T]$, we can truncate the solution $(\boldsymbol{Q}_{t}, \boldsymbol{Y}_t, \boldsymbol{M}_t)_{t\in[0,T]}$ with respect to time, resulting in $(\boldsymbol{Q}_{t}, \boldsymbol{Y}_t, \boldsymbol{M}_t)_{t\in[0,s]}$. In reference to agent $1$, the corresponding strategy $(\psi^{1,a}(\boldsymbol{Y}_t), \psi^{1,b}(\boldsymbol{Y}_t))_{t\in[0,s]}$ will solve an optimal control problem if $u^1(s, (q^1, Q^2_s))$, as a function of $q^1$, can be regarded as the derivative of a new terminal penalty at time $s$. The stochastic maximum principle implies that a sufficient condition for this is: the terminal penalty is concave with respect to $q^1$. To verify the concavity, the definition of the matrix $\chi$ and Theorem \ref{2_dim bsre} yield 
\begin{equation}
    u^1(s, (\tilde{q}^1, q^2)) - u^1(s, (q^1, q^2)) = -C \, (\tilde{q}^1 -q^1)
    \label{con monotoncity}
\end{equation}
for any $\tilde{q}^1, q^1\in\mathbb{R}$, where $C$ is some non-negative random variable. If $P(s, q^1, q^2)$ is some function such that $\partial P / \partial q^1 = u^1$, it can be inferred from \eqref{con monotoncity} that $P$ is concave in $q^1$. Hence, the process $(\psi^{1,a}(\boldsymbol{Y}_t), \psi^{1,b}(\boldsymbol{Y}_t))_{t\in[0,s]}$ solves the stochastic control problem with $P(s, \cdot, Q_s^2)$ as the terminal penalty. On the other hand, due the symmetry $u^1(s, (q^1, q^2)) = u^2(s, (q^2, q^1))$, the new terminal penalty for agent $2$ can be $P(s, \cdot, Q_s^1)$. It suffices to look at one player.

We introduce the $\mathcal{F}_s$-measurable random variable
\begin{equation}
\begin{aligned}
    \mathfrak{q} := \mathbb{E}_s\Big[ A \, \int_s^T \Lambda(0)\,(a_u - b_u) \,du + \int_s^T &\phi_r \, \big( \int_s^r \Lambda(0)\,(a_u  - b_u) \,du\big) \,dr \Big] \\
    & \Big/ \; \mathbb{E}_s\Big[A + \int_s^T \phi_r \, dr \Big],
\end{aligned}
\label{benchmark q}
\end{equation}
where the penalty coefficients are forced to be non-zero. Note that $\mathfrak{q}$ represents a weighted average
of future order imbalances, determined by the penalty coefficients, and $\mathfrak{q}=0$ if $s=T$. By the definition of decoupling fields, the value of $u(s, (\mathfrak{q}, \mathfrak{q}))$ is determined by $(Y_s^1, Y_s^2)$ that solves the FBSDE
\begin{equation}
    \left\{
    \begin{aligned}
     dQ_t^i &= b_t\,\Lambda\big(\psi^{i,b}(\boldsymbol{Y}_t)-\bar{\psi}^{i,b}(\boldsymbol{Y}_t)\big)\,dt-a_t\,\Lambda\big(\psi^{i,a}(\boldsymbol{Y}_t)-\bar{\psi}^{i,a}(\boldsymbol{Y}_t)\big)\,dt,\\
    \, dY_t^i &= 2\phi_t\,Q_t^i\,dt+dM_t^i,\\
    Q_s^i &= \mathfrak{q}, \quad Y_T^i = -2A\,Q_T^i,
    \end{aligned}
    \right.
    \label{bench fbsde}
\end{equation}
for all $i$. Because the initial inventories are identical, it follows 
\begin{equation*}
    Y_s^i = - \mathbb{E}_s\Big[ A \, \big(\mathfrak{q} + \int_s^T \Lambda(0)\,(b_r-a_r) \,dr\big) + \int_s^T \phi_r \, \big(\mathfrak{q} + \int_s^r \Lambda(0)\,(b_u-a_u) \,du \big)\Big].
\end{equation*}
Due to the choice \eqref{benchmark q}, we can now see $Y_s^i=0$ and thus $u(s, (\mathfrak{q}, \mathfrak{q}))$ is a vector with zeros. The concavity implies that the agent is still penalizing the excessive inventories at time $s$. Moreover, conditional on $Q_s^2 = \mathfrak{q}$, the penalty is applied to the deviation from $\mathfrak{q}$ as a generalization of $0$ at time $T$. For the case that $q^2 > q$, Theorem \ref{2_dim bsre} again yields
\begin{equation*}
    u^1(s, (\mathfrak{q}, q^2)) = u^1(s, (\mathfrak{q}, \mathfrak{q}))  -C \, (q^2 - \mathfrak{q}) \leq u^1(s, (\mathfrak{q}, \mathfrak{q}))=0,
\end{equation*}
where $C$ is some non-negative random variable. Notice the (generalized) derivative of $u^1$ with respect to $q^1$ is bounded away from $0$ if the same condition holds for coefficients $\phi$ and $A$. Then, there exists a random variable $q^* \leq \mathfrak{q}$ such that $u^1(s, (q^*, q^2))=0$. In the economic context, it signifies that if the inventory level of the other player is higher than the benchmark value $\mathfrak{q}$ at time $s$, the target of her own inventory will be smaller than $\mathfrak{q}$. The case $q^2 < \mathfrak{q}$ can be similarly discussed.

\vspace{0.2cm}

\section{Heterogeneous Risk Coefficients}
\label{paper 2 section 7}
\noindent While the preceding sections mainly focus on homogeneous agents, this final section delves into heterogeneous risk coefficients under the linear setting. In addition to the well-posedness result, we will present an example to show that the desirable ordering property in the homogeneous case breaks down.

We consider a two-player game in the linear setting of Section \ref{paper 2 section 2}.  Recall that the inventory and cash of the agent $1$ are modelled by $X_t^1$ and $Q_t^1$ accordingly:
\begin{gather*}
    X_t^1 = \int_0^t(S_u + \delta_u^{1,a})\,a_u\,(\,\zeta-\gamma\,(\delta_u^{1,a}-\delta_u^{2,a}) \,)\, du -\int_0^t(S_u - \delta_u^{1,b}) \, b_u \, (\,\zeta-\gamma\,(\delta_u^{1,b}-\delta_u^{2,b}) \,) \, du,\\
    Q_t^1 = q_0^1-\int_0^t a_u \,(\,\zeta - \gamma\,(\delta_u^{1,a}-\delta_u^{2,a}) \,) \, du + \int_0^t b_u \, (\,\zeta-\gamma\,(\delta_u^{1,b} - \delta_u^{2,b})\,)\,du.
    \nonumber   
\end{gather*}
Note that $\bar{\delta}^{1,a}$ can be replaced by $\delta^{2,a}$. The ones of agent $2$ are symmetric. Controlling $\boldsymbol{\delta}^i\in\mathbb{H}^2 \times \mathbb{H}^2$, the agent $i\in\{1,2\}$ aims at maximizing the objective functional
\begin{equation*}
\begin{aligned}
    J(\boldsymbol{\delta}^i; \boldsymbol{\delta}^{-i}):& = \mathbb{E}\Big[X_T^i + S_T\,Q_T^i - \int_0^T \phi_t^i \, \big(Q_t^i\big)^2\,dt - A^i\, \big(Q_T^i\big)^2 \Big]\\
     &= \mathbb{E}\Big[\int_0^T \delta_t^{i,a} \, a_t \, (\,\zeta-\gamma\,(\delta_t^{i,a}-\bar{\delta}_t^{i,a})\,)\,dt + \int_0^T \delta_t^{i,b} \, b_t \, (\,\zeta - \gamma\,(\delta_t^{i,b}-\bar{\delta}_t^{i,b})\,)\,dt\\
     &\hspace{3cm} -\int_0^T \phi_t^i \, \big(Q_t^i\big)^2\,dt - A^i\,\big(Q_T^i\big)^2 \Big],
\end{aligned}
\end{equation*}
where penalties $\phi^i:=(\phi_t^i)_{t\in[0,T]}\in\mathbb{H}^2$ and $A^i\in L^2(\Omega, \mathcal{F}_T)$ are non-negative, satisfying $\phi_t^i\leq\bar{\phi}$ and $A^i\leq\bar{A}$ for some constants $\bar{\phi}, \bar{A}>0$. The goal is to find a Nash equilibrium in the same sense as \eqref{Nash}.

\kong

\begin{theorem}    
\label{heter game}
A strategy profile $(\boldsymbol{\delta}^j)_{1\leq j\leq 2} \in (\mathbb{H}^2\times\mathbb{H}^2)^{2}$ forms a Nash equilibrium if and only if 
\begin{equation}
    \begin{aligned}
        \delta_t^{1,a} &= \frac{2}{3}\,Y_t^1 + \frac{1}{3}\, Y_t^2 + \frac{\zeta}{\gamma}, \quad \delta_t^{1,b} = -\frac{2}{3}\,Y_t^1 - \frac{1}{3}\, Y_t^2 + \frac{\zeta}{\gamma},\\
        \delta_t^{2,a} &= \frac{1}{3}\,Y_t^1 + \frac{2}{3}\, Y_t^2 + \frac{\zeta}{\gamma},  \quad \delta_t^{2,b} = -\frac{1}{3}\,Y_t^1 - \frac{2}{3}\, Y_t^2 + \frac{\zeta}{\gamma},
    \end{aligned}
    \label{alge_trick}
\end{equation}
where $(Y^1, Y^2)$ solves the following system of FBSDEs:
\begin{equation}
\left\{
\begin{aligned}
\;& dQ_t^1  = \zeta \, (b_t - a_t) \, dt + \gamma \, (a_t + b_t)\,(Y_t^1 -Y_t^2)/3 \, dt, \\
& dY_t^1 = 2\phi_t^1 \, Q_t^1 \, dt + dM^1_t, \\
& Q_0^1 = q_0^1,\quad Y_T^1 = -A^1\,Q_T^1;
\end{aligned}
\right.
\label{heter_1}
\end{equation}
\begin{equation}
\left\{
\begin{aligned}
\;& dQ_t^2  = \zeta \, (b_t - a_t) \, dt + \gamma \, (a_t + b_t)\,(Y_t^2 -Y_t^1)/3 \, dt, \\
& dY_t^2 = 2\phi_t^2 \, Q_t^2 \, dt + d M^2_t, \\
& Q_0^2 = q_0^2, \quad Y_T^2 = -A^2\,Q_T^2,
\end{aligned}
\right.
\label{heter_2}
\end{equation}
\noindent with $M^i_t\in\mathbb{M}$ for all $i \in \{1, 2\}$.
\end{theorem}

\begin{proof}
With reference to the homogeneous case, the heterogeneity of the penalty parameters does not affect the concavity of the objective functional. By Theorem \ref{conv_ana}, strategy profile $(\boldsymbol{\delta}^j)_{1\leq j\leq 2} \in (\mathbb{H}^2\times\mathbb{H}^2)^{2}$ forms a Nash equilibrium if and only if it solves
    \begin{equation}
\left\{
\begin{aligned}
\;& dQ_t^i  = -a_t\,\big(\zeta+\gamma\bar{\delta}^{i,a}_t-\gamma\delta_t^{i,a}\big)dt + b_t\,\big(\zeta+\gamma\bar{\delta}^{i,b}_t-\gamma\delta_t^{i,b}\big)dt, \\
& d\delta_t^{i,a} = d\bar{\delta}^{i,a}_t/2 + \phi_t\,Q_t^i\,dt - d\tilde{M}_t^i,\\
& d\delta_t^{i,b} = d\bar{\delta}^{i,b}_t/2 - \phi_t\,Q_t^i\,dt + d\tilde{M}_t^i,\\
& Q_0^i = q_0^i,\quad \delta_T^{i,a} = \zeta/(2\gamma) + \bar{\delta}^{i,a}_T/2-A\,Q_T^i,\quad \delta_T^{i,b} = \zeta/(2\gamma) + \bar{\delta}^{i,b}_T/2 + A\,Q_T^i,
\end{aligned}
\right.
\label{origin_conv}
\end{equation}
with $\tilde{M}^i \in \mathbb{M}$ for all $i \in \{1, 2\}$. System \eqref{heter_1}-\eqref{heter_2} can be obtain from \eqref{origin_conv} via the transform \eqref{alge_trick}.
\end{proof}

\kong

We introduce the following processes and random variables:
\begin{gather*}    
    \boldsymbol{H}_t = \frac{\gamma}{3} \, (a_t + b_t) \, \begin{bmatrix}
        1 & -1\\
        -1 & 1  
    \end{bmatrix}, \hspace{0.4cm}
    \boldsymbol{G}_t =  
    \begin{pmatrix}
        \, \zeta \, (b_t - a_t) \, \\
        \, \zeta \, (b_t - a_t) \, 
        \end{pmatrix}, \\
    \boldsymbol{D}_t = 
    \begin{bmatrix}
        2 \, \phi_t^1 & 0\\
        0 & 2 \, \phi_t^2  
    \end{bmatrix}, \hspace{0.4cm}
    \boldsymbol{A} = \begin{bmatrix}
        2 A^1 & 0\\
        0 & 2 A^2  
    \end{bmatrix}.
\end{gather*}
System \eqref{heter_1}-\eqref{heter_2} can then be written neatly in the vector form
\begin{equation}
\left\{
\begin{aligned}
\;& d\boldsymbol{Q}_t  = \boldsymbol{G}_t \, dt + \boldsymbol{H}_t \, \boldsymbol{Y}_t \, dt, \\
& d\boldsymbol{Y}_t = \boldsymbol{D}_t \, \boldsymbol{Q}_t \, dt + d\boldsymbol{M}_t, \\
& \boldsymbol{Q}_0 = \boldsymbol{q}_0, \quad \boldsymbol{Y}_T = - \boldsymbol{A} \, \boldsymbol{Q}_T.
\end{aligned}
\right.
\label{heter}
\end{equation}

\kong

\begin{theorem}
    The FBSDE \eqref{heter} accepts a unique solution $(\boldsymbol{Q}, \boldsymbol{Y}, \boldsymbol{M}) \in (\mathbb{S}^2 \times \mathbb{S}^2 \times \mathbb{M})^2$. Moreover, the solution has the relation
    \begin{equation}
        \boldsymbol{Y}_t = \boldsymbol{R}_t \, \boldsymbol{Q}_t + \boldsymbol{P}_t,
        \label{linear_decoup}
    \end{equation}
    where the matrix-valued process $\boldsymbol{R}$ solves the BSRE
    \begin{equation}
        d \boldsymbol{R}_t = \big( \boldsymbol{D}_t - \boldsymbol{R}_t \, \boldsymbol{H}_t \, \boldsymbol{R}_t \big)\, dt + \boldsymbol{Z}_t \, dW_t, \quad \boldsymbol{R}_T = -\boldsymbol{A}.
        \label{posit_def_bsre}
    \end{equation}
\end{theorem}

\begin{proof}
    Given the linearity of system \eqref{heter}, we think of the linear ansatz \eqref{linear_decoup}. By plugging \eqref{linear_decoup} back to \eqref{heter} and matching the coefficients, we can see $\boldsymbol{R}$ solves the BSRE \eqref{posit_def_bsre} and $\boldsymbol{P}$ satisfies the linear BSDE
    \begin{equation}
        d\boldsymbol{P}_t = - \big( \boldsymbol{R}_t \, \boldsymbol{H}_t \, \boldsymbol{P}_t + \boldsymbol{R}_t \, \boldsymbol{G}_t \big) + \tilde{\boldsymbol{Z}}_t \, dW_t, \quad \boldsymbol{P}_T = 0.
        \label{equation_P}
    \end{equation}
    Let us look at \eqref{posit_def_bsre} and first remark that both $\boldsymbol{D}_t$ and $\boldsymbol{A}$ are symmetric positive semi-definite matrices for all $t$. The same is true for $\boldsymbol{H}_t$ because it is a symmetric $M_0$-matrix. Equation \eqref{posit_def_bsre} is similar to the type of BSRE studied in \cite{peng1992stochastic}, from which we learn there exists a unique bounded solution $\boldsymbol{R}$. Consequently, the linear BSDE \eqref{equation_P} has bounded coefficients. The resulting Lipschitz property further guarantees the well-posedness of $\boldsymbol{P}$. A solution $(\boldsymbol{Q}, \boldsymbol{Y}, \boldsymbol{M})$ can be obtained by inserting \eqref{linear_decoup} back to \eqref{heter}. Realizing that \eqref{linear_decoup} is also a regular decoupling field, the solution is also unique according to Theorem \ref{local_time}.
\end{proof}

\kong

\begin{example}
    In this example, we will see the ordering property may break down under heterogeneity risk preferences. The following conditions are examined:
    \kong
    
    \begin{itemize}
        \item order flows satisfy $a_t = b_t = 3/2$ for all $t$;\\
        \vspace{-0.2cm}

        \item random variables $A^1$ and $A^2$ are deterministic, with $A^1>A^2>0$;\\
        \vspace{-0.2cm}

        \item running penalties satisfy $\phi_t^1 = \phi_t^2 = 1$;\\
        \vspace{-0.2cm}
        
        \item it holds $0 < q_0^1 < q_0^2 < \frac{A^1}{A^2}\, q_0^1$ and $T$ is large enough.
    \end{itemize}
    \kong
    Due to the lack of noises, BSRE \eqref{posit_def_bsre} degenerates to the Riccati equation
    \begin{equation}
    \begin{aligned}
         d \boldsymbol{R}_t &= \big( \boldsymbol{D}_t - \boldsymbol{R}_t \, \boldsymbol{H}_t \, \boldsymbol{R}_t \big)\, dt\\
         &=\bigg( \begin{bmatrix}
             1 & 0\\
             0 & 1
         \end{bmatrix} - \boldsymbol{R}_t \, \gamma \, \begin{bmatrix}
             1 & 0 \\
             -1 &0
         \end{bmatrix}\,\begin{bmatrix}
             1 & 0 \\
             -1 & 0
         \end{bmatrix}^* \, \boldsymbol{R}_t \bigg)\,dt
    \end{aligned}
         \label{deter_riccati}
    \end{equation}
    such that $\boldsymbol{R}_T = -\boldsymbol{A}$. We apply the asymptotic result studied in Theorem 2.1 by \cite{wonham1968matrix}, after checking that:
    \kong
    
    \begin{itemize}
        \item matrix $\begin{bmatrix}
             1 & 0 \\
             -1 & 0
             \end{bmatrix}$ is stable because it is an $M_0$-matrix;\\
        \vspace{-0.1cm}

        \item $\bigg( \begin{bmatrix}
            1 & 0\\ 
            0 & 1 
        \end{bmatrix}, \boldsymbol{0} \bigg)$ is observable since matrix $\begin{bmatrix}
            1 & 0 & 0 & 0\\
            0 & 1 & 0 & 0 
        \end{bmatrix}$ is of rank $2$.
    \end{itemize}
\kong
Consequently, we know there exists a matrix $\boldsymbol{R}_\infty\in\mathbb{R}^{2\times 2}$ that is the unique negative definite solution of algebraic Riccati equation
\begin{equation}
\begin{aligned}
    0 &= \begin{bmatrix}
             1 & 0\\
             0 & 1
         \end{bmatrix} - \boldsymbol{R} \, \gamma \, \begin{bmatrix}
             1 & 0 \\
             -1 &0
         \end{bmatrix}\,\begin{bmatrix}
             1 & 0 \\
             -1 & 0
         \end{bmatrix}^* \, \boldsymbol{R}\\
         & = \begin{bmatrix}
             1 & 0\\
             0 & 1
         \end{bmatrix} - \boldsymbol{R} \, \gamma \, \begin{bmatrix}
             1 & -1 \\
             -1 & 1
         \end{bmatrix}\, \boldsymbol{R}. 
    \label{alge_riccati}    
\end{aligned}
\end{equation}
Moreover, it holds that
\begin{equation}
    \boldsymbol{R}_\infty = \lim_{t \to -\infty} \boldsymbol{R}_t.
    \label{asymptot}
\end{equation}

We claim that $\boldsymbol{R}_\infty$ is both symmetric and persymmetric. Indeed, given $\boldsymbol{R}_\infty$ as the solution of \eqref{asymptot}, consider the permutation transform
\begin{equation*}
    \tilde{\boldsymbol{R}}_\infty := \begin{bmatrix}
        0 & 1 \\
        1 & 0
    \end{bmatrix} \,
    \boldsymbol{R}_\infty \,
    \begin{bmatrix}
        0 & 1 \\
        1 & 0
    \end{bmatrix}. 
\end{equation*}
Because direct calculation yields
\begin{equation*}
\begin{aligned}
    0 &=  \begin{bmatrix}
        0 & 1 \\
        1 & 0
    \end{bmatrix} \, \begin{bmatrix}
             1 & 0\\
             0 & 1
         \end{bmatrix} \,  \begin{bmatrix}
        0 & 1 \\
        1 & 0
    \end{bmatrix} -  \begin{bmatrix}
        0 & 1 \\
        1 & 0
    \end{bmatrix} \, \boldsymbol{R}_\infty \, \gamma \, \begin{bmatrix}
             1 & -1 \\
             -1 & 1
         \end{bmatrix}\, \boldsymbol{R}_\infty \,  \begin{bmatrix}
        0 & 1 \\
        1 & 0
    \end{bmatrix}\\
    & =  \begin{bmatrix}
             1 & 0\\
             0 & 1
         \end{bmatrix} - \tilde{\boldsymbol{R}}_\infty \, \gamma \, \begin{bmatrix}
             1 & -1 \\
             -1 & 1
         \end{bmatrix}\, \tilde{\boldsymbol{R}}_\infty,
\end{aligned}
\end{equation*}
matrix $\tilde{\boldsymbol{R}}_\infty$ solves the same equation \eqref{asymptot} as $\boldsymbol{R}_\infty$. The symmetric and persymmetric property thus can be deduced from the uniqueness of the solution. We then write
\begin{equation*}
    \boldsymbol{R}_\infty = \begin{bmatrix}
        \mathfrak{r}_1 & \mathfrak{r}_2\\
        \mathfrak{r}_2 & \mathfrak{r}_1
    \end{bmatrix}.
\end{equation*}
Note that $\mathfrak{r}_1 - \mathfrak{r}_2 <0$ due to the negative definiteness. Defining 
\begin{equation*}
    \begin{bmatrix}
        \, Y_\infty^1 \, \\
        \, Y_\infty^2 \,
    \end{bmatrix} :=
    \boldsymbol{R}_\infty \,
    \begin{bmatrix}
        \, q_0^1 \,\\
        \, q_0^2 \,
    \end{bmatrix},
\end{equation*}
we can derive 
\begin{equation*}
    Y_\infty^1 - Y_\infty^2 = (\mathfrak{r}_1 - \mathfrak{r}_2) \, (q_0^1 - q_0^2) > 0. 
\end{equation*}
Since the condition $a_t = b_t$ yields $\boldsymbol{G}_t = \boldsymbol{P}_t = 0$, the decoupling field \eqref{linear_decoup} tells us 
\begin{equation*}
    \begin{bmatrix}
        \, Y_0^1 \, \\
        \, Y_0^2 \,
    \end{bmatrix} =
    \boldsymbol{R}_0 \,
    \begin{bmatrix}
        \, q_0^1 \,\\
        \, q_0^2 \,
    \end{bmatrix}.
\end{equation*}
Combined with the asymptotic behavior \eqref{asymptot}, we can find $T$ large enough such that $Y_0^1 - Y_0^2 > 0$.

To show it can not be true that $Y_t^1 - Y_t^2 \geq 0$ for all $t$, let us assume $\Delta Y_t := Y_t^1 - Y_t^2 \geq 0$ and it follows by definition
\begin{equation*}
    Q_T^1 = q_0^1 + \int_0^T \Delta Y_t \,dt.
\end{equation*}
If we similarly define $\Delta Q_t := Q_t^1 - Q_t^2$, it turns out $(\Delta Q, \Delta Y)$ solves
\begin{equation*}
\left\{
\begin{aligned}
\;& d\Delta Q_t  = 2\, \Delta Y_t \, dt, \\
& d\Delta Y_t = \Delta Q_t \, dt, \\
& \Delta Q_0 = q_0^1 - q_0^2, \quad \Delta Y_T = -2A^2\, \Delta Q_T - 2(A^1 - A^2) \, Q_T^1.
\end{aligned}
\right.
\end{equation*}
We then find 
\begin{equation*}
\begin{aligned}
    \Delta Y_T &= -2A^2\, \Big( q_0^1 - q_0^2 + 2\int_0^T \Delta Y_t \, dt \Big) - 2(A^1 - A^2) \, \Big( q_0^1 + \int_0^T \Delta Y_t \,dt\Big)\\
    & = -2A^2\, \Big(\frac{A^1}{A^2}\,q_0^1 - q_0^2 \Big) - 2(A^1 + A^2)\, \int_0^T \Delta Y_t \,dt\\
    &<0.
\end{aligned}
\end{equation*}
Therefore, although $\Delta Y_0 >0$, it can not hold for all $t$ and there exist some time $s<T$ such that $\Delta Y_s <0$. Finally, utilizing \eqref{alge_trick}, the sign of $\Delta Y$ determines which agent quotes a better price. For example, the positive value of $\Delta Y_0$ indicates that agent $1$ provides a worse sell price (but a better buy price) than agent $2$. For an economical interpretation, we first remark that heterogeneity lies in the terminal penalty coefficient. As $T$ becomes sufficiently large, the solution $\boldsymbol{R}_0$ converges to $\boldsymbol{R}_\infty$, solving the algebraic equation \eqref{alge_riccati} independent of the terminal condition. Essentially, the influence exerted by the terminal penalty diminishes as we move backward in time. Consequently, agents are approximately homogeneous at the initial time. As discussed earlier, the better buy price offered by agent $1$ results from the lower inventory level. However, as $T$ approaches, this approximation loses validity, leading to the `crossing' of quotes.\\
\end{example}

\kong

\sloppy
\printbibliography

\vspace{0.2cm}

\section{Appendix: Stochastic Maximum Principle}
\label{maximum prin}
\noindent This section is devoted to the game version of the stochastic maximum principle. We follow the argument in \cite{carmona2016lectures} that is devoted to the control problem, and start with the \textit{necessary} condition. Fix an admissible strategy profile $(\boldsymbol{\delta}^j)_{1\leq j\leq N}\in(\mathbb{A}\times\mathbb{A})^{N}$ and denote by $(Q^1, \dots, Q^N)\in(\mathbb{S}^2)^N$ the corresponding controlled inventory. Next, in view of player $i$, let us consider $\boldsymbol{\beta}\in \mathbb{A}\times\mathbb{A}$, which is uniformly bounded and satisfies $\boldsymbol{\delta}^i+\boldsymbol{\beta}\in\mathbb{A}\times\mathbb{A}$. The direction $\boldsymbol{\beta}$ will be used to compute the G\^ateaux derivative of $J^i$. For each $\epsilon > 0$ small enough, define an admissible control $\boldsymbol{\delta}^{i, \epsilon}\in \mathbb{A}\times\mathbb{A}$ as $\boldsymbol{\delta}^{i, \epsilon}_t = \boldsymbol{\delta}^i_t + \epsilon\,\boldsymbol{\beta}_t$, and the corresponding controlled state is denoted by $Q^{i,\epsilon}\in\mathbb{S}^2$. Let $V$ be the solution of the equation
\begin{equation*}
    dV_t=\partial_{\boldsymbol{\delta}}\theta(t, \boldsymbol{\delta}^i_t; \boldsymbol{\delta}^{-i}_t)\cdot\boldsymbol{\beta}_t\,dt,
\end{equation*}
with the initial condition $V_0 = 0$, where
\begin{equation*}
    \theta(t, \boldsymbol{\delta}^i; \boldsymbol{\delta}^{-i}):=-a_t\,\Lambda(\delta^{i,a}-\bar{\delta}^{i,a})+b_t\,\Lambda(\delta^{i,b}-\bar{\delta}^{i,b}).
\end{equation*}
It is clear that $V$ is uniformly bounded in any finite time horizon.

\kong

\begin{lemma}
The functional $\boldsymbol{\delta}\hookrightarrow J^i(\boldsymbol{\delta}; \boldsymbol{\delta}^{-i})$ is G\^ateaux differentiable and the derivative reads
\begin{equation}
\begin{aligned}
    \frac{d}{d\epsilon}J^i(\boldsymbol{\delta}^i+\epsilon\boldsymbol{\beta}; \boldsymbol{\delta}^{-i})\big|_{\epsilon=0}:&=\lim_{\epsilon\searrow0}\frac{1}{\epsilon}\,\big[J^i(\boldsymbol{\delta}^i+\epsilon\boldsymbol{\beta}; \boldsymbol{\delta}^{-i})-J^i(\boldsymbol{\delta}^i; \boldsymbol{\delta}^{-i})\big]\\
    &=\mathbb{E}\Big[\int_0^T\big(\partial_q f(t,Q_t^i,\boldsymbol{\delta}_t^i; \boldsymbol{\delta}^{-i}_t)\,V_t + \partial_{\boldsymbol{\delta}}f(t, Q_t^i, \boldsymbol{\delta}_t^i; \boldsymbol{\delta}^{-i}_t)\cdot\boldsymbol{\beta}_t\big)\,dt-2\,A^i\,V_T\,Q_T^i\Big],
    \label{gen_gat_deriv}
\end{aligned}
\end{equation}
where $f$ is the running payoff given as
\begin{equation*}
    f(t, Q_t^i, \boldsymbol{\delta}_t^i; \boldsymbol{\delta}^{-i}_t)=b_t\,\delta_t^{i,b}\,\Lambda(\delta_t^{i,b}-\bar{\delta}_t^{i,b})+a_t\,\delta_t^{i, a}\,\Lambda(\delta_t^{i,a}-\bar{\delta}_t^{i,a})-\phi_t^i\,\big(Q_t^i\big)^2.
\end{equation*}
\end{lemma}
\begin{proof}
Regarding the ask side of the running payoff, we compute that
\begin{equation*}
\begin{aligned}
\frac{1}{\epsilon}\,&\mathbb{E}\Big[\int_0^T a_t\,(\delta_t^{i,a}+\epsilon\,\beta_t^a)\,\Lambda(\delta_t^{i,a}+\epsilon\,\beta_t^a-\bar{\delta}_t^{i,a})\,dt-\int_0^T a_t\,\delta_t^{i,a}\,\Lambda(\delta_t^{i,a}-\bar{\delta}_t^{i,a})\,dt\Big]\\
&=\mathbb{E}\Big[\int_0^Ta_t\,\delta_t^{i,a}\,\frac{1}{\epsilon}\,\big(\Lambda(\delta_t^{i,a}+\epsilon\,\beta_t^a-\bar{\delta}_t^{i,a})-\Lambda^a(\delta_t^{i,a}-\bar{\delta}_t^{i,a})\big)\,dt\Big]+\mathbb{E}\Big[\int_0^Ta_t\,\beta_t^a\,\Lambda(\delta_t^{i,a}+\epsilon\,\beta_t^a-\bar{\delta}_t^{i,a})\,dt\Big].
\end{aligned}
\end{equation*}
To perform the limiting procedure, it suffices to notice
\begin{equation*}
\begin{aligned}
    \int_0^Ta_t\,\delta_t^{i,a}\,\frac{1}{\epsilon}\,\big(\Lambda(\delta_t^{i,a}+\epsilon\,\beta_t^a-\bar{\delta}_t^{i,a})-\Lambda^a(\delta_t^{i,a}-\bar{\delta}_t^{i,a})\big)\,dt&\leq C\,\int_0^Ta_t\,|\delta_t^{i,a}|\,dt,\\
    \int_0^Ta_t\,\beta_t^a\,\Lambda(\delta_t^{i,a}+\epsilon\,\beta_t^a-\bar{\delta}_t^{i,a})\,dt &\leq C\,\int_0^T a_t\,dt,\\
    a_t\,\delta_t^{i,a}\,\frac{1}{\epsilon}\,\big(\Lambda(\delta_t^{i,a}+\epsilon\,\beta_t^a-\bar{\delta}_t^{i,a})-\Lambda^a(\delta_t^{i,a}-\bar{\delta}_t^{i,a})\big)&\leq C\,a_t\,|\delta_t^{i, a}|,\\
    a_t\,\beta_t^a\,\Lambda(\delta_t^{i,a}+\epsilon\,\beta_t^a-\bar{\delta}_t^{i,a}) &\leq C\, a_t,
\end{aligned}
\end{equation*}
and the integrability on the right hand sides. Hence, the dominated convergence theorem yields
\begin{equation*}
\begin{aligned}
    \lim_{\epsilon\to 0}\frac{1}{\epsilon}\,\mathbb{E}\Big[\int_0^T a_t\,(\delta_t^{i,a}&+\epsilon\,\beta_t^a)\,\Lambda(\delta_t^{i,a}+\epsilon\,\beta_t^a-\bar{\delta}_t^{i,a})\,dt-\int_0^T a_t\,\delta_t^{i,a}\,\Lambda(\delta_t^{i,a}-\bar{\delta}_t^{i,a})\,dt\Big]\\
    &=\mathbb{E}\Big[\int_0^Ta_t\,\delta_t^{i,a}\,\big(\Lambda(\delta_t^{i,a}-\bar{\delta}_t^{i,a})\big)'\,\beta_t^a\,dt\Big]+\mathbb{E}\Big[\int_0^Ta_t\,\beta_t^a\,\Lambda(\delta_t^{i,a}-\bar{\delta}_t^{i,a})\,dt\Big]\\
    &=\mathbb{E}\Big[\int_0^T\partial_{\delta^a}f(t, Q_t^i, \boldsymbol{\delta}_t^i; \boldsymbol{\delta}^{-i}_t)\,\beta_t^a\,dt\Big].
\end{aligned}
\end{equation*}
Since the bid side can be computed in the same way, the term with $\partial_{\boldsymbol{\delta}}f$ in \eqref{gen_gat_deriv} is thus verified. With respect to the running inventory penalty, we further calculate that
\begin{equation*}
\begin{aligned}
    \frac{1}{\epsilon}\,\mathbb{E}\Big[\int_0^T \phi_t^i\,&\big(Q_t^{i,\epsilon}\big)^2\,dt-\int_0^T \phi^i_t\,\big(Q_t^{i}\big)^2\,dt\Big]\\
    &=\mathbb{E}\Big[\int_0^T \phi_t^i\,\big(Q_t^{i,\epsilon}+Q_t^{i}\big)\,\Big(-\int_0^ta_u\,\frac{1}{\epsilon}\,\big[\Lambda(\delta_u^{i,a}+\epsilon\,\beta_u^a-\bar{\delta}_u^{i,a})-\Lambda(\delta_u^{i,a}-\bar{\delta}_u^{i,a})\big]\,du\\
    &\hspace{4.1cm}+\int_0^tb_u\,\frac{1}{\epsilon}\,\big[\Lambda(\delta_u^{i,b}+\epsilon\,\beta_u^b-\bar{\delta}_u^{i,b})-\Lambda(\delta_u^{i,b}-\bar{\delta}_u^{i,b})\big]\,du\Big)\,dt\Big].
\end{aligned}
\end{equation*}
Similarly, being aware of the boundedness of $\phi^i$, $Q^i$, $Q^{i, \epsilon}$, $a$,  $b$, as well as
\begin{equation*}
 \frac{1}{\epsilon}\,\big[\Lambda(\delta_u^{i,a}+\epsilon\,\beta_u^a-\bar{\delta}_u^{i,a})-\Lambda(\delta_u^{i,a}-\bar{\delta}_u^{i,a})\big]
 \text{\quad and \quad} \frac{1}{\epsilon}\,\big[\Lambda(\delta_u^{i,b}+\epsilon\,\beta_u^b-\bar{\delta}_u^{i,b})-\Lambda(\delta_u^{i,b}-\bar{\delta}_u^{i,b})\big],
\end{equation*}
we can conclude by the dominated convergence theorem that
\begin{equation*}
    \lim_{\epsilon\to 0}\,\frac{1}{\epsilon}\,\mathbb{E}\Big[\int_0^T \phi_t^i\,\big(Q_t^{i,\epsilon}\big)^2\,dt-\int_0^T \phi^i_t\,\big(Q_t^{i}\big)^2\,dt\Big]=2\,\mathbb{E}\Big[\int_0^T \phi_t^i\,Q_t^i\,V_t\,dt\Big].
\end{equation*}
While the terminal penalty part can be justified in the same way, the proof is complete.
\end{proof}

\kong

\noindent Let $(Y^i,M^i)$ be the adjoint processes associated with $(\boldsymbol{\delta}^j)_{1\leq j\leq N}$, i.e., processes $(Y,Z)$ solve the BSDE
\begin{equation*}
    dY_t^i=2\phi_t^i\,Q_t^i\,dt+dM_t^i, \quad Y_T^i=-2A^i\,Q_T^i.
\end{equation*}
The following duality relation provides an expression of the G\^ateaux derivative of the cost functional in terms of the Hamiltonian of the system.

\kong

\begin{lemma}
The duality relation is given by
\begin{equation*}
    \mathbb{E}[Y_T^i\,V_T]=\mathbb{E}\Big[\int_0^T\big(\partial_{\boldsymbol{\delta}}\theta(t, \boldsymbol{\delta}^i_t; \boldsymbol{\delta}^{-i}_t)\cdot\boldsymbol{\beta}_t\,Y_t^i-\partial_q f(t,Q_t^i,\boldsymbol{\delta}_t^i; \boldsymbol{\delta}^{-i}_t)\,V_t\big)\,dt\Big].
\end{equation*}
Such duality further implies
\begin{equation*}
    \frac{d}{d\epsilon}\,J^i(\boldsymbol{\delta}^i+\epsilon\boldsymbol{\beta}; \boldsymbol{\delta}^{-i})\big|_{\epsilon=0}=\mathbb{E}\Big[\int_0^T\partial_{\boldsymbol{\delta}}H^i(t, Q_t^i, Y_t^i, \boldsymbol{\delta}_t^i; \boldsymbol{\delta}_t^{-i})\cdot\boldsymbol{\beta}_t\,dt\Big].
\end{equation*}
\begin{proof}
The integration by parts yields
\begin{equation*}
\begin{aligned}
    Y_T^i\,V_T&=Y_0^i\,V_0+\int_0^T Y_t^i\,dV_t+\int_0^T V_t\,dY_t^i\\
    &=\int_0^TY_t\, \partial_{\boldsymbol{\delta}}\theta(t, \boldsymbol{\delta}^i_t; \boldsymbol{\delta}^{-i}_t)\cdot\boldsymbol{\beta}_t\,dt+\int_0^T 2\,\phi_t^i\,V_t\,Q_t^i\,dt+\int_0^T V_t\,dM_t^i.
\end{aligned}
\end{equation*}
The duality is obtain through taking the expectation. Using this relation, we can further compute
\begin{equation*}
\begin{aligned}
    \frac{d}{d\epsilon}\,J^i(\boldsymbol{\delta}^i+\epsilon\boldsymbol{\beta} &; \boldsymbol{\delta}^{-i})\big|_{\epsilon=0}\\
    &=\mathbb{E}\Big[\int_0^T\big(\partial_q f(t,Q_t^i,\boldsymbol{\delta}_t^i; \boldsymbol{\delta}^{-i}_t)\,V_t + \partial_{\boldsymbol{\delta}}f(t, Q_t^i, \boldsymbol{\delta}_t^i; \boldsymbol{\delta}^{-i}_t)\cdot\boldsymbol{\beta}_t\big)\,dt+V_T\,Y_T^i\Big]\\
    &=\mathbb{E}\Big[\int_0^TY_t^i\, \partial_{\boldsymbol{\delta}}\theta(t, \boldsymbol{\delta}^i_t; \boldsymbol{\delta}^{-i}_t)\cdot\boldsymbol{\beta}_t\,dt+\int_0^T \partial_{\boldsymbol{\delta}}f(t, Q_t^i, \boldsymbol{\delta}_t^i; \boldsymbol{\delta}^{-i}_t)\cdot\boldsymbol{\beta}_t\big)\,dt\Big]\\
    &=\mathbb{E}\Big[\int_0^T\partial_{\boldsymbol{\delta}}H^i(t, Q_t^i, Y_t^i, \boldsymbol{\delta}_t^i; \boldsymbol{\delta}_t^{-i})\cdot\boldsymbol{\beta}_t\,dt\Big].
\end{aligned}
\end{equation*}
\end{proof}
\end{lemma}

\kong

\noindent With these preliminary steps, we have now arrived at the necessary condition for optimality.

\kong

\begin{theorem}[Necessary condition]
If an admissible strategy profile $(\boldsymbol{\delta}^j)_{1\leq j\leq N}$ is a Nash equilibrium, $(Q^j)_{1\leq j\leq N}$ are the corresponding controlled inventories, and $(Y^j, M^j)_{1\leq j\leq N}$ are the associated adjoint
processes, then it holds for any player $i$ that
\begin{equation*}
    H^i(t, Q_t^i, Y_t^i, \boldsymbol{\delta}_t^i; \boldsymbol{\delta}_t^{-i})\geq H^i(t, Q_t^i, Y_t^i, \boldsymbol{\beta}; \boldsymbol{\delta}_t^{-i}),
\end{equation*}
$dt\times d\mathbb{P}-a.s.$ for any $\boldsymbol{\beta}\in[-\xi, \xi]\times[-\xi, \xi]$.
\end{theorem}
\begin{proof}
Fix any $i\in\{1, \dots, N\}$. Due to the convexity of the admissible space, given any admissible and bounded $\boldsymbol{\beta}\in\mathbb{A}\times\mathbb{A}$, we can choose the perturbation $\boldsymbol{\delta}^{i,\epsilon}=\boldsymbol{\delta}^i+\epsilon\,(\boldsymbol{\beta}-\boldsymbol{\delta}^i)$ which is still admissible. Since $(\boldsymbol{\delta}^j)_{1\leq j\leq N}$ is a Nash equilibrium, the following inequality should hold:
\begin{equation*}
    \frac{d}{d\epsilon}\,J^i(\boldsymbol{\delta}^i+\epsilon\,(\boldsymbol{\beta}-\boldsymbol{\delta}^i); \boldsymbol{\delta}^{-i})\big|_{\epsilon=0}=\mathbb{E}\Big[\int_0^T\partial_{\boldsymbol{\delta}}H^i(t, Q_t^i, Y_t^i, \boldsymbol{\delta}_t^i; \boldsymbol{\delta}_t^{-i})\cdot(\boldsymbol{\beta}_t-\boldsymbol{\delta}_t^i)\,dt\Big]\leq0.
\end{equation*}
From this we can see 
\begin{equation}
    \partial_{\boldsymbol{\delta}}H^i(t, Q_t^i, Y_t^i, \boldsymbol{\delta}_t^i; \boldsymbol{\delta}_t^{-i})\cdot(\boldsymbol{\beta}-\boldsymbol{\delta}_t^i)\leq 0,
    \label{smp_gener}
\end{equation}
$dt\times d\mathbb{P}-$a.s. for all $\boldsymbol{\beta}\in[-\xi, \xi]\times[-\xi, \xi]$. Look at the ask side for example. Regarding the condition \eqref{smp_gener}, it can only happen at time $t$ if one of the following three cases holds:
\begin{gather*}
    \delta_t^{i,a}\in(-\xi, +\xi) \text{\quad with \quad} \partial_{\boldsymbol{\delta}^a}H^i(t, Q_t^i, Y_t^i, \boldsymbol{\delta}_t^i; \boldsymbol{\delta}_t^{-i})=0, \\
    \delta_t^{i,a}=-\xi\leq \beta^a \text{\quad with \quad} \partial_{\boldsymbol{\delta}^a}H^i(t, Q_t^i, Y_t^i, \boldsymbol{\delta}_t^i; \boldsymbol{\delta}_t^{-i})\leq0,\\
    \delta_t^{i,a}=+\xi\geq \beta^a \text{\quad with \quad} \partial_{\boldsymbol{\delta}^a}H^i(t, Q_t^i, Y_t^i, \boldsymbol{\delta}_t^i; \boldsymbol{\delta}_t^{-i})\geq0.
\end{gather*}
Suppose that
\begin{equation*}
    \partial_{\boldsymbol{\delta}^a}H^i(t, Q_t^i, Y_t^i, \boldsymbol{\delta}_t^i; \boldsymbol{\delta}_t^{-i})=-a_t\, \big(\Lambda^a(\delta_t^{i,a}-\bar{\delta}_t^{i,a}) \big)'\,\Big[Y_t-\delta_t^{i,a}-\frac{\Lambda(\delta_t^{i,a}-\bar{\delta}_t^{i,a})}{\big(\Lambda(\delta_t^{i,a}-\bar{\delta}_t^{i,a})\big)'}\Big]=0,
\end{equation*}
the monotonicity of $\delta+\frac{\Lambda(\delta)}{(\Lambda(\delta))'}$ (see Assumption \ref{inten_assu}) implies that $\delta_t^{i,a}$ maximizes the corresponding part of the Hamiltonian. In another case, think of $\delta_t^{i,a}=-\xi$ with  $\partial_{\boldsymbol{\delta}^a}H^i(t, Q_t^i, Y_t^i, \boldsymbol{\delta}_t^i; \boldsymbol{\delta}_t^{-i})\leq0$. Since  $H^i$ is first increasing and then decreasing, this property infers that $\delta_t^{i,a}=-\xi$ is the maximizer of the Hamiltonian on the interval $[-\xi, \xi]$. While the remaining case and the bid side can be discussed in a similar fashion, all three cases imply that the strategy $\boldsymbol{\delta}^i$ maximizes the $H^i$ along the optimal paths. The proof is complete since the previous discussion holds for any player.
\end{proof}

\kong

The \textit{sufficient} condition is more straightforward to derive, stated as below.

\kong

\begin{theorem}[Sufficient condition]
Let $(\boldsymbol{\delta}^j)_{1\leq j\leq N}$ be an admissible strategy profile, $(Q^1\, \dots, Q^N)$ be the corresponding controlled inventories, and $(Y^j, M^j)_{1\leq j\leq N}$ be the associated adjoint processes. If it holds $dt\times d\mathbb{P}$-a.s. that
\begin{equation}
    H^i(t, Q_t^i, Y_t^i, \boldsymbol{\delta}_t^i; \boldsymbol{\delta}_t^{-i})=\sup_{\boldsymbol{\beta}\in [-\xi,\xi]^2}H^i(t, Q_t^i, Y_t^i, \boldsymbol{\beta}; \boldsymbol{\delta}_t^{-i}),
    \nonumber
\end{equation}
for all $i\in\{1, \dots, N\}$, then $(\boldsymbol{\delta}^j)_{1\leq j\leq N}$ is a Nash equilibrium.
\begin{proof}
Fix any $i\in\{1, \dots, N\}$. Let $\boldsymbol{\beta}\in \mathbb{A}\times\mathbb{A}$ be a generic admissible strategy and $Q^{i'}$ be the state process associated with the profile $(\boldsymbol{\beta},\boldsymbol{\delta}^{-i})$. Due to the concaveness of the terminal penalty, we have
\begin{equation}
    \begin{aligned}
        \mathbb{E}\big[-A^i\,\big(Q_T^i\big)^2 + A^i\,\big(Q^{i'}_T\big)^2\big]&\geq \mathbb{E}\big[-2 A^i\, Q^{i}_T\,(Q^{i}_T-Q^{i'}_T)\big]\\
        &=\mathbb{E}\big[Y_T^i\,(Q_T^i-Q^{i'}_T)\big]\\
        &=\mathbb{E}\Big[\int_0^T2\phi_t^i\,Q_t^i\,(Q_t^i-Q^{i'}_t)\,dt+\int_0^TY_t^i\,d(Q_t^i-Q^{i'}_t)\Big].\\
    \end{aligned}
    \nonumber
\end{equation}
By the definition of the Hamiltonian, one can also obtain
\begin{equation}
    \begin{aligned}    \mathbb{E}\int_0^T\big[ f(t, Q_t^i, \boldsymbol{\delta}_t^i; \boldsymbol{\delta}^{-i}_t) &- f(t, Q_t^{i'}, \boldsymbol{\beta}_t; \boldsymbol{\delta}^{-i}_t) \big]\,dt\\
    &=\mathbb{E}\int_0^T\big[H(t, Q_t^i, \boldsymbol{\delta}_t^i; \boldsymbol{\delta}^{-i}_t)-H(t, Q_t^{i'}, \boldsymbol{\beta}_t; \boldsymbol{\delta}^{-i}_t)\big]\,dt-\mathbb{E}\int_0^TY_t^i\,d(Q_t^i-Q^{i'}_t),
    \end{aligned}
    \nonumber
\end{equation}
where we recall that $f$ is the running payoff of agent $i$. Combining two results above, we can get
\begin{equation}
\begin{aligned}
    J^i(\boldsymbol{\delta}^i;\boldsymbol{\delta}^{-i})-J^i(\boldsymbol{\beta};\boldsymbol{\delta}^{-i})&=\mathbb{E}\int_0^T\big[ f(t, Q_t^i, \boldsymbol{\delta}_t^i; \boldsymbol{\delta}^{-i}_t) - f(t, Q_t^{i'}, \boldsymbol{\beta}_t; \boldsymbol{\delta}^{-i}_t) \big]\,dt+\mathbb{E}\big[-2A^i\,(Q_T^i)^2 +2A^i\, (Q^{i'}_T)^2\big]\\
    &\geq\mathbb{E}\int_0^T\big[H(t, Q_t^i, \boldsymbol{\delta}_t^i; \boldsymbol{\delta}^{-i}_t)-H(t, Q_t^{i'}, \boldsymbol{\beta}_t; \boldsymbol{\delta}^{-i}_t)+2\phi_t^i\,Q_t^i\,(Q_t^i-Q^{i'}_t)\big]\,dt\\
    &\geq\mathbb{E}\int_0^T\big[H(t, Q_t^i, \boldsymbol{\beta}_t; \boldsymbol{\delta}^{-i}_t)-H(t, Q_t^{i'}, \boldsymbol{\beta}_t; \boldsymbol{\delta}^{-i}_t)+2\phi_t^i\,Q_t^i\,(Q_t^i-Q^{i'}_t)\big]\,dt\\
    &=\mathbb{E}\int_0^T\phi_t^i\,\big(Q_t^i-Q^{i'}_t\big)^2\,dt\\
    &\geq 0.
\end{aligned}
    \nonumber
\end{equation}
The proof is complete since the previous discussion holds for any player.
\end{proof}
\end{theorem}

\vspace{0.2cm}

\section{Appendix: Non-smooth Analysis and Implicit Function}
\noindent Let $X$ be a separable Banach space. Given function $F:X\to\mathbb{R}$, the generalized directional derivative in $x \in X$ with respect to the direction $h\in X$ is given by
\begin{equation*}
    F^\circ(x;h):=\limsup_{\substack{y\to x \\ t\to 0}}\frac{1}{t} \, \big[F(y+t\,h)-F(y)\big].
\end{equation*}
In addition, we review the Rademacher's theorem.

\kong

\begin{theorem}[Rademacher]
    For some $m, l\in\mathbb{N}$, let $U\subseteq\mathbb{R}^m$ be open and $F:U \to \mathbb{R}^l$ be locally Lipschitz continuous. Then $F$ is Fr\'{e}chet differentiable in almost every $x\in U$.\\
\end{theorem}

\kong

\noindent In view of Rademacher's theorem, we introduce the generalized derivative and the generalized Jacobian matrix as follows.\\

\kong

\begin{definition}
(1) For a locally Lipschitz function $F : X \to \mathbb{R}$, its \textit{(Clarke) generalized derivative} in $x\in X$ is defined by
    \begin{equation*}
        \partial F(x):=\big\{ x^*\in X^* \,:\, \langle x^*, h\rangle_X \leq F^\circ(x;h), \; \forall h\in X\big\}.
    \end{equation*}
Here, $\langle\cdot, \cdot\rangle_X$ denotes the duality product between $X$ and $X^*$. Specifically, if $X=\mathbb{R}^m$, then $F$ is differentiable on $\mathbb{R}^m \, \backslash \, E_F$ for a set $E_F\subseteq \mathbb{R}^m$ of Lebesgue measure $0$. Then, the generalized derivative can be characterized by
\begin{equation*}
        \partial F(x)= \text{co }\big\{ \lim_{n\to\infty} \nabla F(x_n) \,:\, x_n\to x, \, x_n\notin E_F \, \big\},
\end{equation*}
where co $A$ denotes the convex hull of $A\subseteq \mathbb{R}^m$.

(2) For a locally Lipschitz function $F : \mathbb{R}^m \to \mathbb{R}^l$, its \textit{(Clarke) generalized Jacobian} in $x\in \mathbb{R}^m$ is defined as 
\begin{equation*}
    \partial F(x):= \text{co }\big\{ \lim_{n\to\infty} \nabla F(x_n) \,:\, x_n\to x, \, x_n\notin E_F \, \big\}.
\end{equation*}
Here, we let $F$ be differentiable on $\mathbb{R}^m \, \backslash \, E_F$ for a set $E_F\subseteq \mathbb{R}^m$ of Lebesgue measure $0$.\\
\end{definition}

\kong

\begin{remark}
Unlike some literature, we will not use $\partial_C$ to denote the generalized derivative. Instead, we will rely on the subscript notation to represent the partial derivative. For illustration,
\begin{equation*}
\begin{aligned}
    \partial_x F(x, y) &\text{ denotes the generalized derivative of $F$ with respect to its first variable;}\\ 
    \partial F(x, y) &\text{ denotes the generalized derivative of $F$ with respect to the variable } (x,y).
\end{aligned}
\end{equation*}
It is worth noting that in the smooth case, the generalized derivatives and classical derivatives coincide, so there should be no ambiguity in our notation.\\
\end{remark}

\kong

\noindent With these prerequisites in place, we now proceed to present the proof of the implicit function theorem.

\kong

\begin{proof}(Propostion \ref{general_implicit})
Given the first two conditions, one already knows there exists a unique locally Lipschitz function $f:\mathbb{R}^m\to\mathbb{R}^n$ such that $F(x,y)=0$ and $x=f(y)$ are equivalent in the set $\mathbb{R}^n\times\mathbb{R}^m$, according to Theorem 4 in \cite{galewski2018global}. It suffices to verify the global Lipschitz property.

First, the proof of Theorem 4 is accomplished by concatenating the local result given by Theorem 2 in \cite{galewski2018global}. For every $y\in\mathbb{R}^m$, the resulting function $f$ is Lipschitz in a neighborhood containing $y$, and thus is locally Lipschitz overall. Subsequently, the proof of Theorem 2 is built upon the inverse function theorem, referring to Theorem 7.1.1 in \cite{clarke1990optimization}. Illustrated by the Corollary on page 256 in \cite{clarke1990optimization}, we can define the function $\tilde{F}:\mathbb{R}^{m+n}\to\mathbb{R}^{m+n}$ by
    \begin{equation*}
        \tilde{F}(x,y)=\big[y, \;F(x,y)\big],
    \end{equation*}
and apply the inverse function theorem to $\Tilde{F}$ to justify the local version of the implicit function theorem (i.e., Theorem 2 in \cite{galewski2018global}). The function $f$ is then Lipschitz (in a neighborhood) because of the Lipschitz inverse according to Theorem 7.1.1. Finally, the Lipschitz coefficient of the Lipschitz inverse is $1/\delta$, provided that the distance between the distance between $\partial \tilde{F}(\hat{x},\hat{y})\,\tilde{S}$---letting $\tilde{S}$ signify the unit sphere in $\mathbb{R}^{m+n}$---and $0$ is $2\delta$, for a $(\hat{x},\hat{y})\in\mathbb{R}^n\times\mathbb{R}^m$ fixed at the beginning.

Therefore, given that there exists a constant $\upsilon>0$ such that the distance between $\partial \tilde{F}(x,y)\,\tilde{S}$ and $0$ is at least $\upsilon$ for all $(x,y)\in\mathbb{R}^n\times\mathbb{R}^m$, the above procedure tells us that the function $f$ is Lipschitz in every small neighborhood, with the Lipschitz coefficient being bounded uniformly by $2/\upsilon$. It follows $f$ is (globally) Lipschitz.
\end{proof}

\kong

\noindent We finish with the calculus rules of the generalized derivative.

\kong

\begin{theorem}[\cite{clason2017nonsmooth}, \cite{clarke1990optimization}]
\label{sum rule}
Let $F : X \to \mathbb{R}$ be locally Lipschitz continuous near $x \in X$ and $\kappa \in \mathbb{R}$, then
\begin{equation*}
    \partial (\kappa \, F)(x) = \kappa \partial F(x).
\end{equation*}
Let $G : X \to \mathbb{R}$ also be locally Lipschitz continuous near $x \in X$, then
\begin{equation*}
    \partial (F + G)(x) \subseteq \partial F(x) + \partial G(x): = \big\{ f + g \,:\, f \in \partial F(x), \, g \in \partial G(x) \big\}.
\end{equation*}
The above are also true for the case $F, G : \mathbb{R}^m \to \mathbb{R}^l$.\\
\end{theorem}

\end{document}